\documentclass[10pt,twocolumn,amsmath,amssymb,aps,pra,secnumarabic,
    nofootinbib,groupedaddress,floatfix]{revtex4-1}

\usepackage{mathtools,mathptmx,amsthm,tikz,enumitem,eucal,pifont,comment}
\usepackage[colorlinks=true,linkcolor=red!70!black,citecolor=blue!70!black,
    urlcolor=magenta!70!black,backref=page]{hyperref}
\setlist[enumerate,1]{font=\bfseries,label=\arabic*.}
\setlist[enumerate,2]{font=\bfseries,label=(\alph*)}
\numberwithin{equation}{section}

\makeatletter\def\@bibdataout@init{}\def\pre@bibdata{}\makeatother

\counterwithin{figure}{section}

\usetikzlibrary{decorations.pathreplacing}
\usetikzlibrary{arrows,shadings}
\colorlet{darkblue}{blue!70!black}
\colorlet{darkgreen}{green!50!black}
\colorlet{darkred}{red!70!black}
\colorlet{lightblue}{blue!20!white}
\colorlet{lightgreen}{green!25!white}
\colorlet{halfblue}{blue!50!white}

\newtheorem{theorem}{Theorem}[section]
\newtheorem{conjecture}[theorem]{Conjecture}
\newtheorem{corollary}[theorem]{Corollary}
\newtheorem{lemma}[theorem]{Lemma}
\newtheorem{proposition}[theorem]{Proposition}
\theoremstyle{definition}

\newtheorem{algo}{Algorithm}
\newenvironment{algorithm}[1]
    {\renewcommand\thealgo{\textsf{#1}}\algo}{\endalgo}
\theoremstyle{remark}
\newtheorem*{remark}{Remark}

\newtheorem*{example}{Example}
\newtheorem*{examples}{Examples}

\newcommand{\Conj}[1]{Conjecture~\ref{#1}}
\newcommand{\Cor}[1]{Corollary~\ref{#1}}

\newcommand{\Fig}[1]{Figure~\ref{#1}}
\newcommand{\Lem}[1]{Lemma~\ref{#1}}
\newcommand{\Prop}[1]{Proposition~\ref{#1}}
\newcommand{\Sec}[1]{Section~\ref{#1}}
\newcommand{\Thm}[1]{Theorem~\ref{#1}}
\newcommand{\Alg}[1]{Algorithm~\ref{#1}}
\newcommand{\equ}[1]{equation~\eqref{#1}}
\newcommand{\Equ}[1]{Equation~\eqref{#1}}

\newcommand{\ie}{\emph{i.e.}}
\newcommand{\Ie}{\emph{I.e.}}

\newcommand{\ccP}{\mathsf{P}}
\newcommand{\BPP}{\mathsf{BPP}}
\newcommand{\BQP}{\mathsf{BQP}}
\newcommand{\NP}{\mathsf{NP}}
\newcommand{\uSVP}{\mathsf{uSVP}}

\newcommand{\sZ}{\mathsf{Z}}

\newcommand{\PSL}{\operatorname{PSL}}
\newcommand{\GL}{\operatorname{GL}}
\newcommand{\poly}{\operatorname{poly}}
\newcommand{\Poly}{\operatorname{Poly}}
\newcommand{\Vol}{\operatorname{Vol}}
\newcommand{\lcm}{\operatorname{lcm}}

\newcommand{\tr}{\operatorname{tr}}

\renewcommand{\Pr}{\operatorname{Pr}}
\newcommand{\Ex}{\operatorname{Ex}}
\renewcommand{\wr}{\operatorname{\,wr\,}}

\providecommand{\abs}[1]{\lvert#1\rvert}
\providecommand{\norm}[1]{\lVert#1\rVert}

\newcommand{\eps}{\epsilon}

\newcommand{\into}{\hookrightarrow}

\newcommand{\normaleq}{\unlhd}
\newcommand{\onto}{\twoheadrightarrow}
\newcommand{\spar}{\shortparallel}
\renewcommand{\setminus}{\smallsetminus}
\renewcommand{\tensor}{\otimes}

\newcommand{\yes}{\text{yes}}
\newcommand{\no}{\text{no}}
\newcommand{\bit}{\text{bit}}
\newcommand{\ceil}[1]{\lceil #1 \rceil}
\newcommand{\ket}[1]{\lvert #1\rangle}
\newcommand{\bra}[1]{\langle #1\rvert}
\newcommand{\braket}[1]{\langle #1\rangle}
\newcommand{\ketbra}[1]{\ket{#1}\bra{#1}}

\newcommand{\Z}{\mathbb{Z}}
\newcommand{\R}{\mathbb{R}}
\newcommand{\C}{\mathbb{C}}
\newcommand{\Q}{\mathbb{Q}}

\newcommand{\hA}{\widehat{A}}
\newcommand{\hB}{\widehat{B}}
\newcommand{\hP}{\widehat{P}}
\newcommand{\cP}{\mathcal{P}}
\newcommand{\cS}{\mathcal{S}}
\newcommand{\cX}{\mathcal{X}}
\newcommand{\cY}{\mathcal{Y}}
\newcommand{\hdelta}{\widehat{\delta}}
\newcommand{\htau}{\widehat{\tau}}
\newcommand{\hphi}{\widehat{\phi}}
\newcommand{\hpsi}{\widehat{\psi}}
\newcommand{\hxi}{\widehat{\xi}}

\newcommand{\tTheta}{\widetilde{\Theta}}
\newcommand{\tY}{\widetilde{Y}}
\newcommand{\tvv}{\widetilde{\vv}}
\newcommand{\tvx}{\widetilde{\vx}}
\newcommand{\tvy}{\widetilde{\vy}}
\newcommand{\tz}{\widetilde{z}}

\newcommand{\va}{\vec{a}}
\newcommand{\vb}{\vec{b}}
\newcommand{\vc}{\vec{c}}
\newcommand{\ve}{\vec{e}}
\newcommand{\vm}{\vec{m}}
\newcommand{\vnabla}{\vec{\nabla}}
\newcommand{\vs}{\vec{s}}
\newcommand{\vt}{\vec{t}}
\newcommand{\vu}{\vec{u}}
\newcommand{\vv}{\vec{v}}
\newcommand{\vw}{\vec{w}}
\newcommand{\vx}{\vec{x}}
\newcommand{\vy}{\vec{y}}
\newcommand{\vz}{\vec{z}}

\def\app#1#2{\mathrel{\setbox0=\hbox{$#1\sim$}%
    \setbox2=\hbox{\rlap{\hbox{$#1\propto$}}\lower1.1\ht0\box0}%
    \raise0.25\ht2\box2}}
\def\approxprop{\mathpalette\app\relax}

\def\api#1#2{\mathrel{%
    \setbox0=\hbox{$#1\sim$}%
    \setbox2=\hbox{\rlap{\hbox{$#1\in$}}\lower1.5\ht0\box0}%
    \raise0.35\ht2\box2}}
\def\approxin{\mathpalette\api\relax}

\newcommand{\defeq}{\stackrel{\text{def}}=}

\newenvironment{eq}[1]{\begin{equation} \label{#1}}
    {\end{equation}\ignorespacesafterend}
\newcommand{\eatline}{\vspace{-\baselineskip}}
\newcommand{\back}{\hspace{-.5em}}
\newcommand{\hback}{\hspace{-.25em}}

\begin{document}
\title{The hidden subgroup problem for infinite groups}
\author{Greg Kuperberg}
\email{greg@math.ucdavis.edu}
\thanks{Partly supported by NSF grants CCF-1716990, CCF-2009029,
    and CCF-2317280.}
\affiliation{University of California, Davis}

\date{\today}

\begin{abstract}
\centerline{\textit{\normalsize Dedicated to Eleanor Rieffel, my
    friend and colleague for 42 years and counting.}}
\vspace{\baselineskip}

Following the example of Shor's algorithm for period-finding in the
integers, we explore the hidden subgroup problem (HSP) for discrete
infinite groups.  On the hardness side, we show that HSP is $\NP$-hard
for the additive group of rational numbers, and for normal subgroups of
non-abelian free groups.  We also indirectly reduce a version of the short
vector problem to HSP in $\Z^k$ with pseudo-polynomial query cost.  On the
algorithm side, we generalize the Shor--Kitaev algorithm for HSP in $\Z^k$
(with standard polynomial query cost) to the case where the hidden subgroup
has deficient rank or equivalently infinite index.  Finally, we outline a
stretched exponential time algorithm for the abelian hidden shift problem
(AHShP), extending prior work of the author as well as Regev and Peikert.
It follows that HSP in any finitely generated, virtually abelian group
also has a stretched exponential time algorithm.
\end{abstract}

\maketitle

\tableofcontents

\section{Introduction}
\label{s:intro}

Some of the most important quantum algorithms are those that solve
rigorously stated computational problems, and that are superpolynomially
faster than classical alternatives.  This includes Shor's famous quantum
algorithm for period finding \cite{Shor:factor}, with the equally famous
applications of integer factorization and discrete logarithm problems.
The hidden subgroup problem has been a popular framework for many such
algorithms \cite{ME:hidden}.  We recall its formulation: If $G$ is a
discrete group and $X$ is an unstructured set, a function $f:G \to X$
\emph{hides} a subgroup $H \subseteq G$ means that $f(x) = f(y)$ if and
only if $y = xh$ for some $h \in H$.  In other words, $f$ hides $H$ when
it is right $H$-periodic and otherwise injective, as summarized in this
commutative diagram:
\begin{eq}{e:hsp}
\begin{tikzpicture}
\draw (-1.25,0) node (G) {$G$};
\draw (1.25,0) node (X) {$X$};
\draw (0,-1.25) node (H) {$G/H$};
\draw[->] (G) -- node[above] {$f$} (X);
\draw[->>] (G) -- (H);
\draw[right hook->] (H) -- (X);
\end{tikzpicture} \end{eq}
Such an $f$ is also called a \emph{hiding function}, while $H$ is called a
\emph{hidden subgroup}.   In an HSP algorithm, $f$ is typically functional
input: It is given by either a subroutine or an oracle that returns
individual values $f(x)$.  Given a group $G$ and given a hiding function
$f$, the \emph{hidden subgroup problem} (HSP) is then the question of
calculating the hidden subgroup $H$.

The computational complexity of HSP depends greatly on the ambient group
$G$.  We may also consider variations in which we promise that $H$ is a
normal subgroup, or that $H$ is drawn from a specific conjugacy class of
subgroups, etc.  The most interesting cases of HSP are those that have
an efficient quantum algorithm, but no efficient classical algorithm.
When $f$ is given by an oracle, there is often an unconditional proof that
HSP is classically intractable.

The most widely discussed case of Shor's algorithm solves HSP when $G =
\Z$, and thus $H = p\Z$ for some period $p$, in quantum polynomial time.
Even though the group $\Z$ is infinite and Shor's algorithm is one motivation
for HSP, many of the results since then have been about HSP for finite
groups \cite[Ch.~VII]{CD:algebraic}.  In this article, we will examine HSP
in several cases when the ambient group $G$ is infinite (and discrete).
We establish quantum algorithms in some key cases, but our more surprising
results are hardness results instead, and we lead with these.

We first consider the additive group $G = \Q$, which is abelian but is
not finitely generated.  To state our result, we simplify the hidden
subgroup problem to a decision problem.  We define the hidden subgroup
\emph{existence} problem (HSEP) to be the yes-no question:  Is the hidden
subgroup $H$ non-trivial?  (In \Sec{ss:hidden} and \Prop{p:hsepnp}, we
will discuss why a natural version of HSEP lies in $\NP$.)

\begin{theorem} Let $\Q$ be the set of rational numbers viewed as a
discrete group under addition.  Then HSEP in the quotient $\Q/\Z$ is
$\NP$-complete, using either an integer factoring oracle, or assuming a
conjecture in number theory (\Conj{c:smoothp1}).  Equivalently, HSP in
$\Q$ is $\NP$-complete if the question is to determine whether the hidden
subgroup strictly contains $\Z$.
\label{th:rational} \end{theorem}

Note that \Thm{th:rational} yields unconditional $\NP$-hardness with
$\BQP$ reduction, since Shor's algorithm can factor integers in quantum
polynomial time.

Our second result is about the normal hidden subgroup problem (NHSP) and
its existence version (NHSEP) in a non-abelian group, meaning HSP with
the extra restriction that the hidden subgroup $H$ is normal.

\begin{theorem} NHSEP in a finitely generated, non-abelian free group $F_k$
is $\NP$-complete.
\label{th:free} \end{theorem}

We contrast \Thm{th:free} with several results for finite groups.  NHSP in
a finite group $G$ has a polynomial-time quantum algorithm whenever
the quantum Fourier transform on $G$ can be computed in polynomial time
\cite{HRT:normal}.  For this reason and others, NHSP is considered easier
than the general HSP, although general HSP does have polynomial quantum
query complexity for all finite $G$ \cite{EHK:hsp}.  By contrast,
the proofs of Theorems~\ref{th:rational} and \ref{th:free} both relativize
(see \Sec{ss:complex}).  As a corollary, both of these infinite HSP problems
have high query complexity.

\begin{corollary} HSEP in $\Q/\Z$ and NHSEP in a non-abelian free group
$F_k$ both have quantum query complexity $2^{\Theta(n)}$, where $n$ is
the minimum bit complexity to describe a hidden subgroup.
\label{c:expquery} \end{corollary}

The abelianization of the free group $F_k$ is the free abelian group
$\Z^k$, but the abelian analogue of \Thm{th:free} is not the same HSP
that is solved by any version of Shor's algorithm.  The reason is that the
length of a query $w \in F_k$ in \Thm{th:free} is its word length, which
is analogous to the 1-norm $\norm{\vv}_1$ of $\vv \in \Z^k$ rather than
the bit complexity $\norm{\vv}_\bit$ of $\vv$.  In other words, the query
complexity is pseudo-polynomial (polynomial in $\norm{\vv}_1$) rather
than polynomial (in $\norm{\vv}_\bit$). We can also obtain a hardness
result for this alternate HSP for $\Z^k$.  The same pseudo-polynomial
query cost model arises in a quantum algorithm due to Childs, Jao, and
Soukharev \cite{CJS:isogenies} that identifies isogenies between ordinary
elliptic curves.  (However, the hidden subgroups that they consider are
different from the ones in the hardness reduction in \Thm{th:svp} below.)

\begin{theorem} Let $f:\Z^k \to X$ be a hiding function with
pseudo-polynomial query cost, and suppose that $f$ hides a subgroup $H =
\braket{\vv} \subseteq \Z^k$ which is generated by a single vector $\vv$.
Suppose that there is a (quantum) algorithm to find $H$ in polynomial
time, uniformly for all $k$ and $f$.  Then the unique short vector problem
$\uSVP_{a,\Poly(a,k)}$ for an integer lattice $L \subseteq \Z^k$ has a
(quantum) algorithm that runs in time $\poly(a,k,\norm{L}_\bit)$.
\label{th:svp} \end{theorem}

In \Sec{ss:learning}, we outline another reduction from HSP on $\Z^k$ with
pseudo-polynomial query cost to identify a lattice from samples with
error.  The two reductions together are similar to Regev's quantum reduction
from short vector problems to learning with errors (LWE) \cite{Regev:lwe}.
However, Regev's reduction is much deeper than ours.

Turning to positive results, it has sometimes been said that HSP is solved
for finitely generated abelian groups $G$ \cite{CD:algebraic,ME:hidden}.
However, the well-known algorithm due to Shor and Kitaev for this problem
\cite{Kitaev:stab} assumes that the hidden subgroup $H$ has finite index
in $G$.  We give a solution in the trickier case when $H$ may have infinite
index or equivalently lower rank.

\begin{theorem} HSP in $\Z^k$ (with polynomial query cost for the hiding
function) can be solved for arbitrary hidden subgroups uniformly in quantum
polynomial time $\poly(k,n)$, where $k$ is the dimension and $n$ is the
bit complexity of the answer.
\label{th:abelian} \end{theorem}

Our proof of \Thm{th:abelian} begins with the same steps as other Shor-type
algorithms.  We apply the hiding function $f$ to a superposition of a finite
truncation of $\Z^k$.  Ignoring the output register, we then apply a QFT to
the input register and measure a Fourier mode.   In Shor's algorithm and in
Kitaev's generalization, the Fourier mode is deciphered using continued
fractions to recognize rational values from their digit expansions.
Deciphering the Fourier mode in \Thm{th:abelian} is more difficult than
in the Shor--Kitaev algorithm.  Our solution uses the LLL lattice reduction
algorithm \cite{LLL:factoring}, facts about generalized Fourier transforms
(or Pontryagin duality), and various error estimates.

The time complexity in \Thm{th:abelian} is uniformly polynomial in the
dimension $k$ along with other parameters.  Implicitly, $k$ is given
in unary, which is consistent with the standard notation that requires
one or more bits for each coordinate of a vector in $\Z^k$, including for
vanishing coordinates.  We can also consider compressed notation for sparse
vectors in an abelian group with exponential rank $k$.  As discussed in
\Sec{ss:simon}, HSP even in the simplest case $(\Z/2)^k$ becomes $\NP$-hard
if $k$ is exponential and vectors are given in sparse notation.

\Sec{ss:simon} also defines sparse and dense encodings of elements of
$\Z^\infty$ (the group of infinite sequences of integers with finite
support).  We can restate Theorems~\ref{th:svp} and \ref{th:abelian},
together with the proof of \Thm{th:simon}, as results about $\Z^\infty$
using different possible encodings of its elements.

\begin{corollary} Consider HSP in the group $\Z^\infty$ with a finitely
generated hidden subgroup $H \subseteq \Z^\infty$ whose generator matrix has
bit complexity $n$.  If $\Z^\infty$ is given with dense encoding of vectors
and binary encoding of coefficients, then HSP can be solved in quantum
time $\poly(n)$.  If $\Z^\infty$ is given with dense encoding of vectors and
unary encoding of coefficients, then HSP is $\uSVP_{a,\Poly(a,k)}$-hard. If
$\Z^\infty$ is given with sparse encoding of vectors and either unary or
binary encoding of coefficients, then HSP is $\NP$-hard.
\label{c:infabel} \end{corollary}

The positive case of \Cor{c:infabel} is equivalent to \Thm{th:abelian},
given the observation that $H \subseteq \Z^{\poly(n)}$.  The negative
case is similarly equivalent to \Thm{th:svp}.

We also state a corollary of \Thm{th:abelian} which applies to any fixed
finitely generated abelian group $G$.  In other words, it is a solution
to the abelian hidden subgroup problem (AHSP) in the finitely-generated case.

\begin{corollary} If $G$ is a finitely generated abelian group, then
HSP in $G$ for an arbitrary hidden subgroup $H$ can be solved in quantum
polynomial time in the bit complexity of the answer.
\label{c:fgabel} \end{corollary}

\Prop{p:virtual} in \Sec{ss:explicit} describes the relevant encoding of
elements of $G$ for \Cor{c:fgabel}, in the greater generality of finitely
generated, virtually abelian groups.  Briefly, if a finitely generated
group $G$ has an efficient algorithm to compute canonical compressed words,
then those words work as a standard encoding of elements of $G$.   If $G$
is abelian, then this amounts to describing every $g \in G$ as an ordered
product of generators of $G$ with the exponents in binary.  See also
\Sec{ss:hidden} for a brief proof of \Cor{c:fgabel} given \Thm{th:abelian}.

We conclude with a sketch of a result about the abelian hidden shift problem
(AHShP).  See \Sec{ss:hidden} for a definition of the hidden shift problem
(HShP) and a discussion of its relation to HSP.

\begin{theorem}[Outlined] Let $H \subseteq \Z^k$ be a visible subgroup
of any rank with a generator matrix $M$, and let $h = \norm{M}_\bit$.
Consider AHShP in the quotient group $\Z^k/H$ with a hidden shift $\vs$,
$n \ge \norm{\vs}_\bit$, and assume that $n = \Omega((k+\log(h))^2)$.
Then there is a quantum algorithm for this problem with quantum time
complexity $2^{O(\sqrt{n})}$, quantum space complexity $\poly(n)$, and
classical space complexity $2^{O(\sqrt{n})}$.  The hiding function query
cost may be polynomial or may be as high as $2^{O(\sqrt{n})}$.
\label{th:shift} \end{theorem}

Here, $H$ is a \emph{visible} subgroup means that its generator matrix $M$
is directly given as part of the input.

In \Sec{s:ahshp}, we will give an algorithm for AHShP and an informal
analysis that together amount to an outline of a proof of \Thm{th:shift}.
We omit some of the statistical estimates that would be needed for a rigorous
proof.  Our algorithm is a variation of prior algorithms for dihedral
HSP due to the author \cite{K:dhsp,K:dhsp2}, Regev \cite{Regev:dhsp},
and Peikert \cite{Peikert:csidh}.

As discussed in \Sec{ss:hidden}, \Thm{th:shift} yields the following
corollary for HSP for virtually abelian groups (VAHSP).  Recall that a
group $G$ is \emph{virtually abelian} when it has a finite-index abelian
subgroup $K \subseteq G$.  The corollary generalizes the remark that HSP
for a dihedral group is equivalent to HShP for a cyclic group.

\begin{corollary} Let $G$ be a finitely generated, virtually abelian group.
Then HSP in $G$ (with compressed word polynomial query cost) can be solved
in time $2^{O(\sqrt{n})}$, where $n$ is the bit complexity of the answer.
\label{c:virtual} \end{corollary}

\acknowledgments

The author would like to thank Goulnara Arzhantseva, Andrew Childs, Sean
Hallgren, Derek Holt, Ilya Kapovich, Misha Kapovich, Emmanuel Kowalski,
Michael Larsen, Chris Peikert, Eleanor Rieffel, and Mark Sapir for useful
discussions.  The author would also like to thank Vivian Kuperberg and
Yuval Wigderson for many helpful comments about the manuscript.

\section{Background}
\label{s:background}

\subsection{Basic notation}
\label{s:notation}

We will use the asymptotic notation $\poly(n)$ to mean $O(n^\alpha)$ and
$\Poly(n)$ to mean $\Omega(n^\alpha)$.  We will also extend this notation to
several variables; for instance, $\poly(x,y)$ means $O((x+y)^\alpha)$, etc.
We also follow the convention that if $f(x) = \Omega(g(x))$, then $f$ and $g$
are both eventually positive and $f(x)/g(x)$ is eventually bounded below.
Finally, we will often use asymptotic notation with or without uniformity
assumptions.   If $f(x) = O(g(x))$, meaning that $|f(x,y)| \le Cg(x)$
eventually, the constant $C$ may or may not be independent of $y$; and
the estimate might only be uniform when $y$ lies in a restricted range.

We will often use vector notation $\vx$ to refer to points in $\Z^k$,
$\R^k$, etc., mainly to suggest the geometric interpretation of vectors.
This notation is optional and we do not use it for quantum states (we use
Dirac notation $\ket{\psi}$ instead), nor for bit strings except when they
are used for linear algebra in $(\Z/2)^k$.

\subsection{Computational complexity}
\label{ss:complex}

We refer to Arora--Barak \cite{AB:modern} for basics of computational
complexity theory, but we also clarify several points.

As explained in \Sec{s:intro}, Theorems~\ref{th:rational} and \ref{th:free}
are stated for the hidden subgroup existence problem (HSEP) rather than
the hidden subgroup problem (HSP).  Both problems are $\NP$-hard given
the hypotheses of the theorems.  However, it does not make sense to
ask whether HSP is $\NP$-complete because it is a function problem and
therefore cannot be in $\NP$.  By contrast, HSEP is a decision problem;
as stated in \Prop{p:hsepnp} below, a natural interpretation of HSEP is
in $\NP$ for any explicit group $G$.

Since HSEP is a decision problem, we can also ask whether it is $\NP$-hard
with Karp--Post reduction rather than Cook--Turing reduction.  This is
what Theorems~\ref{th:rational} and \ref{th:free} both achieve.  The only
caveat is that the reduction in \Thm{th:rational} requires either integer
factorization or \Conj{c:smoothp1}.  As explained in \Sec{s:intro}, since
Shor's algorithm can factor integers in $\BQP$, we obtain a reduction in
$\BQP$ unconditionally.

Second, it may seem confusing that HSP can be $\NP$-hard, given
that in full generality it is an oracle problem.  To address this,
given a decision problem $d \in \NP$, we construct a hiding function
$f$ that has a polynomial-time reduction to $d$, thus replacing the
oracle by an \emph{instance}.  In fact, the construction works just
as well if $d \in \NP^g$, where $g$ is another oracle.  In other words,
Theorems~\ref{th:rational} and \ref{th:free} are $\NP$-completeness results
that relativize.   Relativizing $\NP$-completeness is an unusual concept
that is possible in this case because HSP is an oracle problem.

Third, the computational complexity of any oracle problem can be modified by
declaring a time cost function for oracle calls, rather than the more usual
convention of unit cost.  If an oracle query has length $n$, then the oracle
call has a tacit time cost of $\Omega(n)$, given just the time to write the
query to the oracle tape or register.  However, we can also declare a time
fee $t(n)$ for such an oracle call that depends on the input length $n$,
or even a fee $t(x)$ that depends on the specific input $x$ to the oracle.
In the statement of \Thm{th:svp}, we assume that the hiding function $f$
is an oracle with \emph{pseudo-polynomial} cost.  This means that the cost
to evaluate $f(\vv)$ with $\vv \in \Z^k$ is a polynomial in the 1-norm
$\norm{\vv}_1$ and the dimension $k$, rather than polynomial in the bit
complexity of a description of $\vv$.

An oracle fee can also often be emulated by demanding verbose input.  This is
another way to implement a suitable special case of the oracle assumption in
\Thm{th:svp}.  If the components of a vector $\vv \in \Z^k$ are written
in unary, then its total symbol length is $\Theta(\norm{\vv}_1+k)$, counting
also the negation symbol and the commas to separate vector components.

\subsection{Explicit groups and canonical words}
\label{ss:explicit}

To interpret the hidden subgroup problem as a computational problem, we need
a computational model for both the group $G$ and the hiding function $f$.
In this subsection, we consider a computational model of a discrete,
countably infinite group $G$.  We call such a $G$ an \emph{explicit group}
if it is given as a subset of the set of strings $A^*$ over some alphabet
$A$ with a membership function, a group law, and an inverse function:
\[ \begin{array}{c@{\qquad}c}
d:A^* \to \{\yes,\no\} & d(x) = x \stackrel{?}{\in} G \\
m:G\times G \to G & m(x,y) = xy \\
i:G \to G & i(x) = x^{-1}
\end{array} \]
We assume that all three of these functions are computable in $\ccP$
or functional $\ccP$.  Aaronson and the author previously gave a related
definition of an explicit sequence of finite groups \cite[Def.~5.7]{K:qcap},
except with a trivial membership test.

To interpret our theorems, we need to describe the additive group of
rationals $\Q$, the free groups $F_k$, and the free abelian groups $\Z^k$
as explicit groups.  We assume familiar notation for each of these three
examples, in which case the algorithms for $d$, $m$, and $i$ are all
straightforward:
\begin{enumerate}
\item We put each rational $t \in \Q$ in lowest terms $t = a/b$ using the
Euclidean algorithm, which runs in polynomial time.  The numerator and
denominator are written in binary.
\item We express each element $w \in F_k$ as the unique reduced word
in the $k$ generators of the free group $F_k$, and their inverses.
\item In the version of HSP that we solve in \Thm{th:abelian}, we express
each element $v \in \Z^k$ as a list of integers in binary.  Or we can
match the oracle cost hypothesis of \Thm{th:svp} if we write each
coefficient $v_j$ of $v$ in unary.  (We also consider sparse
encoding of vectors in \Sec{ss:simon}.)
\end{enumerate}

If $G$ and $H$ are two explicit groups, we define an \emph{efficient
homomorphism} to be a homomorphism $f:G \to H$ that can be computed in
deterministic polynomial time.   Likewise if $f$ is an isomorphism,
then it is \emph{efficient} when $f$ and $f^{-1}$ are both efficient.
Note that two explicit groups can be isomorphic as groups
without any efficient isomorphism between them. In other words, that
the same group can have more than one explicit encoding.

\begin{examples} The group $\PSL(2,\Z)$ is isomorphic to a free product
$(\Z/2) * (\Z/3)$ \cite{Alperin:psl}.  This yields two explicit encodings,
the former using integer matrices and the latter using words in the
presentation
\[ (\Z/2) * (\Z/3) = \braket{a,b \mid a^2 = b^3 = 1}. \]
More precisely, we can canonically name elements of $(\Z/2) * (\Z/3)$
as words in the two letters $a$ and $b$ that do not contain either $aa$
or $bbb$.  These two explicit encodings cannot be equivalent under any
isomorphism due to the equation
\[ \begin{bmatrix} 1 & N \\ 0 & 1 \end{bmatrix}
    = \begin{bmatrix} 1 & 1 \\ 0 & 1 \end{bmatrix}^N. \]
This matrix is an $N$th power and its word must have length $\Theta(N)$,
but its matrix form has bit complexity $\Theta(\log(N))$.

Another example is the infinitely generated abelian group $(\Z/2)^\infty$
in \Sec{ss:simon}.  We can encode an element $x \in (\Z/2)^\infty$
with a variable-length, finite bit string.  Alternatively, we can
list the terms of a basis expansion of $x$ with a sparse encoding as in
equation \eqref{e:sparse}.  The two encodings are not equivalent under any
automorphism of $(\Z/2)^\infty$.  In the sparse encoding, exponentially
many linearly independent vectors have polynomial-length encodings, which
is not possible in the standard dense encoding.
\end{examples}

Our definition of an explicit group $G$ only stipulates a polynomial-time
algorithm to multiply two group elements and group inverses, equivalently a
polynomial-time algorithm to evaluate words in $G$ of any bounded length.
We are interested in two refinements of this concept. We call an explicit
group $G$ \emph{word efficient} if there is an algorithm to evaluate words
in $G$ which is uniformly polynomial time for all words.  This concept
can be taken one step further if we consider algebraic circuits $C$ in
the group law of $G$.  In some of the literature, such circuits are also
called \emph{straight-line programs} \cite{BS:matrix,Lohrey:groups}.  In a
straight-line program, we build a word $w$ in $G$ by naming intermediate
words that can then be reused.  Such a word $w$ can also be called a
\emph{compressed word}, since similar data structures are used in data
compression methods that name and reuse subwords.  We call an explicit
group $G$ \emph{compressed-word efficient} if it has a polynomial time
algorithm to evaluate algebraic circuits in elements of $G$.

Suppose that $G$ is a group generated by the finite set $A$.  Then every
element $g \in G$ can be expressed, non-uniquely, by words in the set $W_A =
(A \cup A^{-1})^*$.  In some cases, there is an efficient algorithm to
convert any word $w \in W_A$ to a \emph{canonical word} with the same
value $g \in G$, and that represents each $g \in G$ uniquely.  We can
extend the same concept to compressed words.  It may happen that $G$
has a polynomial-time algorithm to convert every compressed word in the
alphabet $A$ that represents an element $g \in G$ to a \emph{canonical
compressed word} for the same $g$.  It is easy to see that a canonical
word algorithm or a canonical compressed word algorithm for $G$ yields
an explicit structure for $G$.  The following proposition is also not
difficult to prove, but it is more subtle.

\begin{proposition} Let $G$ be a group generated by a finite set $A$.
\begin{enumerate}
\item If $G$ has an explicit structure given by a canonical word algorithm
for words in $W_A$, then $G$ is word efficient. If $f:G \to H$ is a
homomorphism to another word-efficient group $H$, then $f$ is efficient.
In particular, any two canonical word algorithms (even in different finite
generating sets) for $G$ yield isomorphic explicit structures.

\item If $G$ has an explicit structure given by a canonical compressed-word
algorithm for compressed words in $W_A$, then $G$ is compressed-word
efficient. If $f:G \to H$ is a homomorphism to another word-efficient
group $H$, then $f$ is efficient.  In particular, any two canonical
compressed-word algorithms for $G$ yield isomorphic explicit structures.
\end{enumerate}
\label{p:canonical} \end{proposition}

We now examine evaluation of words and compressed words for the three types
of groups defined above.
\begin{enumerate}
\item The additive group of rationals $\Q$ is not finitely generated,
so we cannot compare their notation to canonical words or canonical
compressed words.  However, the standard notation for rational numbers is
efficient for algebraic circuits that only use the group law of $\Q$.
\item The standard notation for elements of a free group $F_k$
using reduced words is a canonical (uncompressed) word structure.
\item If each vector $v \in \Z^k$ is expressed in binary, then the notation
is equivalent to canonical compressed words in $\Z^k$.  If each $v$ is
expressed in unary, then this is equivalent to canonical (uncompressed)
words in $\Z^k$.
\end{enumerate}

\begin{remark} Remarkably, there is an efficient algorithm to calculate
canonical compressed words in $F_k \cong F_A$, where $A$ is an alphabet with
$k$ letters.  This is established by Lohrey \cite{Lohrey:groups}, whose work
is the culmination of older results by various authors.  In Section 3.4
of his book, Lohrey describes a recompression technique due to J\'ez that
converts any algebraic circuit $C$ in the free monoid $W_A$ into a canonical
form $D$.  This ``recompression" algorithm depends only on the value of $C$
and not otherwise on its circuit structure (Lohrey, personal communication).
In Section 4.4, Lohrey reviews his own result that converts an algebraic
circuit $C$ in the free group $F_A$ to a circuit $D$ in the free monoid
$W_A$, such that the value of $D$ is automatically in reduced form (and
matches the value of $C$).  The two results together yield a canonical
compressed word algorithm in $F_A$.  We will not use this explicit model
of a free group in this article, but see the end of \Sec{ss:freeproof}.
\end{remark}

We conclude with a simple result on finitely generated, virtually abelian
groups.  This result defines and motivates our explicit structure for
any group in \Cor{c:fgabel} or \Cor{c:virtual}.

\begin{proposition} If $G$ is a finitely generated, virtually abelian group,
then it has an efficient algorithm to compute canonical compressed words.
\label{p:virtual} \end{proposition}

\begin{proof} By hypothesis, $G$ has a finite-index abelian subgroup $K$.
We can assume that $K \cong \Z^k$ is free abelian and that $K \normaleq
G$ is normal.  (Proof: An arbitrary such $K_0$ is finitely generated and
thus has a finite-index free abelian subgroup $K_1$.  We can then let $K =
K_2$ be the kernel of the action of $G$ on $G/K_1$.)  We can choose as our
generating set for $G$ a free abelian basis $e_1,e_2,\ldots,e_k$ for $K$,
together with one representative $b \in G$ for each coset of $K$ in $G$.
Then certainly every $x \in G$ can be uniquely represented in the form
\[ x = ba = be_1^{v_1}e_2^{v_2}\cdots e_k^{v_k}, \]
where the exponents are written in binary.   It is easy to express these
forms with algebraic circuits in the generators of $G$.

It is less obvious that any algebraic circuit $C$ in the generators of $G$
can be efficiently converted into this standard form.  The non-trivial
property that we must check is to show that the vector of exponents $\vv$
does not grow any worse than exponentially as $C$ is evaluated.  When a
gate in $C$ computes the product of $x_1 = b_1a_1$ and $x_2 = b_2a_2$,
the value is canonicalized into the form
\begin{eq}{e:virtab} x_1x_2 = b_1b_2(b_2^{-1}a_1b_2)a_2
    = b_3a_3(b_2^{-1}a_1b_2)a_2, \end{eq}
where $b_3$ and $a_3$ depend only on $b_1$ and $b_2$.  To show that the
growth in the exponents is manageable under the gates of $C$, we use the
fact that $H$ has a positive-definite quadratic form $Q:K \to \Z$ which
is invariant under conjugation by $G$.  (Proof: Starting with
any positive-definite form $Q_0$ on $K$, we can let $Q = Q_1$
be the sum of the orbit of $Q_0$ under conjugation by elements of the
finite group $G/K$.)  Given such a $Q$, we obtain
\begin{eq}{e:qtriang} \sqrt{Q(a_3(b_2^{-1}a_1b_2)a_2)}
    \le \sqrt{Q(a_3)} + \sqrt{Q(a_1)} + \sqrt{Q(a_2)} \end{eq}
by conjugation invariance and by the triangle inequality.  Combining
equations \eqref{e:virtab} and \eqref{e:qtriang}, a gate of $C$ that
creates $x_1x_2$ from some $x_1$ and $x_2$ at most doubles (and adds a
constant to) the largest available value of $\sqrt{Q(a)}$.  By induction,
the exponents created are at most exponential in the number of gates of $C$.
\end{proof}

\subsection{Hidden subgroups and hidden shifts}
\label{ss:hidden}

Besides an explicit group, the other ingredient in the hidden subgroup
problem is the hiding function $f$.  As explained in \Sec{s:intro}, $f$
is typically provided as functional input to an HSP algorithm, either as
a subroutine or an oracle.  Except when otherwise specified, the cost to
evaluate $f$ is polynomial in the length of the input.  We then parametrize a
bound $O(g(n))$ the computation time of the HSP algorithm by $n$, the cost of
the calls to $f$ needed to witness generators of the hidden subgroup $H$.
(This is more logical than letting $n$ be the length of an algorithm
to calculate $f$.  With that more naive complexity parametrization, and
without any bound on the generators of $H$, HSEP becomes equivalent to
the halting problem when $G = \Z$ and in many other cases.)

For example, consider the original hidden subgroup problem solved by
Shor with $G = \Z$.  In this case, $H = h\Z$ for some period $h$, and
Shor's algorithm runs in quantum polynomial time $\poly(\norm{h}_\bit) =
\poly(\log(h))$.

The following proposition relates explicit groups and hiding functions to
the intuitive fact that HSEP is in $\NP$.

\begin{proposition} Suppose that $G$ is an explicit group with a hidden
subgroup $H$, and that the running time for HSEP is parameterized by the
bit complexity of a non-trivial element $x \in H \subseteq G$.  Suppose also
that a hiding function $f$ for $H$ can be evaluated in polynomial time.
Then HSEP is in $\NP$.
\label{p:hsepnp} \end{proposition}

Although the proof is immediate, we outline it anyway.

\begin{proof} The $\NP$ certificate for HSEP is a non-trivial element $x
\in H \subseteq G$.  The verifier checks that $x \in G$ using the membership
test, and after that checks that $f(x) = f(1)$.
\end{proof}

Note also that if $K$ is a visible subgroup of $G$, and if there is
an efficient algorithm to calculate canonical representatives of cosets
of $K$ in $G$, then HSP in $K$ reduces to HSP in $G$.  Likewise if $K$ is
a normal subgroup of $G$, and if there is an algorithm to calculate the
quotient map $G \onto Q \cong G/K$ for a group $Q$ isomorphic to $G/K$,
then HSP in the $Q$ also reduces to HSP in $G$.

\begin{proof}[Proof of \Cor{c:fgabel}] If $e_1,e_2,\ldots,e_k$ are generators
for a finitely generated abelian group $G$ with canonical compressed words,
then there is an efficient surjective homomorphism $h:\Z^k \onto G$ given by
\[ h(\vv) = e_1^{v_1} e_2^{v_2} \cdots e_k^{v_k}. \]
If $f:G \to X$ is a hiding function, then so is the composition $f \circ
h:\Z^k \to X$, thus reducing HSP in $G$ to HSP in $\Z^k$.
\end{proof}

We define the \emph{hidden shift problem} (HShP), which is related to the
hidden subgroup problem via partial reductions in both directions.  Given a
group $G$ and a (known) subgroup $H \subseteq G$, the hidden shift problem
takes a function $f_0$
\[ \begin{tikzpicture}
\draw (-1.25,0) node (G) {$G$};
\draw (1.25,0) node (X) {$X$};
\draw (0,-1.25) node (H) {$G/H$};
\draw[->] (G) -- node[above] {$f_0$} (X);
\draw[->>] (G) -- (H);
\draw[right hook->] (H) -- (X);
\end{tikzpicture} \]
which is $H$-periodic on the right and otherwise injective, and another
function $f_1:G \to X$ defined by
\[ f_1(xs) = f_0(x) \]
for an unknown right shift $s$.  It follows that $f_1$ is
$s^{-1}Hs$-periodic.  The problem is to calculate the shift $s$ given
access to $f_0$ and $f_1$.

In one direction, we can partly reduce HSP to HShP as follows.  Suppose that
$f:G \to X$ hides an unknown subgroup $H$, and suppose that $K \subseteq
G$ is a known subgroup of $G$ with an efficient algorithm to find the hidden
subgroup $H \cap K \subseteq K$ from the hiding function $f|_K$.  Then for
each coset $xK \subseteq G$, the restriction $f|_{xK}$ is either a hidden
shift of $f|_K$ (if $H$ intersects $xK$) or has disjoint values from those
of $f|_K$ (if $H$ does not intersect $xK$).  An algorithm for HShP in $K$
can thus determine the set $H \cap xK$ for any one value of $x$.  If $[G:K]$
is small, and in some other cases, this yields an efficient reduction from
HSP in $G$ to HShP in $K$.

In the other direction, we can reduce HShP in $G$ to HSP in the wreath
product $G \wr \Z/2$ as follows.  Recall that this wreath product
is defined as the semidirect product $G^2 \rtimes \Z/2$; it consists
of elements $(x_0,x_1) \in G^2$ together with an involution $z$ such that
\[ a(x_0,x_1)a = (x_1,x_0). \]
We define a joint hiding function $f:G \wr \Z/2 \to X^2$ by
\begin{align*}
f((x_0,x_1)) &= (f_0(x_0),f_1(x_1)) \\
f(a(x_0,x_1)) &= (f_1(x_1),f_0(x_0)).
\end{align*}
If $f_0$ and $f_1$ differ by a hidden shift $s$ as above, and if $f_0$
is $H$-periodic and otherwise injective, then $f$ hides the subgroup
\[ H \cup a(s^{-1},s)H \subseteq G \wr \Z/2. \]
(Remark: This reduction of HShP to HSP generalizes a well-known reduction
from the graph isomorphism problem to graph automorphism problem.  Namely,
the set of isomorphisms between two graphs $\Gamma_0$ and $\Gamma_1$
carry the same information as the set of automorphisms of the disjoint
union $\Gamma_0 \sqcup \Gamma_1$.)

There is a simpler reduction from HShP to HSP when $G$ is abelian
\cite{K:dhsp,K:dhsp2}.  In this case, there is a semidirect product $G
\rtimes \Z/2$ generated by $G$ and an involution $a$ such that $axa = x^{-1}$
for $x \in G$.  We can define a hiding function $f:G \rtimes \Z/2 \to X$ by
\[ f(x) = f_0(x) \qquad f(xa) = f_1(x) \]
for $x \in G$.  Then $f$ hides
\[ H \cup saH \subseteq G \rtimes \Z/2. \]

\begin{proof}[Proof of \Cor{c:virtual}] Let $G$ be finitely generated and
virtually abelian, let $K \subseteq G$ be a finite-index abelian subgroup,
and suppose that $f:G \to X$ hides $H \subseteq G$.  We can compute $H
\cap K$ by \Cor{c:fgabel}.   Then since $G$ is a fixed group for the
algorithm and $[G:K]$ is finite, we can invoke HShP in $K/(H \cap K)$ a
bounded number of times (at most once for each generator of the quotient
group $G/K$) to find $H$.
\end{proof}

Finally, we mention a generalization of HSP and HShP that involves both
left and right multiplication in a non-abelian group $G$.  In the greatest
generality, a function $f$ on a group $G$ hides two subgroups $K$ and $H$
if it is constant on double cosets $KxH$ and otherwise injective:
\[ \begin{tikzpicture}
\draw (-1.25,0) node (G) {$G$};
\draw (1.25,0) node (X) {$X$};
\draw (0,-1.25) node (KH) {$K \backslash G/H$};
\draw[->] (G) -- node[above] {$f$} (X);
\draw[->>] (G) -- (H);
\draw[right hook->] (H) -- (X);
\end{tikzpicture} \]
This leads to various questions that may involve partial information about
the hidden structure.  For instance, $H$ might be known with an explicit
encoding for $G/H$, while $K$ might be unknown, so that we can view $f$ as
a function on the coset space $G/H$ with hidden $K$-periodicity.  Likewise,
two functions with right coset or double coset periodicity might differ by
a left hidden shift.  This generalization of HSP is not relevant for our
results, since double cosets $H$ and $K$ in $G$ are equivalent to single
cosets of $HK$ in $G/K$ if $K$ is normal (and vice versa if $H$ is normal).

\section{Proof of \Thm{th:rational}: HSP in $\Q$}
\label{s:rational}

Our proofs of Theorems \ref{th:rational}, \ref{th:free}, and \ref{th:svp}
all use hiding functions with a similar structure.  Given a group $G$
and a subgroup $H \subseteq G$, we can consider a hiding function $f$
that takes values in $G$ itself,
\[ \begin{tikzpicture}
\draw (-1.25,0) node (G) {$G$};
\draw (1.25,0) node (G1) {$G$};
\draw (0,-1.25) node (H) {$G/H$};
\draw[->] (G) -- node[above] {$f$} (G1);
\draw[->>] (G) -- (H);
\draw[right hook->] (H) -- (G1);
\end{tikzpicture} \]
and such that $f \circ f = f$.  In other words, $f(g) \in gH$ and is
a \emph{canonical representative} of the coset $gH$, generalizing
the concept of canonical words as described in \Sec{ss:explicit}.
For example, if $G = \Z$ and $H = n\Z$, then one familiar system of
canonical representatives is the half-open interval of remainders
\[ [0,n)_\Z \defeq \{0,1,2,\dots,n-1\}. \]
Computed canonical representatives thus generalize computed remainders in
the integer division algorithm.

\subsection{Partial fractions}
\label{ss:parfrac}

Our idea to calculate canonical representatives of cosets of a subgroup of
$\Q$ is based on the concept of partial fractions from calculus.  Partial
fractions have a generalized form in any principal ideal domain, where they
are related to the Chinese remainder theorem and the B\'ezout equation.
Recall that every rational function $a(x)/b(x)$ over $\R$ is uniquely a
sum of special terms called \emph{partial fractions}.  For example,
\[ \frac{x^6+x^4-x+2}{x^5-x^2} = x - \frac{x-1}{x^2+x+1}
    + \frac{1}{x-1} + \frac{1}{x} - \frac{2}{x^2}. \]
In this decomposition, the first term is a polynomial, while each other
term is of the form $r(x)/p(x)^k$, where $p(x)$ is irreducible over $\R$ and
$\deg r(x) < \deg p(x)$.  We will need the following analogous lemma in $\Q$.

\begin{lemma}[Partial fractions in $\Q$] Every rational number $a/b \in \Q$
has a unique representation
\begin{eq}{e:partial} \frac{a}{b} = n + \sum_{(p,k,r)} \frac{r}{p^k},
\end{eq}
where $n$ is any integer, each $p$ is a prime number, each exponent $k$
is positive, each numerator $r$ satisfies $1 \le r < p$, and each pair
$(p,k)$ appears at most once.  Moreover, given the prime factorization of
the denominator $b$, the partial fraction decomposition can be computed
in $\ccP$.
\label{l:partial} \end{lemma}

For example,
\begin{eq}{e:pfex}
\frac1{360} = -2 + \frac12 + \frac18 + \frac23 + \frac19 + \frac35.
\end{eq}
\Lem{l:partial} is stated and proved by Bourbaki
\cite[Sec.~VII.2.3]{Bourbaki:algebra2} in the greater generality of
an arbitrary principal ideal domain $A$ and its fraction field $F(A)$.
Partial fractions as they arise in calculus are then the special case $A =
\R[x]$.  However, since Bourbaki does not discuss algorithms, we provide a
(similar) proof that explains why it also provides an efficient algorithm.

Before proving \Lem{l:partial}, we restate it with a slightly different
canonical form.   We can sum all of the terms for each prime $p$ to retain
\equ{e:partial}, but now with the rule that each prime $p$ (rather than each
pair $(p,k)$) appears exactly once in the summation, and each numerator $r$
satisfies $1 \le r < p^k$ and $p\nmid r$.  For example, the abbreviated
form of \equ{e:pfex} is
\[ \frac1{360} = -2 + \frac58 + \frac79 + \frac35. \]

\begin{proof} If $b_1, b_2, \ldots, b_\ell$ are coprime
integers and $b$ is their product, then the standard Chinese remainder
theorem
\begin{eq}{e:chin} \Z/b \cong \Z/b_1 \oplus \Z/b_2 \oplus \cdots
    \oplus \Z/b_\ell \end{eq}
can be rescaled by $b$ and thus expressed as the group isomorphism
\begin{eq}{e:chinb} (\frac1b\Z)/\Z = (\frac1{b_1}\Z)/\Z \oplus
    (\frac1{b_2}\Z)/\Z \oplus \cdots \oplus (\frac1{b_\ell}\Z)/\Z. \end{eq}
We write \eqref{e:chinb} as an equality, which is valid because
$(\frac1{b_j}\Z)/\Z \subseteq (\frac1b\Z)/\Z$ for each $j$, and the right
side of \eqref{e:chinb} is then an internal direct sum.  Note that either
version of the Chinese remainder theorem has an algorithm in $\ccP$ using
modular reciprocals, and thus the B\'ezout equation and the Euclidean
algorithm.

We can also lift \equ{e:chinb} to a set arithmetic summation of intervals:
\begin{eq}{e:chinz} \frac1b\Z = \Z + \frac1{b_1}[0,b_1)_\Z +
    \frac1{b_2}[0,b_2)_\Z + \cdots + \frac1{b_\ell}[0,b_\ell)_\Z. \end{eq}
Since \eqref{e:chinb} is a bijection, \eqref{e:chinz} is also
a set arithmetic bijection.

Finally, since \eqref{e:chin} is a ring isomorphism, it yields the
multiplicative group isomorphism
\begin{eq}{e:chinm} (\Z/b)^\times \cong (\Z/b_1)^\times \oplus
    (\Z/b_2)^\times \oplus \cdots \oplus (\Z/b_\ell)^\times. \end{eq}
We can combine \equ{e:chinm} with \equ{e:chinz} to obtain the set arithmetic
bijection
\begin{eq}{e:parbij} (\frac1b\Z)^\times = \Z +
    \frac1{b_1}[0,b_1)_\Z^\times + \frac1{b_2}[0,b_2)_\Z^\times
    + \cdots + \frac1{b_\ell}[0,b_\ell)_\Z^\times. \end{eq}
Here $(\frac1b\Z)^\times$ is the set of rational numbers $a/b$ with
$\gcd(a,b) = 1$, while $[0,b)_\Z^\times$ is the set of integers $0 \le a <
b$ with $\gcd(a,b) = 1$.  Finally, \equ{e:parbij} is the bijection we want
if we take $b_j = p_j^{k_j}$ to be the prime power factors of $b$.
\end{proof}

\subsection{Prime numbers}

Our proof of \Thm{th:rational} depends on a few results and a conjecture
concerning prime numbers.  Recall the prime number theorem: For any $\eps >
0$, if an integer $n$ is randomly chosen in the range $x < n < x(1+\eps)$,
then the probability that $n$ is prime approaches $1/(\ln x)$ in ratio as
$x \to \infty$.  We will need the following important refinement
due to Ingham.

\begin{theorem}[Ingham \cite{Ingham:consec}]  Let $\eps > 0$ be fixed
and let $x$ be a real parameter.  Then the probability that a randomly
chosen integer $x^3 < n < (x+\eps)^3$ is prime approaches $1/(3 \ln x)$
in ratio as $x \to \infty$.  In particular, if $x$ is sufficiently large,
then such a prime $n$ exists.
\label{th:ingham} \end{theorem}

We also need the following result of \'Zra{\l}ek, which is based
on the Pollard $p-1$ factoring algorithm.  Recall that a number $n$
is \emph{$b$-smooth} when all of its prime factors are at most $b$.

\begin{theorem}[\'Zra{\l}ek {\cite[Cor.~4.6]{Zralek:pollard}}] If $n$
has a prime divisor $p$ such that $p-1$ is $b$-smooth for some $b$, then
we can find some non-trivial divisor of $n$ or prove that
$n$ is prime in deterministic time $\poly(b,\ln n)$.
\label{th:zralek} \end{theorem}

\begin{corollary} By induction, we can find all prime factors $p$ of $n$
such that $p-1$ is $b$-smooth in deterministic time $\poly(b,\ln n)$.
\label{c:zralek} \end{corollary}

\begin{remark} Pollard's $p-1$ algorithm combined with the Miller--Rabin
primality test yields an algorithm in $\BPP$ that calculates the same thing
as the algorithm in \Thm{th:zralek} or \Cor{c:zralek}.  Both algorithms
can also be derandomized assuming standard conjectures in number theory.
\'Zra{\l}ek's contribution is an unconditional derandomization of Pollard's
$p-1$ algorithm, together with the remark that we can replace Miller--Rabin
by the more recent deterministic AKS primality test \cite{AKS:primes}.
\end{remark}

Finally, we will use the following conjecture about existence of
prime numbers.

\begin{conjecture} Let $n > 0$ be an integer and let $x \in \{0,1\}^n$ be a
vector of bits of length $n$.  Then there is a prime $p$ such that $p-1$
is $\poly(n)$-smooth, and such that for every $1 \le k \le n$, the $k$th
odd prime $q_k$ divides $p-1$ exactly $x_k$ times.
\label{c:smoothp1} \end{conjecture}

\Conj{c:smoothp1} can be argued non-rigorously, if we interpret the prime
number theorem as a heuristic for frequencies of prime numbers with various
properties.  (Such a heuristic framework is known as the Cram\'er model.)

\subsection{The proof}
\label{ss:ratproof}

\begin{proof}[Proof of \Thm{th:rational}] For clarity, we work in $\Q$
rather than in $\Q/\Z$.  For simplicity, we first assume an integer
factoring oracle.

Let $d \in \NP$ be a decision problem with a predicate $z \in \ccP$, so that
$d(x) = \yes$ if and only if there exists $y$ such that $z(x,y) = \yes$.
By hypothesis, a witness $y$ is a bit string with $\abs{y} = \poly(\abs{x})$.
We assume a minimum length $\abs{y} = m \ge m_0$, where $m_0$ exists by
\Thm{th:ingham} with the following property:   Every bit string $y$ of
length $m \ge m_0$ can be prepended with a 1 and appended by some bit
string of length $2m$ to obtain a $3m+1$-bit prime $p$.

Given our assumptions, we can replace the predicate $z(x,y)$ with a new
predicate $\tz(x,p)$ for the same decision problem $d(x)$, such that
every witness $p$ of $\tz$ is a prime number.  Given a putative prime $p$
to be used as a witness, we first check that it is prime using either
the factoring oracle or the AKS primality algorithm \cite{AKS:primes}.
We also reject $p$ if its binary expansion does not have $3m+1$ bits for
some $m \ge m_0$.  If $p$ does have $3m+1$ bits, then after skipping the
leading bit (which is always 1), we let $y$ be the next $m$ high bits.
In this case, we let $\tz(x,p) = z(x,y)$.   By \Thm{th:ingham}, the
predicates $z$ and $\tz$ produce the same decision problem $d \in \NP$.

We construct the hiding function $f$ from $\tz(x,p)$ for some fixed
value of $x$ as follows.  We expand an input $a/b$ to $f$ into partial
fractions using \Lem{l:partial}.  We then let the value of $f(a/b)$
be the same expansion with the integer term $n$ deleted and with each
term $r/p$ deleted when $\tz(x,p) = \yes$.  (We can also make $f(a/b)$
a canonical representative in the standard syntax for fractions by summing
this new expansion.)  The hidden subgroup $H$ is thus generated by 1 and by
$1/p$ for every prime $p$ which is an accepted certificate.  As desired,
an algorithm to answer whether $H \ne \Z$ would also answer whether $d(x)
= \yes$.

We now turn to the version of the theorem that does not require a factoring
oracle, but is conditional on \Conj{c:smoothp1}.  In this case, we can
map a prime $p$ given by \Conj{c:smoothp1} to a certificate $y$ for the
predicate $z(x,y)$ by encoding $y$ in the exponents of the first $m$ odd
primes in the prime factorization of $p-1$.  We can then let $\tz(x,p) =
z(x,y)$, when the certificate $p$ is a prime of this form.   In constructing
the modified predicate $\tz$, we also reject any composite certificate $p$
using the AKS primality test, and any prime certificate $p$ that does match
the description in \Conj{c:smoothp1} by finding all small factors of $p-1$.

We construct a modified hiding function $f$ using
\Cor{c:zralek}.  Given an input $a/b$, the corollary gives us a partial
factorization of the denominator
\[ b = p_1^{k_1} p_2^{k_2} \dots p_\ell^{k_\ell} c, \]
such that each prime factor $q = p_j$ can be any prime $q$ that fits the
conclusion of \Conj{c:smoothp1}, and $c$ is a leftover factor which may
be composite but which is coprime to the other factors.  The proof of
\Lem{l:partial}, in particular \equ{e:parbij}, applies equally well to
any factorization of $b$ into coprime factors rather than only the prime
power factorization.  This yields a canonical form
\[ \frac{a}{b} = \frac{n}{c} + \sum_{j=1}^\ell \frac{r_j}{p_j^k}. \]
We can still let $H$ be generated by 1 and by $1/q$ for each witness
$q$ accepted by $\tz(x,q)$.   To calculate the canonical representative
of $\frac{a}{b}+H$, we can discard all terms on the right with $k=1$ and
$\tz(x,p_j) = \yes$, and we can reduce $n$ mod $c$ instead of
discarding it entirely.
\end{proof}

\subsection{A sparse version of Simon's problem}
\label{ss:simon}

One lesson of our results is that the computational complexity of HSP for an
infinite group depends greatly on how elements are encoded.  This lesson
also applies to the original HSP solved by Simon \cite{Simon:power},
namely HSP in finite groups of the form $(\Z/2)^k$, which was a precursor
to Shor's algorithm.

Let $(\Z/2)^\infty$ be the group of infinite binary vectors with finitely
many non-zero coordinates.  We can use the standard dense encoding of
elements of $(\Z/2)^\infty$ as bit strings with the infinite tail of
zeroes removed. With this encoding, HSP in $(\Z/2)^\infty$ is equivalent
to the standard version in $(\Z/2)^k$ that is solved by Simon's algorithm.
Given the complexity parametrization described in \Sec{ss:hidden}, we
can assume that $H$ is generated by vectors with polynomially bounded bit
strings, so that $H \subseteq (\Z/2)^k$ with $k$ polynomially bounded.

Alternatively, we can uniquely encode every vector $\vv \in (\Z/2)^\infty$
as a sparse sum
\begin{eq}{e:sparse} \vv = \ve_{y_1} + \ve_{y_2} + \cdots + \ve_{y_\ell},
\end{eq}
where by definition $\ve_y$ is the $y$th standard basis vector and the
indices are listed in strictly increasing order.  Crucially, we also write
each index $y_j$ in binary.

\begin{theorem} HSEP for the group $(\Z/2)^\infty$ with sparse encoding
is $\NP$-hard.
\label{th:simon} \end{theorem}

The proof of \Thm{th:simon} is a simplified version of the proof of
\Thm{th:rational}.  We assume a decision problem $d \in \NP$ defined
by a predicate $z \in \ccP$ and particular input $x$ to $d$.  We assume
for convenience that every witness $y$ that $z(x,y)$ accepts ends in 1.
We can then define a hiding function $f:(\Z/2)^\infty \to (\Z/2)^\infty$
as follows: Given $\vv \in (\Z/2)^\infty$ in the form \eqref{e:sparse},
we define $f(\vv)$ by deleting those terms $\ve_y$ such that $z(x,y) = \yes$.
Then $f$ hides a subgroup of $(\Z/2)^\infty$ that is non-trivial if and
only if $d(x) = \yes$, as desired.

In \Thm{th:simon}, we can also replace the infinitely generated group
$(\Z/2)^\infty$ by a finite group $(\Z/2)^k$, provided that $k$ is
exponential in $\abs{x}$ and we still use sparse notation for $\vv \in
(\Z/2)^k$.

Finally, we can extend the ideas in this subsection to the group $\Z^\infty$.
A vector $\vv \in \Z^\infty$ can either be written in dense encoding as
a vector in $\vv \in \Z^n \subseteq \Z^\infty$; or in sparse encoding
along the lines of \equ{e:sparse}, except with integer coefficients.
As discussed in \Thm{th:svp} and \Sec{ss:complex}, each integer coefficient
can be written either in binary or unary.  The proof of \Thm{th:simon} shows
that HSP in $\Z^\infty$ is $\NP$-hard with sparse encoding of the vectors,
and with either encoding of the coefficients.  Meanwhile \Cor{c:infabel}
addresses the two cases with dense encoding of vectors in $\Z^\infty$.

\section{Proof of \Thm{th:free}: NHSP in free groups}
\label{s:free}

As in the proof of \Thm{th:rational}, we can prove \Thm{th:free} with
a hiding function $f$ on the free group $F_k$ whose values are canonical
representatives of cosets of a hidden subgroup $N \normaleq F_k$.  Since $N$
is normal, it can be \emph{normally} generated by certificates $w \in F_k$
that satisfy an $\NP$ predicate, meaning that $N$ is generated by these
elements and their conjugates.  Also since $N$ is normal, the coset space
$F_k/N$ is a quotient group $G \cong F_k/N$, which is thus a finitely
presented group.  Canonical representatives in $F_k$ of cosets of a normal
subgroup $N$ are then exactly canonical words for group elements $g \in G$.

An efficient algorithm to compute canonical words in $G \cong F_k/N$
yields an efficient algorithm for the word problem in $G$.  By definition,
the \emph{word problem} in $G$ is the algorithmic question of whether two
words $v$ and $w$ represent the same element of $G$. Equivalently, the word
problem asks whether an arbitrary word $u = v^{-1}w$ represents $1 \in G$,
or equivalently whether $u \in N$.  An algorithm for the word problem also
implies an algorithm for canonical words:  Given a word $w$, we can let its
canonical name be its shortlex equivalent as defined in \Sec{ss:present}
below, and we can find $v$ by searching through all words up to the length
of $w$.  However, an efficient algorithm for the word problem does not
necessarily yield an efficient algorithm to compute canonical words.

\begin{remark} The word problem in finitely presented, non-commutative groups
is a major topic in its own right.  It has a polynomial-time algorithm in
many finitely presented groups \cite{ECHLSPT:word}, it can be equivalent
to the halting problem in others \cite{Boone:word,Novikov:word}, and it is
computable but arbitrarily intractable in yet others \cite{KMS:complex}.
Determining whether a finitely presented group is trivial is a related
problem that is also as hard as the halting problem.  In our proof of
\Thm{th:free}, we will choose $G \cong F_k/N$ from a special class of groups
that not only have a polynomial-time algorithm to compute canonical words,
but even an algorithm that is uniformly efficient given only query access
to the relators.
\end{remark}

\subsection{Group presentations}
\label{ss:present}

In this subsection and the next two, we will review some material from
combinatorial group theory and small cancellation theory.  We will largely
follow an exposition by Strebel \cite{Strebel:small}, although some of
the terminology in \Sec{ss:small} is from Peifer \cite{Peifer:artin}.

Let $A$ be an alphabet, which in parts of this subsection and the next two
can be either finite or infinite.  Recall from \Sec{ss:explicit} that $W_A =
(A \cup A^{-1})^*$ is the semigroup of words in letters in $A$ and their
formal inverses.  The free group $F_A$ can be defined as the quotient of
$W_A$ by the equivalence $aa^{-1} \sim 1 \sim a^{-1}a$ for every $a \in A$.
A word $w \in W_A$ is \emph{reduced} if it does not have an adjacent pair
of inverse letters.  Every element $F_A$ is represented by a unique reduced
word, which is thus a canonical word.  By abuse of notation, we will also
view $F_A$ as the subset of $W_A$ consisting of reduced words.

Let $C_A$ be the set of conjugacy classes in $F_A$. Since every word
$w \in F_A$ is conjugate to its cyclic permutations, we can specify a
conjugacy class in $C_A$ by a \emph{cyclic word} $t$, meaning an orbit
of words under cyclic permutation.   The conjugacy class described by $t$
does not change if we remove an adjacent, inverse pair of letters anywhere
along the circle of letters of $t$.  If $t$ does not have such a pair of
letters, then it is \emph{cyclically reduced}, or reduced as a cyclic word.
Every conjugacy class in $C_A$ is represented by a unique reduced cyclic
word, and we will identify $C_A$ with the set of such words.  Also, by abuse
of notation, we will assume that we can concatenate two linear words $v$
and $w$ at both ends rather than at one end to form a cyclic word $t = vw$.

A \emph{presented} group $G$ is a construction of the form
\[ G = \braket{A \mid R} \defeq F_A/N_R, \]
where $R$ is a set of words called \emph{relators}, and $N_R \normaleq F_A$
is the subgroup of $F_A$ normally generated by $R$.  We can also describe $G$
as the quotient of $F_A$ generated by the equivalences $r \sim 1$ for all $r
\in R$, and we will write $v \sim w$ when $v,w \in F_A$ represent
the same element in $G$.  If $A$ is finite, then $G$ is \emph{finitely
generated}; if $R$ is also finite, then $G$ is \emph{finitely presented}.

If $G$ is a presented group, then we call $w \in F_A$ a \emph{shortest
word}, or in geometric language a \emph{geodesic}, if it is shortest
in the equivalence class of words $w$ that represent $[w] = g \in G$.
If the alphabet $A$ of the group $G$ is finite and we assume some ordering
of $A \cup A^{-1}$, then there is a corresponding \emph{shortlex} ordering
of all of the words in $W_A$ given by first ordering words by length, and
then within each fixed length ordering them alphabetically.  The shortlex
representative $w$ for $[w] = g \in G$ is by definition the one that is first
in the shortlex ordering; equivalently, the alphabetically first geodesic.

Since $N_R$ is normally generated by the relators $R$, we can take $R$ to
be a set of conjugacy classes in $F_A$ rather than a set of words in $F_A$.
Then since $R \subseteq C_A$, we can take the relators in $R$ to be reduced
cyclic words.  We can also assume that $r \in R$ implies $r^{-1} \in R$,
and that $1 \notin R$.  A set of cyclically reduced relators with these
properties is called \emph{symmetrized}.

As remarked above, if $G$ is a finitely presented group, then this by
itself yields few restrictions on the computational complexity of the word
problem or the canonical word problem in $G$.  To prove \Thm{th:free},
we will consider specific presented groups that have efficient solutions
to their word problems, because their relators do not share long subwords.
This type of behavior for relators is called \emph{small cancellation}.
More formally, if $\lambda > 0$, a presented group $G$ is $C'(\lambda)$
small cancellation means that if a subword $s$ of a relator $r \in R$ appears
elsewhere either in $r$ or in another relator $r'$, then $\abs{s} < \lambda
\abs{r}$.  This definition is due to Greendlinger \cite{Greendlinger:word},
who used it to state and prove \Thm{th:green}, in a paper that initiated
small cancellation theory as a topic in group theory.

\subsection{Diagrams of relators}
\label{ss:diagram}

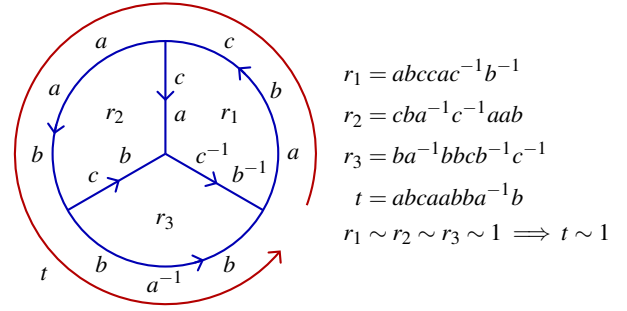
\begin{figure}
\begin{tikzpicture}[thick,draw=darkblue,>=angle 90]
\draw (0,0) circle (1.5);
\draw[->] (90:.75) -- (90:.7); \draw (0,0) -- (90:1.5);
\draw[->] (210:.75) -- (210:.7); \draw (0,0) -- (210:1.5);
\draw[->] (330:.75) -- (330:.8); \draw (0,0) -- (330:1.5);
\draw[->] (45:1.5) arc (45:50:1.5);
\draw[->] (165:1.5) arc (165:170:1.5);
\draw[->] (285:1.5) arc (285:290:1.5);
\foreach \th/\x in {0/a,30/b,60/c,120/a,150/a,180/b,240/b,270/a^{-1},300/b}
    \draw (\th:1.7) node {$\x$};
\draw (90:.5) node[right] {$a$};
\draw (90:1) node[right] {$c$};
\draw (210:.5)+(.1,0) node[above left] {$b$};
\draw (210:1)+(.1,0) node[above left] {$c$};
\draw (330:.5)+(-.15,0) node[above right] {$c^{-1}$};
\draw (330:1)+(-.1,0) node[above right] {$b^{-1}$};
\draw (30:1) node {$r_1$};
\draw (150:1)+(.2,0) node {$r_2$};
\draw (270:.9) node {$r_3$};
\draw[->,darkred] (-20:2) arc (-20:320:2);
\draw (-1.6,-1.6) node {$t$};
\draw (2.25,0) node[right] {$\begin{aligned}
    r_1 &= abccac^{-1}b^{-1} \\
    r_2 &= cba^{-1}c^{-1}aab \\
    r_3 &= ba^{-1}bbcb^{-1}c^{-1} \\
    t &= abcaabba^{-1}b \\
    r_1 &\sim r_2 \sim r_3 \sim 1 \implies t \sim 1
    \end{aligned}$};
\end{tikzpicture}
\caption{A cyclic van Kampen diagram in a group with at least
    3 generators and least 3 relators.}
\label{f:cyclic} \end{figure}

We can model the calculus of relators in a presented group $G$ by a
planar structure called a \emph{cyclic (van Kampen) diagram}.
By definition, such a diagram $D$ is an undirected planar graph whose
vertices all have valence at least 3 (but which is allowed multiple edges
and self-loops) with the following type of decoration:
\begin{enumerate}
\item Each edge of $D$, if given a direction, is assigned a non-trivial,
reduced word $w \in F_A$.  The same edge is assigned $w^{-1}$ in the
other direction.
\item The cyclic word of each face of $D$ (other than the outer face)
is some relator $r \in R$.
\end{enumerate}
Finally, the \emph{value} of the diagram $D$ is the cyclic word of its
outer face, read counterclockwise.  We also require that the value of $D$
be cyclically reduced.  \Fig{f:cyclic} shows an example.

Just as we can cancel adjacent inverse letters in a word $w \in W_A$, we
can also cancel adjacent faces in a cyclic diagram $D$ that are labelled
with inverse relators.  More precisely, there is a way to cancel two faces
$P$ and $Q$ if they share exactly one edge, and if $P$ and $Q$
yield inverse linear words $r$ and $r^{-1}$ when traversed from a common
vertex $p$.  We call $D$ \emph{reduced} if no such pair of faces exists.

\begin{remark} If $P$ and $Q$ share more than one edge, then it is only
possible to cancel them (and remove the other faces that they surround)
if all of their intersections lie in matching positions along the inverse
relators $r$ and $r^{-1}$.  On the other hand, we can cancel $P$ and $Q$
when they only share a vertex $p$ rather than an edge containing it
\cite[Sec.~11.6]{Olshanskii:geometry}. Following Strebel, we will
not use these more general ways to cancel faces.
\end{remark}

\begin{proposition}[After van Kampen
{\cite[Sec.~2.1~\&~Thm.~18]{Strebel:small}}] The value of any cyclic
diagram $D$ over a presented group $G$ is trivial in $G$.  Conversely,
if $t$ is a reduced cyclic word which is trivial in $G$, then $t$ is the
value of a reduced diagram $D$.
\label{p:value} \end{proposition}

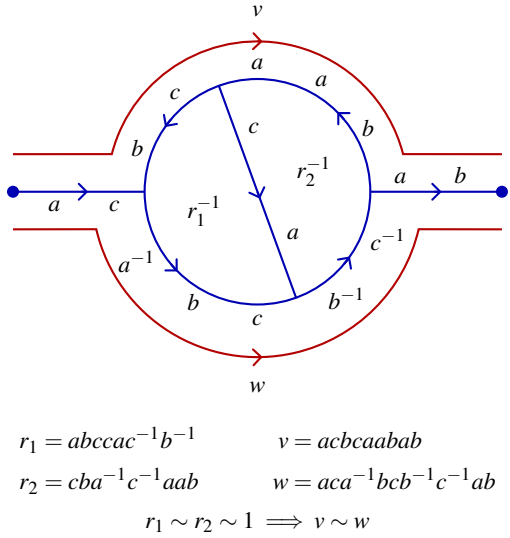
\begin{figure}
\begin{tikzpicture}[thick,draw=darkblue,fill=darkblue,>=angle 90]
\draw (0,0) circle (1.5);
\draw (110:1.5) -- (290:1.5);
\draw[->] (0,0) -- (290:.1);
\draw[->] (40:1.5) arc (40:45:1.5);
\draw[->] (140:1.5) arc (140:145:1.5);
\draw[->] (220:1.5) arc (220:225:1.5);
\draw[->] (320:1.5) arc (320:325:1.5);
\fill (-3.25,0) circle (.08); \draw (-3.25,0) -- (-1.5,0);
\draw[->] (-2.35,0) -- (-2.25,0);
\fill (3.25,0) circle (.08); \draw (1.5,0) -- (3.25,0);
\draw[->] (2.35,0) -- (2.45,0);
\foreach \th/\x in {30/b,60/a,90/a,130/c,160/b,240/b,270/c}
    \draw (\th:1.7) node {$\x$};
\foreach \th/\x in {210/a,310/b,340/c} \draw (\th:1.85) node {$\x^{-1}$};
\draw (-2.7,0) node[below] {$a$};
\draw (-1.9,0) node[below] {$c$};
\draw (1.9,0) node[above] {$a$};
\draw (2.7,0) node[above] {$b$};
\draw (110:.7) node[above right] {$c$};
\draw (290:.8) node[above right] {$a$};
\draw (20:.8) node {$r_2^{-1}$};
\draw (200:.75) node {$r_1^{-1}$};
\draw[darkred] (-3.25,.5) -- (165.5:2) arc (165.5:14.5:2) -- (3.25,.5);
\draw[->,darkred] (90:2) arc (90:88:2);
\draw[darkred] (-3.25,-.5) -- (193.1:2.2) arc (193.1:346.9:2.2)
    -- (3.25,-.5);
\draw[->,darkred] (270:2.2) arc (270:272:2);
\draw (0,2.4) node {$v$};
\draw (0,-2.6) node {$w$};
\draw (0,-3) node[below] {$\begin{gathered} \begin{aligned}
    r_1 &= abccac^{-1}b^{-1} &\qquad v &= acbcaabab \\
    r_2 &= cba^{-1}c^{-1}aab &\qquad w &= aca^{-1}bcb^{-1}c^{-1}ab
    \end{aligned} \\ r_1 \sim r_2 \sim 1 \implies v \sim w
    \end{gathered}$};
\end{tikzpicture}
\caption{An equality diagram in a group with 3 generators and 2 relators.}
\label{f:equiv} \end{figure}

We will consider an extension of a cyclic diagram called an
\emph{equality diagram} to capture when two reduced words $v,w \in F_A$
are equivalent.  An equality diagram $D$ has the same definition as a
cyclic diagram, except that we label two vertices on the outer face of $D$
as the \emph{start} $p$ and \emph{end} $q$, and we allow these vertices
to have any valence (including 1 or 2).   The value of $D$ is now
the pair $(v,w)$, where $v$ is an outer arc from $p$ to $q$ clockwise
and $w$ is the counterclockwise arc.  See \Fig{f:equiv} for an example.

We write $v \sim_D w$ to summarize that $v$ and $w$ are equivalent via
the diagram $D$.   We let $D^{-1}$ denote the reflected diagram, so
that $v \sim_D w$ if and only if $w \sim_{D^{-1}} v$. In this notation,
\Prop{p:value} implies that $v \sim w$ if and only if $v \sim_D w$ for
some reduced equality diagram $D$.

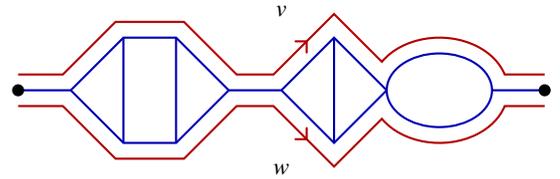
\begin{figure}
\begin{tikzpicture}[thick,draw=darkblue,scale=.7,>=angle 90]
\draw (0,0) -- (1,0) -- (2,1) -- (3,1) -- (4,0) -- (5,0) -- (6,1) -- (7,0);
\draw (1,0) -- (2,-1) -- (3,-1) -- (4,0) (5,0) -- (6,-1) -- (7,0);
\draw (2,1) -- (2,-1) (3,1) -- (3,-1) (6,1) -- (6,-1);
\draw (8,0) ellipse (1 and .7);
\draw (9,0) -- (10,0);
\fill (0,0) circle ({.08/.7});
\fill (10,0) circle ({.08/.7});
\draw[darkred] (0,.3) -- (.85,.3) -- (1.85,1.3) -- (3.15,1.3) -- (4.15,.3)
    -- (4.85,.3) -- (6,1.45) -- (6.908,.542)
    arc (147.2:17.46:1.3 and 1) -- (10,.3);
\draw[->,darkred] (5.45,.9) -- (5.5,.95);
\draw[darkred] (0,-.3) -- (.85,-.3) -- (1.85,-1.3) -- (3.15,-1.3)
    -- (4.15,-.3) -- (4.85,-.3) -- (6,-1.45) -- (6.908,-.542)
    arc (-147.2:-17.46:1.3 and 1) -- (10,-.3);
\draw[->,darkred] (5.45,-.9) -- (5.5,-.95);
\draw (5,1.5) node {$v$};
\draw (5,-1.5) node {$w$};
\end{tikzpicture}
\caption{A thin equality diagram with three thin disks.}
\label{f:thin} \end{figure}

An equality diagram $D$ for $v \sim w$ is \emph{thin} if each face intersects
each of $v$ and $w$ in an edge.  (Since we assume that $v$ and $w$ are
both reduced, $D$ is necessarily reduced as well.)  Equivalently, $D$ is
thin when it is a concatenation of edges, digons, triangles, and squares.
We also define a \emph{thin disk} to be a thin diagram that (with its
faces filled in) is a topological disk.  A general thin diagram $D$ is a
concatenation of thin disks and edges; see \Fig{f:thin} for an example.

\begin{remark} Peifer \cite[Sec.~5]{Peifer:artin} defines a thin equality
diagram so that every relator meets each of $v$ and $w$ in at least a
vertex rather than necessarily an edge.  Our definition could be called
``strictly thin", or his could be called ``weakly thin".
\end{remark}

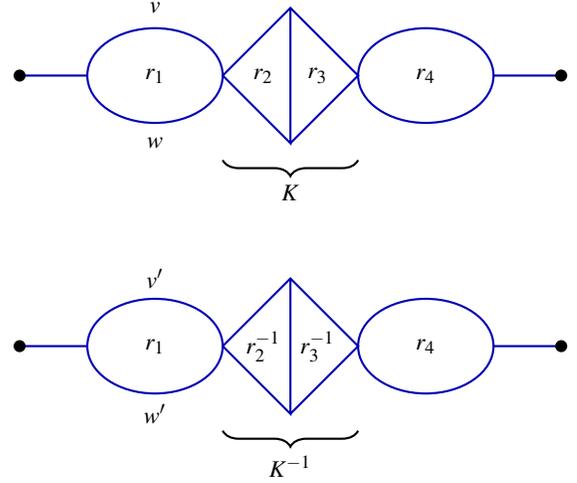
\begin{figure}
\begin{tikzpicture}[thick,draw=darkblue,scale=.9]
\draw (0,0) -- (1,0); \draw (2,0) ellipse (1 and .7);
\draw (4,1) -- (4,-1) -- (3,0) -- (4,1) -- (5,0) -- (4,-1);
\draw (6,0) ellipse (1 and .7); \draw (7,0) -- (8,0);
\fill (0,0) circle ({.08/.9}); \fill (8,0) circle ({.08/.9});
\draw (2,1) node {$v$}; \draw (2,-1) node {$w$};
\draw (2,0) node {$r_1$}; \draw (3.6,0) node {$r_2$};
\draw (4.4,0) node {$r_3$}; \draw (6,0) node {$r_4$};
\draw[black,decoration={brace,mirror,amplitude=1.5ex},decorate]
    (3,-1.25) -- node[below=1.5ex] {$K$} (5,-1.25);
\begin{scope}[yshift=-4cm]
\draw (0,0) -- (1,0); \draw (2,0) ellipse (1 and .7);
\draw (4,1) -- (4,-1) -- (3,0) -- (4,1) -- (5,0) -- (4,-1);
\draw (6,0) ellipse (1 and .7); \draw (7,0) -- (8,0);
\fill (0,0) circle ({.08/.9}); \fill (8,0) circle ({.08/.9});
\draw (2,1) node {$v'$}; \draw (2,-1) node {$w'$};
\draw (2,0) node {$r_1$}; \draw (3.6,0) node {$r^{-1}_2$};
\draw (4.4,0) node {$r^{-1}_3$}; \draw (6,0) node {$r_4$};
\draw[black,decoration={brace,mirror,amplitude=1.5ex},decorate]
    (3,-1.25) -- node[below=1.5ex] {$K^{-1}$} (5,-1.25);
\end{scope}
\end{tikzpicture}
\caption{Twisting a thin diagram.}
\label{f:twist} \end{figure}

We will use thin diagrams in \Thm{th:canonical} below to classify
all of the geodesics that represent a given group element $g \in G$
in a small cancellation group.

If $v \sim_D w$ is a thin equality diagram, then we define a
\emph{transversal} of $D$ to be an embedded path along its edges that
connects the ends of $D$, and a \emph{boundary transversal} to be a
transversal that does not cut through the interior of any thin disk $K$
in $D$.  If $v$ and $w$ are both geodesics, then their length along every
disk $K$ must be equal; therefore, all of the boundary transversals are
also geodesics.   If some disk $K$ has an interior edge, then $D$ also has
non-boundary transversals.  Some subset of the non-boundary transversals
might also be geodesics, or maybe none are.  If $D$ has any non-boundary
transversals, then there is at least one which is not a geodesic.

If $v \sim_D w$ is a thin equality diagram between geodesics, then we define
a \emph{twist} of $D$ to be a new equality diagram $v' \sim_{D'} w'$ that
is obtained by flipping over a thin disk $K$ in $D$, \ie, replacing $K$
by $K^{-1}$, as in \Fig{f:twist}.  We can attempt the same definition for
a thin equality diagram $v \sim_D w$ between arbitrary words $v$ and $w$.
However, in the general case the result $v' \sim_{D'} w'$ might not be a
legal equality diagram according to the definition in \Sec{ss:diagram},
because one or both of $v'$ and $w'$ might be non-reduced.  If $v$ and $w$
are both geodesics, then this anomaly is impossible, because the reduced
form of $v'$ would be shorter than $v$.  If $v$ and $w$ are both geodesics,
then $v \sim_D w$ and any twist $v' \sim_{D'} w'$ contain the same geodesics.

\subsection{Small cancellation}
\label{ss:small}

The goal of this subsection is the following theorem.

\begin{theorem} If $G$ is a finitely presented, $C'(1/6)$ small cancellation
group, then there is an algorithm to find the shortlex canonical form for
any word $w$ which is uniformly polynomial time in the presentation of $G$
and the choice of $w$.
\label{th:shortlex} \end{theorem}

In \Sec{ss:freeproof}, we will refine the proof of \Thm{th:shortlex}
to prove \Thm{th:free}.

The author did not find \Thm{th:shortlex} in the existing literature,
although related results are known.  For instance, we can find $w$
in polynomial time in its length for a fixed choice of $G$ from
the established results that $C'(1/6)$ groups are word hyperbolic
\cite[Thm.~36]{Strebel:small} and that word hyperbolic groups are automatic
\cite[Thm.~3.4.5]{ECHLSPT:word}.

We will establish \Thm{th:shortlex} as a corollary of the following result
about geodesics in a $C'(1/6)$ group.

\begin{theorem} If $G$ is a presented $C'(1/6)$ small cancellation group
and if $g \in G$, then there is a thin diagram $D$ between some pair of
geodesics for $g$ that contains every geodesic for $g$ as a transversal.
$D$ is canonical up to thin diagram twists.  If $G$ is finitely
presented, then $D$ can be computed in polynomial time, uniformly in the
presentation of $G$ and a word $w$ that represents $g$.
\label{th:canonical} \end{theorem}

We will define certain terms used in the statement of \Thm{th:canonical}
later in this subsection: shortlex, thin diagram, twist, and transversal.

Note that Arzhantseva and Dru\c{t}u \cite{AD:small} established the existence
part of \Thm{th:canonical} for $C'(1/8)$ groups.  (They refer to a thin
equality diagram with the listed extra properties as a \emph{criss-cross
decomposition}.)

\begin{remark} \Thm{th:canonical} provides a computed canonical name for
each $g \in G$ when $G$ is a $C'(1/6)$ group.   This is enough to construct
a hiding function that we can use to prove \Thm{th:free}, albeit that a
canonical word as provided by \Thm{th:shortlex} would also work.  However,
we will need to refine Theorems \ref{th:shortlex} and \ref{th:canonical}
for a special class of $C'(1/6)$ groups so that the algorithms can work
with only query access to the relators.
\end{remark}

In the rest of this subsection, let $G = \braket{A | R}$ be a presented
group with a symmetrized relator set $R$.

If $v,w \in F_A$ are two reduced words, then we call $v$ a \emph{reducing
Dehn move} of $w$ if $v \sim_D w$ via an equality diagram $D$ with only one
face $P$, and also $\abs{v} < \abs{w}$.  Explicitly, if $P$ is labelled
by $r \in R$, then it has a decomposition $r = s^{-1}t$ with $\abs{s} >
\abs{r}/2$ where $w = xsy$ and $v = xty$.  A word $w \in F_A$ is \emph{Dehn
reduced} if it does not admit any reducing Dehn moves.  The non-deterministic
greedy algorithm that makes a Dehn-reduced word from any word is called
\emph{Dehn's algorithm}.  Certainly any geodesic is Dehn reduced, but most
Dehn-reduced words in most presented groups are not geodesics.

\begin{theorem}[Greendlinger {\cite[Thm.~25]{Strebel:small}}] If $G$ is a
$C'(1/6)$ presented group, then the trivial word is the only Dehn-reduced
representative of $1 \in G$.  If $G$ is also finitely presented, then
Dehn's algorithm efficiently solves the word problem in $G$.
\label{th:green} \end{theorem}

\begin{remark} In light of \Thm{th:green}, one might hope that if $G$
has some reasonable small cancellation property, then Dehn's algorithm can
also find a canonical word or at least a shortest word.  However, even in
many $C'(\lambda)$ groups with arbitrarily small $\lambda > 0$, there are
Dehn-reduced words that are not shortest \cite[Ex.~38]{Strebel:small}.
\end{remark}

\Thm{th:green} quickly follows from a lemma that is at least as important
as the theorem itself.

\begin{lemma}[Greendlinger {\cite[Rem.~26]{Strebel:small}}] Let $D$ be a
non-trivial reduced cyclic diagram over a $C'(1/6)$ group $G$. Then either
$D$ is a single relator, or it has an outer edge that is more than half
of the length of its relator.
\label{l:green} \end{lemma}

\Lem{l:green} is a simplification of what is now known as
Greendlinger's lemma \cite{Greendlinger:word}.   The full version of the
lemma says that every cyclic diagram has between two and six long outer
edges, and when there are fewer of them there are better lower bounds on
their lengths.  To prove \Thm{th:green}, it suffices to know that at least
one cyclic diagram has a long outer edge.  However, later results such as
\Lem{l:thin} use the full strength of Greendlinger's lemma.

In the rest of this subsection, we will consider diagrams in a $C'(1/6)$
group $G$.

\begin{lemma}[Strebel {\cite[p.~258]{Strebel:small}}] Every reduced equality
diagram $v \sim_D w$ between Dehn-reduced words $v$ and $w$ is thin.
\label{l:thin} \end{lemma}

Strebel states \Lem{l:thin} in the proof of his Proposition 39(i).  He refers
to the proof of his Theorem 35 for the argument in the main special case
of a thin disk.  (Note that Lyndon and Schupp \cite[Thm~V.5.5]{LS:combin}
earlier established a similar lemma about annular thin diagrams.)

\begin{lemma} Let $v \sim_D w$ be an equality diagram between Dehn-reduced
words, and let $e$ be an outer edge of $D$ labelled by $s$ that lies on
a face $P$ labelled by $r$.
\begin{enumerate}
\item If $P$ is a digon, then $\abs{s} = \abs{r}/2$.
\item If $P$ is at one end of a thin disk, then $\abs{s} > \abs{r}/3$.
\item If $P$ is in the middle of a thin disk, then $\abs{s} > \abs{r}/6$.
\end{enumerate}
\label{l:lengths} \end{lemma}

\begin{proof} Let $e$ be an edge of $D$ that could be either an inner
edge or an outer edge, and that is labelled by $s$ and lies on a face $P$
labelled by $r$.    If $e$ is an inner edge, then the $C'(1/6)$ condition
says that $\abs{s} < \abs{r}/6$.  Meanwhile if $e$ is an outer edge, then
$\abs{s} \le \abs{r}/2$ since $v$ and $w$ are both Dehn reduced.  Since $F$
has two outer edges in all cases, and since it has 0, 1, and 2 inner edges
(respectively) in each of the three cases, the inequalities are obtained
by subtraction.
\end{proof}

\begin{lemma} Let $u \sim_D v \sim_E w$ be two equality diagrams
between Dehn-reduced words which each have only one thin disk.
Then exactly one of the following holds:
\begin{enumerate}
\item The combined diagram $u \sim_{D \cup E} w$ (as in \Fig{f:concat})
is thin.
\item The disks of $D$ and $E^{-1}$ contain a common thin disk $K$.
If $K$ is taken to be maximal, then the union $D \cup_K E^{-1}$ is a thin
diagram for an equivalence $z \sim v$ for some reduced word $z$.
\end{enumerate}
\label{l:merge} \end{lemma}

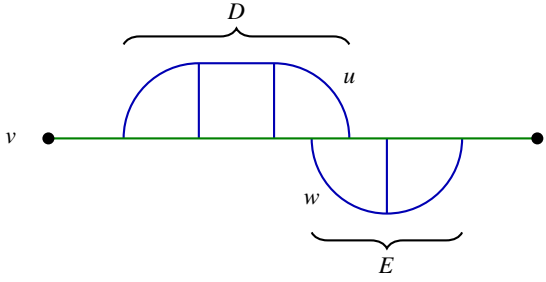
\begin{figure}
\begin{tikzpicture}[thick,draw=darkblue]
\draw (1,0) arc (180:90:1) -- (3,1) arc (90:0:1);
\draw (2,0) -- (2,1); \draw (3,0) -- (3,1);
\draw (3.5,0) arc (180:360:1);
\draw (4.5,0) -- (4.5,-1); \draw[darkgreen] (0,0) -- (6.5,0);
\draw[black,decoration={brace,mirror,amplitude=1.5ex},decorate]
    (3.5,-1.25) -- node[below=1.5ex] {$E$} (5.5,-1.25);
\draw[black,decoration={brace,amplitude=1.5ex},decorate]
    (1,1.25) -- node[above=1.5ex] {$D$} (4,1.25);
\fill (0,0) circle (.08); \fill (6.5,0) circle (.08);
\draw (-.5,0) node {$v$};
\draw (4,.8) node {$u$};
\draw (3.5,-.8) node {$w$};
\end{tikzpicture}
\caption{Concatenating two thin disks to get a thin disk (with the dividing
path indicated in green).}
\label{f:concat} \end{figure}

\begin{proof} We assume without loss of generality that the disk of $D$
begins no later along $v$ than the disk of $E$ does.  In case 1, the union
$D \cup E$ is a thin diagram when either (a) their disks do not share any
length along $v$; or (b) the last face of $D$ meets only the first face
of $E$.  In case 1(b), as shown in \Fig{f:concat}, $D$
and $E$ concatenate to make a larger thin disk $D \cup E$.

Henceforth, we assume that case 1 does not hold.  We work from left to
right along $v$ to show instead that the vertices of $D$ and $E$ must
match along $v$, except that the disk of $D$ may start before the disk
of $E$, and either disk may end first.  Also working from left to right,
we will show that the faces of $D$ and $E$ match with inverse relators,
where they both have faces.

\begin{figure}
\begin{tikzpicture}[thick,draw=darkblue]
\draw (2,0) arc (180:0:1);
\draw (1,0) arc (180:360:1.5 and 1);
\draw[darkgreen] (0,0) -- (5,0);
\draw (2.75,-.4) node {$r^{-1}$}; \draw (3,.4) node {$r$};
\draw[very thick,draw=darkred,|-|] (1,0) -- (2,0) arc (180:120:1);
\draw (1.5,.1) node[above] {$x$};
\draw[very thick,draw=darkred,->] (1,0) -- (1.6,0);
\draw[very thick,draw=darkred,->] (2,0) arc (180:145:1);
\fill (2,0) circle (.08) node[inner sep=4pt,below] {$p$};
\draw (.5,1) node {$D$}; \draw (.5,-1) node {$E$};
\draw (130:1)++(3,0) node[left] {$x^{-1}$};
\end{tikzpicture}
\caption{A trivalent vertex shared by inverse relators creates a non-reduced
    word or relator.}
\label{f:trivalent} \end{figure}
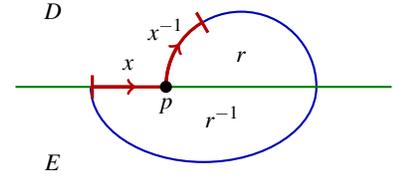

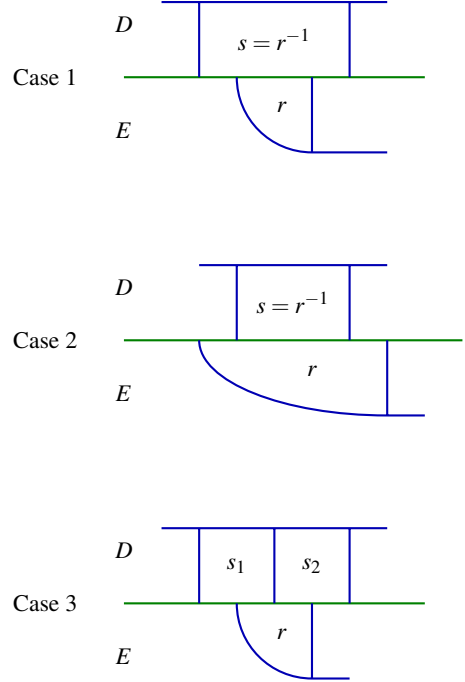
\begin{figure}
\begin{tikzpicture}[thick,draw=darkblue]
\draw[darkgreen] (0,0) -- (4,0);
\draw (.5,1) -- (3.5,1) (1,0) -- (1,1) (3,0) -- (3,1);
\draw (1.5,0) arc (180:270:1) -- (3.5,-1);
\draw (2.5,-1) -- (2.5,0);
\draw (0,.7) node {$D$} (0,-.7) node {$E$};
\draw (2.1,-.4) node {$r$} (2,.5) node {$s = r^{-1}$};
\draw (-.5,0) node[left] {Case 1};
\tikzset{yshift=-3.5cm}
\draw[darkgreen] (0,0) -- (4.5,0);
\draw (1,1) -- (3.5,1) (1.5,0) -- (1.5,1) (3,0) -- (3,1);
\draw (1,0) arc (180:270:2.5 and 1) -- (4,-1) (3.5,-1) -- (3.5,0);
\draw (0,.7) node {$D$} (0,-.7) node {$E$};
\draw (2.25,.5) node {$s = r^{-1}$} (2.5,-.4) node {$r$};
\draw (-.5,0) node[left] {Case 2};
\tikzset{yshift=-3.5cm}
\draw[darkgreen] (0,0) -- (4,0);
\draw (.5,1) -- (3.5,1) (1,0) -- (1,1) (2,0) -- (2,1) (3,0) -- (3,1);
\draw (1.5,0) arc (180:270:1) -- (3,-1) (2.5,-1) -- (2.5,0);
\draw (0,.7) node {$D$} (0,-.7) node {$E$};
\draw (1.5,.5) node {$s_1$} (2.5,.5) node {$s_2$} (2.1,-.4) node {$r$};
\draw (-.5,0) node[left] {Case 3};
\end{tikzpicture}
\caption{Three hypothetical configurations for $D$
    and $E$ at the first face of $E$.}
\label{f:footmatch} \end{figure}

We will use a sublemma which is illustrated in \Fig{f:trivalent}: Under our
hypotheses, $D$ and $E$ cannot have a pair of faces with inverse relators
that share a trivalent vertex $p$.  Because, in such a situation, either
$u$ or $w$ is not reduced, or there is a third relator that meets $p$
which is not reduced.

Let $r$ be the relator of the first face of $E$.  By assumption, $r$ meets
the disk of $D$ in some positive length.  Moreover, if the segment $x$
of $r$ along $v$ contains exactly one vertex of $D$, then it cannot be
the last vertex.  Working in the diagram $D \cup E$, we consider separate
cases and we make repeated use of \Lem{l:lengths} and $C'(1/6)$:
\begin{enumerate}
\item If $x$ has no vertices of $D$, then it is part of a single face in
$D$ with relator $s$.  In this case $\abs{x} > \abs{r}/3$, so $s = r^{-1}$.
\item If $x$ has two consecutive vertices of $D$ (and possibly others), then
they are the endpoints of a segment $y$ of a face in $D$ with relator $s$.
In this case $\abs{y} > \abs{s}/6$, so again $s = r^{-1}$.
\item If $x$ has exactly one vertex in the middle of $D$, then it divides
into two segments $x_1$ and $x_2$ that lie along two faces with relators
$s_1$ and $s_2$.  In this case $\abs{x} = \abs{x_1} + \abs{x_2} > \abs{r}/3$,
so at least one of $s_1$ and $s_2$ equals $r^{-1}$.
\end{enumerate}
The three cases are illustrated in \Fig{f:footmatch}.  Contrary to the
figure, inverse relator faces of $D$ and $E$ cannot share a trivalent
vertex, because that would create a non-reduced word or relator as in
\Fig{f:trivalent}.  The only allowed case is that the first two vertices
of $E$ are consecutive vertices of $D$ and the corresponding faces match
exactly.

We proceed by induction from left to right to show that the faces of $D$ and
$E$ continue to match until one of them runs out of faces.  The inductive
assumption is that the previous faces end at the same vertex along the
word $v$.  Thus the current face $E$ with relator $r$ starts at the same
position as the current face of $D$ with relator $s$.  One of them must end
no later than the other, so \Lem{l:lengths} and $C'(1/6)$ again conspire
to yield $s = r^{-1}$.  Again they cannot share a trivalent vertex, so
the two relators also end at the same position.

Having shown that faces of $D$ and $E$ match along $v$, we must also
show that the same faces have identical internal edges within $D$
and $E$.  In general, if two relators $r$ and $s$ in a legal diagram $D$
(one with reduced relators and a reduced outside) share an edge, then the
value of that edge is forced to be a maximal common substring of $r^{-1}$
and $s$.  Thus, when internal edges of $D$ and $E$ share a vertex in the
word $v$, they have to match each other.

This establishes a maximal common disk $K$ inside both $D$ and $E^{-1}$.
When they are merged to form $D \cup_K E^{-1}$, its new boundary $z$
consists of part of $u$ and $w$ spliced together along a positive length.
This shows that $z$ is reduced, which we need to know that $D \cup_K E^{-1}$
is a legal diagram.
\end{proof}

\begin{remark} Although we will not need this for our purposes, the proof
of \Lem{l:merge} can be extended to show that $z$ is Dehn reduced.
\end{remark}

\begin{lemma} If $u \sim_D v \sim_E w$ are two equality diagrams
between shortest words, then possibly after some twists, they
can be merged into an equality diagram $x \sim_C y$
between two shortest words $x$ and $y$.
\label{l:smerge} \end{lemma}

\begin{proof} We can assume that both $D$ and $E$ are non-trivial.
Arguing by induction, we will reduce any example $u \sim_D v \sim_E w$
to another example $u' \sim_{D'} v' \sim_{E'} w'$ with fewer total disks.
Let $K$ and $L$ be the first (or leftmost) thin disks of $D$ and $E$,
respectively.  Assume without loss of generality that the right endpoint of
$K$ is at least as far to the left as that of $L$.  If $K$ does not meet
$L$ in a positive length, then we can let $u' \sim_{D'} v' \sim_{E'} w'$
be everything to the right of $K$ and attach $K$ on the left after $D'$
and $E'$ are merged.  Otherwise we can form $L'$ by merging $K$ and $L$
as in \Lem{l:merge}.  $L'$ is necessarily a single thin disk, and we can
define $D'$ by deleting $K$ and $E'$ by replacing $L$ with $L'$.  Also,
since $u$, $v$, and $w$ are all shortest words, all boundary transversals of
$D$ and $E$ (and their subwords) are also shortest words, so \Lem{l:merge}
necessarily produces another equality diagram between shortest words.
\end{proof}

\begin{lemma} If $v \sim w$ are shortest words, then they only have one
reduced equality diagram, necessarily a thin diagram.
\label{l:unique} \end{lemma}

\begin{figure}
\begin{tikzpicture}[thick,draw=darkblue,scale=1.3]
\draw (1,0) arc (180:0:1) (2,0) -- (2,1);
\draw (1,0) arc (180:270:1) -- (3,-1) arc (270:360:1)
    (2,-1) -- (2,0);
\draw[darkgreen] (0,0) -- (5,0);
\draw[darkred,very thick,-|] (3,-1) -- (3,0) arc (0:60:1);
\draw[darkred,very thick,->] (3,-1) -- (3,-.4);
\draw[darkred,very thick,|-<] (3,0) arc (0:30:1);
\draw[darkred,very thick,->] (3,-1) arc (270:305:1);
\draw[darkred,very thick,-|] (3,-1) arc (270:330:1);
\draw (2.5,-.5) node {$r^{-1}$}; \draw (2.4,.4) node {$r$};
\draw (1,1) node {$K$} (1,-1) node {$L$};
\draw (30:1.35)++(2,0) node {$x$};
\draw (3.1,-.5) node[right] {$x$};
\draw (300:1.35)++(3,0) node {$x$};
\fill (0,0) circle ({.08/1.3});
\draw (5,0) node[right] {$v$};
\draw (2,1) node[above] {$u$}; \draw (2,-1) node[below] {$u$};
\end{tikzpicture}
\caption{If $L$ has more faces than $K$, then the first extra face of $L$
    is not reduced.}
\label{f:notreduced} \end{figure}
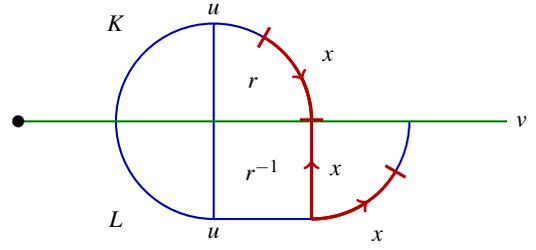

\begin{figure}
\begin{tikzpicture}[thick,draw=darkblue,scale=1.4]
\draw[darkgreen] (0,0) -- (4.5,0);
\draw (1,0) arc (180:0:1) arc (180:90:1.5 and 1);
\draw (2,1) -- (2,-1);
\draw (1,0) arc (180:270:1) arc (270:360:2 and 1) arc (180:240:1);
\draw[-|,darkred,very thick] (3,0) arc (0:60:1);
\draw[-<,darkred,very thick] (3,0) arc (0:30:1);
\draw[-|,darkred,very thick] (3,0) -- (4,0);
\draw[->,darkred,very thick] (3,0) -- (3.6,0);
\draw (45:1.25)++(2,0) node {$x$};
\draw (3.5,.1) node[above] {$x$};
\draw[black,|-|] (2,0)++(177:1.2) arc (177:63:1.2);
\draw[black,|-|] (2,0)++(183:1.2) arc (183:270:1.2) arc (270:336:2.2 and 1.2);
\draw[black,|-|] (1.1,-.2) -- (3,-.2);
\fill (0,0) circle ({.08/1.4});
\draw (2,0)++(120:1.2) node[above left] {$\ell$};
\draw (1.6,-.2) node[below] {$h$} (2,-1.2) node[below] {$\ell$};
\draw (2.5,-.6) node {$r^{-1}$} (2.4,.4) node {$r$};
\draw (4.5,1) node[right] {$u$} (4.5,0) node[right] {$v$};
\draw (4.5,-.866) node[right] {$u$};
\draw (.8,1) node {$K$} (.8,-1) node {$L$};
\end{tikzpicture}
\caption{If $K$ and $L$ have the same faces but $L$ ends later, then $u$
    and $v$ cannot both be shortest.}
\label{f:shortcut} \end{figure}
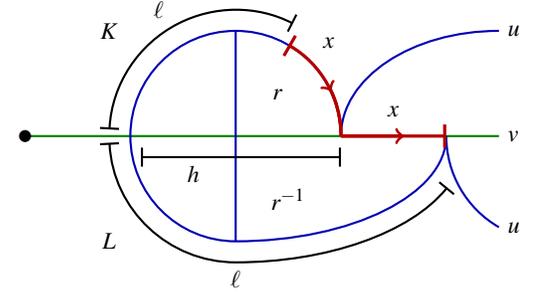

\begin{proof} \Lem{l:thin} says that any equality diagram for $v \sim w$
is thin, so we only need to show that any two such diagrams $D$ and $E$
are the same.  We first consider the trivial case that $v = w$.  In this
case, $D$ and $E$ cannot have any faces since relators are reduced.
Assuming the non-trivial case $v \ne w$, both $D$ and $E$ must have faces.

Let $K$ and $L$ be the first thin disks of $D$ and $E$, respectively.
They must start at the same position along both $v$ and $w$, namely the point
at which $v$ and $w$ first differ.  We claim that they must also end at the
same position and that $L = K^{-1}$. By \Lem{l:merge}, we can merge them,
but that by itself only says that $L$ and $K$ have inverse relators
until one of them runs out of faces.

Suppose that $K$ and $L$ do not have the same number of faces, without
loss of generality that $L$ has more.  In this case, the first extra
face of $L$ is not reduced, as shown in \Fig{f:notreduced}.

Suppose that $K$ and $L$ have the same number of faces, but the last faces
end in different positions, without loss of generality that $L$ ends later.
(In order for $u$ to be reduced, this requires that a new disk in $D$
starts immediately after $K$.)  This case would be possible if $u$ and
$v$ were merely Dehn-reduced; we claim that it is not possible if they
are geodesics.  Using the labelling in \Fig{f:shortcut} and using that $u$
and $v$ are both shortest words, we can calculate the length to the end of
$K$ in two ways and the length to the end of $L$ in three different ways,
yielding the relations
\[ h = \ell+\abs{x} \qquad \ell = h+\abs{x}. \]
The equations tell us that the segment $x$ has zero length, contradicting
the current hypothesis that $L$ ends later than $K$.

If $K$ and $L$ have the same faces and begin and end at the same positions,
then \Lem{l:merge} tells us that $L = K^{-1}$.  Given that the first
thin disks of $D$ and $E$ are the same, we can let $D'$ and $E'$ be the
remainder with this disk deleted and proceed by induction.
\end{proof}

\begin{proof}[Proof of \Thm{th:canonical}] We first establish the existence
of a thin diagram $D$ between shortest words that contains every shortest
word as a transversal.  Among all thin diagrams between shortest words
for some element $g \in G$, let $v \sim_D w$ be one that has as many
faces as possible.  We claim that this $D$ contains every shortest word
for $g$.  Suppose to the contrary that $D$ does not contain some shortest
word $u$.  By \Lem{l:unique}, there is a unique thin diagram $u \sim_E v$.
By \Lem{l:smerge}, we can merge $D$ and $E$, possibly after some twists.
By hypothesis, this merge $D'$ of $D$ and $E$ cannot enlarge $D$; it must
instead be a twist of $D$.

To show that $D$ is unique up to twists among equality diagrams with the
stated properties, let $x \sim_E y$ be an equality diagram between some
pair of shortest words for the group element $g$.  Since $D$ contains $x$
and $y$ as transversals, it contains an equality diagram $x' \sim_{E'} y'$
in which $x$ and $y$ are boundary transversals.  By \Lem{l:unique},
$E'$ must be a twist of $E$.  Note that a word $z$ which arises as a
transversal of a thin diagram $D$ determines its path in $D$, since the
relator faces of $D$ are reduced.  If $E$ contains every shortest word
for $g$, then so does $E'$.  In particular $E' \subseteq D$ contains all
of the boundary transversals of $D$, so we obtain $E' = D$ and that $D$
and $E$ differ only by twists, as desired.

Finally, if $G$ is finitely presented, then we argue that a canonical
diagram $D$ can be computed in polynomial time from any word $w$ that
represents $g$.  To see this, we first Dehn reduce $w$ using arbitrary
choices in Dehn's non-deterministic greedy algorithm.  Assuming that $w$
is Dehn reduced, we can search for thin disks that shorten it further.
We claim that we can find all possible thin disks that attach to $w$ in
polynomial time.  This claim follows from \Lem{l:lengths} together with the
$C'(1/6)$ condition.  The lemma implies that we can calculate the relator
$r$ of any face of a thin disk $K$ from its position alone, given that only
one $r$ can meet $w$ in more than $1/6$ of its length.  In particular,
once we choose both endpoints of $K$, then each face and its relator are
uniquely determined in sequence from left to right and can be calculated
by searching through the relators of $G$.  With the ability to find all
thin disks that attach to $w$, we can apply one that shortens it.  If
no such thin disk exists, then \Lem{l:thin} tells us that $w$ is shortest.

Assume now that $w$ is a shortest word for $g$.  We can grow the canonical
diagram $v \sim_D w$ by looking for thin disks that attach to $w$ and do
not lengthen it.  We can recognize when such a disk $K$ already embeds
in $D$.  If some thin disk $K$ does not embed in $D$, then we can assume
that $K$ is the shortest such disk.  In this case we can attach $K$ to $D$
after twisting at most two consecutive thin disks in $D$ itself.  When $D$
is maximal, it has strictly less than $\abs{w}$ faces, so it can be grown
to its full size in polynomial time.
\end{proof}

As promised, \Thm{th:shortlex} is a quick corollary of \Thm{th:canonical}.

\begin{proof}[Proof of \Thm{th:shortlex}] Since \Thm{th:canonical} includes
a polynomial-time algorithm to calculate a thin diagram $D$ that contains
every shortest word for $[w] = g \in G$ as a transversal, it suffices to
search in $D$ for the shortlex transversal.  This is achieved by an abstract
version of Dijkstra's algorithm described by Sobrinho \cite{Sobrinho:qos}.
The standard version of Dijkstra's algorithm finds the shortest path
between two nodes in a graph whose edges have assigned lengths.  The same
algorithm can also find the minimal path if the edges are assigned elements
of an ordered monoid $S$, provided that the identity $1 \in S$ is the least
element of $S$.  In particular this applies to the semigroups of words $W_A$
in the shortlex order.  The key property of $W_A$ that enables Dijkstra's
algorithm is that concatenation of words preserves the shortlex ordering.
\end{proof}

\begin{remark} We think that \Thm{th:canonical} can be extended to establish
a thin diagram between some pair of Dehn-reduced words for each $g \in G$
that contains all Dehn-reduced words as transversals, and that such a diagram
is unique up to reversible twists.   The complication arises that a twist of
a thin diagram $v \sim_D w$ can be effectively irreversible, even when $v$
and $w$ are Dehn reduced.  If $D$ has two touching disks $K_1$ and $K_2$
and we twist only $K_2$, then the result is a diagram $v' \sim_{D'} w'$
which might be illegal because $v'$ and $w'$ might not be reduced.  We can
reduce $v'$ and $w'$ and glue part of $K_1$ and $K_2^{-1}$ in $D'$ to make
a thin diagram $v'' \sim_{D''} w''$; the transformation of $D$ to $D''$
can be called a \emph{locking twist}.  If $K_1$ and $K_2$ are connected
by an edge labelled by a short enough word instead of touching outright,
then $v' \sim_{D'} w'$ is legal, but $v'$ or $w'$ might not be Dehn reduced.
In this case, we can connect $K_1$ and $K_2$ in $D'$ with a Dehn-reducing
face to make another type of locking twist.
\end{remark}

\begin{example} Let $G$ be generated by
\[ a_1,a_2,a_3,a_4,b_1,b_2,b_3,b_4,c_1,c_2,c_3 \]
with relators
\[ a_1 c_1 b_1^{-1}, \quad a_2 c_2 b_2^{-1} c_1, \quad
    a_3 c_3 b_3^{-1} c_2, \quad a_4 b_4^{-1} c_3. \]
Although these relators are not $C'(1/6)$, small cancellation theory
generalizes to groups whose generators have unequal formal lengths.
(We can also substitute each letter by an arbitrary product of distinct
letters from a larger alphabet.)  We stipulate that
\[ \abs{a_i} = \abs{b_j} \gg \abs{c_1} = \abs{c_3} > \abs{c_2} \]
for all $i$ and $j$.  With sufficiently disparate lengths, this presented
group $G$ is $C'(\lambda)$ with $\lambda$ arbitrarily small.  Now let
\[ v = a_1 a_2 c_2 b_3 b_4 \qquad w = b_1 b_2 c_2^{-1} a_3 a_4. \]
The reader can check that:
\begin{enumerate}
\item $v$ and $w$ are Dehn reduced.
\item $v \sim w$ by two different thin diagrams $D$ and $E$.
\item $D$ and $E$ become the same diagram $v'' \sim_{D''} w''$ after a
locking twist of each one.
\end{enumerate}
The construction is a counterexample to several suppositions:
\begin{enumerate}
\item The words $v$ and $w$ are Dehn reduced, but they are not shortest.
\item A twist of a thin diagram between Dehn-reduced words is not
    always reversible.
\item In contrast to \Lem{l:unique}, a thin diagram between
Dehn-reduced words is not always unique.
\end{enumerate}
\end{example}

\subsection{The proof of \Thm{th:free}}
\label{ss:freeproof}

After the material in \Sec{ss:small}, the proof of \Thm{th:free} is
relatively short, but it is still non-trivial.

\begin{proof}[Proof of \Thm{th:free}] We first establish the result when
the ambient group is free group $F_{14}$ with exactly 14 generators (see
also \Fig{f:tomcat}) over the alphabet
\[ A = \{a_1,b_1,a_2,b_2,\dots,a_7,b_7\}. \]
As before, let $d \in \NP$ be a decision problem which
is given by a predicate $z \in \ccP$, so that
\[ d(x) = \exists y, z(x,y) \]
with $\abs{y} = \poly(\abs{x})$.  We assume that all certificates $y$
are binary strings of the same length $m$.  Given two symbols $a$ and $b$,
we can convert $y$ to a word $y(a,b) \in \{a,b\}^*$ by replacing each $0$
with $a$ and each $1$ with $b$.  We fix the decision input $x$, and
for each accepted certificate $y$, we make a relator
\begin{eq}{e:relator} r_y = y(a_1,b_1)y(a_2,b_2)\cdots y(a_7,b_7).
\end{eq}
Note that each $r_y$ is cyclically reduced.  Let $R = \{r_y\}$ be the set
of these relators, and let $N_R \normaleq F_A \cong F_{14}$ be the
subgroup normally generated by $R$.  The finitely presented group $G =
\braket{A \mid R}$ is clearly $C'(1/6)$, so that all of the results
of \Sec{ss:small} apply.

\begin{figure}
\begin{center}
\frame{\includegraphics[width=3in]{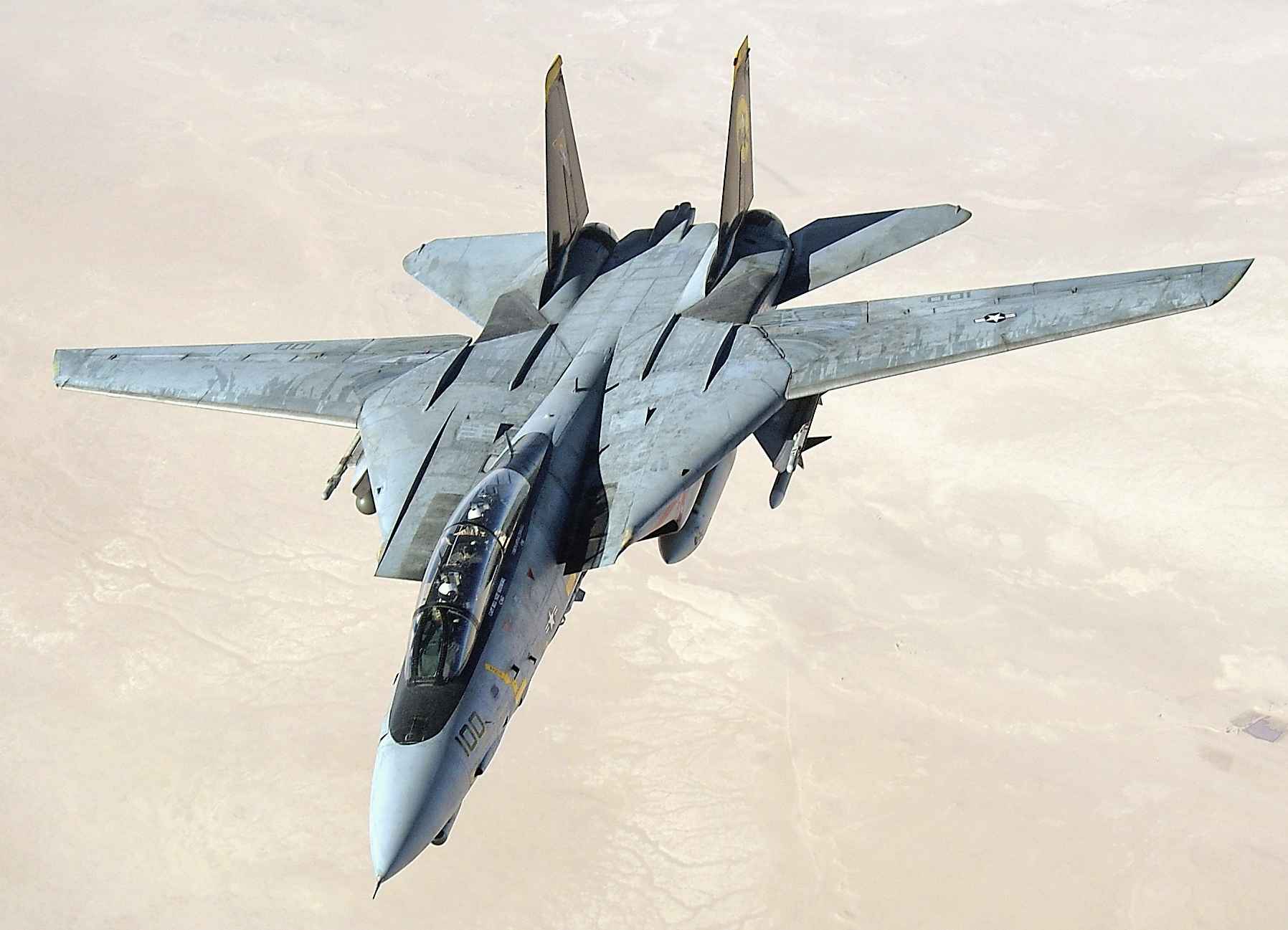}}
\end{center}
\caption{The U.S. Navy F-14 Tomcat fighter jet, not to be confused with the
    free group $F_{14}$ on 14 generators.}
\label{f:tomcat} \end{figure}

The subgroup $N_R$ is non-trivial if and only if $d(x) = \yes$, and it has
polynomial-length elements when it is non-trivial.  Thus, we will achieve our
reduction if we can efficiently compute a hiding function $f$ for $N_R$.
Following \Sec{ss:small}, let $f(w)$ be the shortlex equivalent of $w$
in the coset $wN_R$.  We can compute $f(w)$ using a refinement of the
algorithm in the proof of \Thm{th:canonical}, followed by the algorithm
in \Thm{th:shortlex}.  To review, the algorithm in \Thm{th:canonical}
builds a thin diagram $D$ that contains all shortest words $v \sim w$
as transversals.   We adapt this algorithm to the circumstance that we are
not directly given the relator set $R$, but instead we can ask questions
about $R$ via the predicate $z$.

The algorithm in \Thm{th:canonical} begins by Dehn-reducing the word $w$.
Each step of Dehn's algorithm looks for a relator $r$ to attach to the
evolving $w$ along a substring $s$ with $\abs{s} > \abs{r}/2 = 7m/2$.
The string $s$ is thus long enough that it can only match at most one
position in one relator $r = r_y$ given by \equ{e:relator}.  It is likewise
easy to calculate both $r$ and $y$ from $s$, and we can evaluate the
predicate $q(x,y)$ to tell if $y$ exists and therefore $r$ exists in $R$.
The algorithm then looks for thin disks that further shorten $w$ until it
is shortest.  By \Lem{l:lengths} and the $C'(1/6)$ condition, the relator
$r$ of each face of such a thin disk attaches to the further evolving $w$
along a substring $s$ with $\abs{s} > \abs{r}/6 = 7m/6$.  We can again
calculate $r = r_y$ and $y$ from $s$, and we can again evaluate $q(x,y)$
to determine if $r$ and $y$ exist.

This part of the algorithm concludes by extending $w$ to a thin diagram
$D$ that contains all shortest words as transversals.   This process is
also based on searches for thin disks that attach to the evolving $D$.
As before, we can determine $r = r_y$ and $y$ from the existing structure,
and use the predicate to know whether $r$ is a relator in $R$.

Having constructed the thin diagram $D$ guaranteed by \Thm{th:canonical},
the generalized Dijkstra's algorithm in \Thm{th:shortlex} works the same
as before.

Finally, we want to generalize this argument to $F_k$ for all $k \ge 2$
rather than just $k=14$. We can first reduce any $F_k$ to the case of
$F_2$ by adding all but two of the generators as relators.  Then we want
to imitate the solution that we obtained in $F_{14}$ in $F_2$ instead.
For this purpose, say that $F_2$ is generated by $B = \{a,b\}$.  We can
encode the 14 generators used in \equ{e:relator} in these 18-letter subwords:
\begin{eq}{e:subs} a_k = ab^kab^{15-k}a \qquad b_k = ab^{k+7}ab^{8-k}a.
\end{eq}
With this substitution, a cyclic relator $r_y$ defined as in \equ{e:relator}
has length $126m$, and the cuts between the subwords are distinguished
as the only positions where the letter $a$ is repeated.  It is easy to
check that any common substring $s$ between two such relators is strictly
contained in a string of $m+1$ subwords, and thus has length $\abs{s}
< 18(m+1)$.  Provided that $m \ge 6$ (which we can ensure by padding the
certificates), the presented group $G = \braket{B \mid R}$ is a $C'(1/6)$
group, and the algorithm to compute the shortlex equivalent $f(w)$ of each
$w$ still works as before.
\end{proof}

\subsection{Variations and open problems}

We briefly consider some variations of the hypotheses of \Thm{th:free}.

Without going into details, there is a simplified version of \Thm{th:free}
and its proof when the hidden subgroup $H \subseteq F_k$ is non-normal
and finitely generated rather than finitely normally generated.  If $H
\subseteq F_k$ is any finitely generated subgroup with a given list of
generators, then there is a polynomial-time algorithm to compute canonical
representatives of $F_k/H$ which is based on building the Schreier graph
of $F_k/H$ \cite[Sec.~2.3.1]{GT:topo}.  Following \equ{e:relator}, the
Schreier graph can still be navigated efficiently in $F_6$ given only
oracle access to the generators of $H$, if those generators all have the form
\[ h_y = y(a_1,b_1)y(a_2,b_2)y(a_3,b_3) \]
for certificates $y$ of an $\NP$ predicate $z(x,y)$.   The construction
can also be modified for any non-abelian $F_k$ using the substitution
\eqref{e:subs}.

Another variation is to promise that the hidden subgroup $H \subseteq F_k$
has finite index $[F_k:H] < \infty$.   In this case, $H$ necessarily has
$[F_k:H](k-1)+1$ generators.  If the output is a full description of $H$
and has length $\poly(n)$ for some parameter $n$, then the index is also
$\poly(n)$.  However, in this case, HSP is in classical $\ccP$ using the
Schreier--Sims algorithm \cite{Sims:study}.  On the other hand, if the
question is whether $H$ is a non-trivial element of length $\poly(n)$
and $[F_k:H]$ can have exponential size, then we conjecture that HSEP is
again $\NP$-hard.

Suppose that the hidden subgroup $H \normaleq F_k$ is normal and has finite
index, and thus that $G = F_k/H$ is a finite group of order $[F_k:H]$.
Assuming a bound on $\abs{G}$ which is exponential in some parameter $n$,
Pyber \cite{Pyber:enum} showed that there are only $\exp(\poly(n))$ choices
for $G$, and thus a polynomial-length description of $H$ of some kind.
Unfortunately, for instance if $G$ is cyclic, the generators of $H$ might
still have exponential length; worse, it might only be feasible to reach
$\poly(n)$ elements of $G$ and thus only $\poly(n)$ values of the hiding
function $f$.  However, if we can efficiently evaluate the hiding function
$f(w)$ for a compressed word given by an algebraic circuit, then it is
feasible to reach every element of $G$ as well as a full set of generators
of $H$.   We optimistically conjecture that there is an efficient algorithm
for NHSEP in this case, or at least an algorithm with polynomial query
complexity.  If it exists, it would generalize the algorithm of Hallgren,
Russell, and Ta-Shma for NHSP in finite groups \cite{HRT:normal}.

Finally, we can go back to NHSP in $F_k$ when $H$ can have infinite index,
but where elements of $F_k$ are encoded as canonical compressed words.
Setting up the question in this form uses the striking result, mentioned in
the remark at the end of \Sec{ss:explicit}, that there is an efficient
way to compute a canonical compressed word in $F_k$ from any compressed
word \cite{Lohrey:groups}.  We conjecture that NHSEP in this case is also
$\NP$-complete.

\section{\Cor{c:expquery}: Exponential lower bounds}
\label{s:expquery}

In this section, we briefly argue \Cor{c:expquery} as a corollary of
the proofs of Theorems~\ref{th:rational} and \ref{th:free}.  The corollary
follows quickly from the following observations.
\begin{enumerate}
\item As explained in \Sec{ss:complex}, the reductions in the proofs of
Theorems~\ref{th:rational} and \ref{th:free} both relativize.  Both proofs
construct a hiding function $f$ from a predicate $z$, and relativization
means that $z$ can use an oracle or simply be one.

\item Since \Cor{c:expquery} posits a bound on quantum query complexity,
we can assume a quantum algorithm to compute the hiding function $f$
in \Sec{ss:ratproof}, so we can use Shor's algorithm to factor integers.
(Even if we used a slow algorithm to factor integers, the hiding function
$f$ in \Sec{s:rational} would still have low query complexity, and would
thus still establish the reduction to the query complexity of unstructured
search in step 5.)

\item In \Sec{ss:ratproof}, assuming a factoring algorithm, the generators
of the hidden subgroup $H \subseteq \Z$ are 1 and fractions of the form $1/p$,
where each $p$ is a prime number whose bit complexity is linear in the
length of a certificate $y$ that satisfies the predicate $z$.  The bit
string $y$ is the left third of the binary digits of $p$ after the leading 1.
To review, \Thm{th:ingham} implies that this map from primes to bit strings
is surjective for sufficiently long bit strings.

\item In \Sec{ss:freeproof}, each normal generator of the hidden subgroup $H$
is a word $r_y$ which is a linear expansion of the corresponding certificate
$y$ by \equ{e:relator}.

\item A quantum search for a unique $n$-bit certificate in an oracle has
quantum query complexity $\Theta(2^{n/2})$ because Grover's algorithm
exists and is optimal \cite{Grover:fast,BBBV:strengths,Zalka:optimal}.
(Or the bound could be expressed as $\tTheta(2^{n/2})$ given polynomial
query cost.)
\end{enumerate}

Instead of a lower bound on query complexity, we can also appeal to the
exponential time hypothesis (ETH) \cite{IP:ksat} to obtain a conditional
variation of \Cor{c:expquery}.  If ETH holds, then HSEP in $\Q/\Z$ and
NHSEP in a non-abelian free group both require exponential computation
time.  The standard version of ETH is a lower bound on computation time
for classical algorithms for 3-SAT, but it seems reasonable to conjecture
that it also holds for quantum algorithms.

\section{Proof of \Thm{th:svp}: Solving SVP with HSP}
\label{s:svp}

Our proof of \Thm{th:svp} requires stricter hypotheses than either
\Thm{th:rational} or \ref{th:free}.  In those two cases, we obtain a
polynomial-time reduction from any particular problem $d \in \NP$ (or even
in $\NP^h$, where $h$ is an oracle) to HSEP for $\Q/\Z$ or $F_k$.  Thus,
HSEP and therefore HSP is $\NP$-hard when $d$ is $\NP$-hard, regardless
of whether HSP has an efficient algorithm for other hiding functions.

In \Thm{th:svp}, we assume a polynomial-time algorithm for HSP on $\Z^k$
that correctly handles realistic hiding functions provided by a reduction,
and that also correctly handles artificial hiding functions that are
provided by an oracle.  If the HSP algorithm can correctly handle all
hiding functions rather than just realistic ones, then it can be used to
solve the short vector problem with polynomial parameters.

Before proving \Thm{th:svp}, we first state it more precisely.  As mentioned
in \Sec{ss:complex}, given a vector $\vx \in \Z^k$ as input to the hiding
function, we assume an oracle time cost which is polynomial in $\norm{\vx}_1$
and $k$.   We assume a two-sided bound, both $\poly(\norm{\vx}_1,k)$
and $\Poly(\norm{\vx}_1,k)$.  Recall also from \Sec{ss:hidden} that the
performance of an algorithm for HSP is rated by the time that it takes
to confirm a generating set of the subgroup $H$.  In our proof, $H =
\braket{\vv}$ will be a rank 1 lattice, or discrete line, generated by a
vector $\vv \in \Z^k$.  Thus, ``polynomial time" amounts to pseudo-polynomial
time in the entries of $\vv$ and polynomial time in $k$.  The existence
of a pseudo-polynomial-time algorithm is the hypothesis of \Thm{th:svp}.

\begin{figure}
\begin{tikzpicture}[thick,scale=.375]
\begin{scope}
\clip (0,0) circle (6.5);
\draw[darkgreen,dashed,fill=lightgreen] (0,0) circle (2.5);
\foreach \y in {-7,-6,...,7} {
    \foreach \x in {-1,0,1} {
    \draw[darkred] ({Mod(7*\x-\y+10,21)-10},\y) node {\small \ding{164}};
    \foreach \a/\d in {1/34,2/37,3/40,4/93,5/52,6/74} {
        \draw[darkgray] ({Mod(7*\x+\a-\y+10,21)-10},\y)
        node {\small \ding{\d}}; } } }
\end{scope}
\draw (5.6,5.6) node {$\Z^2$};
\draw (-8,0) node {$f(\vx)$};
\draw (6,-5.6) node {\textcolor{darkred}{\ding{164}}\ $= L$};
\tikzset{shift={(0,-16)}};
\begin{scope}
\clip (0,0) circle (6.5);
\draw[darkgreen,dashed,fill=lightgreen] (0,0) circle (2.5);
\foreach \y in {-10,-9,...,10} {
    \draw[darkred] ({-\y},\y) node {\small \ding{164}};
    \foreach \x/\d in {
        -9/61,-8/126,-7/171,-6/42,-5/168,-4/56,-3/93,-2/52,-1/74,
        1/34,2/37,3/40,4/89,5/108,6/72,7/118,8/252,9/115} {
        \draw[darkgray] ({\x-\y},\y) node {\small \ding{\d}}; } }
\end{scope}
\draw (5.6,5.6) node {$\Z^2$};
\draw (-8,0) node {$g(\vx)$};
\draw (6,-5.6) node {\textcolor{darkred}{\ding{164}}\ $= H$};
\end{tikzpicture}
\caption{The hiding functions $f(\vx)$ and $g(\vx)$ agree within a search
    radius.  An HSP algorithm would find $H$ and thus solve SVP in $L$.}
\label{f:dingbats} \end{figure}

We also state the short vector problem $\uSVP$ more precisely.  Let $L
\subseteq \Z^k$ be an integer lattice given by a lattice basis, let $0
< a \le b$ be two integer constants, and let $p \in [1,\infty]$ be a
norm parameter.  Then the \emph{$p$-norm unique short vector problem}
$\uSVP^p_{a,b}$ is the question of finding the shortest non-zero vector
$\vv \in L$, given the promise that $\norm{\vv}_p \le a$, and that
$\norm{\vx}_p > b$ for any $\vx \in L$ which is not proportional to $\vv$
\cite{AD:unique}.  The conclusion of \Thm{th:svp} is a quantum algorithm
for the problem $\uSVP^p_{a,\Poly(a,k)}$.  Given that the promise gap can
be any polynomial, the value of $p$ is not important and is omitted in
the statement of the theorem and in the proof.

\begin{proof}[Proof of \Thm{th:svp}] Let $B_{r,1}(\vec0)$ denote the
1-norm ball in $\Z^k$ of radius $r$, centered at the origin.  Suppose that
there is a (quantum) algorithm $\sZ$ for HSEP in $\Z^k$ with the specified
properties.  By hypothesis, let $f:\Z^k \to X$ be a function that hides a
discrete line $H = \braket{\vv}$.  If $\vv \in B_{a,1}(\vec0)$, then $\sZ$
promises to calculate whether $H$ is non-trivial in polynomial time in the
query cost for vectors of length $a$.  Since the query cost is $\poly(a,k)$,
Algorithm $\sZ$ thus also runs in time $\poly(a,k)$.  The total running
time of $\sZ$ includes the fiat time cost of each evaluation $f(\vx)$.
These evaluations are not restricted to the ball of radius $a$, but the
evaluation at any vector $\vx$ has a fiat $\Poly(\norm{\vx}_1,k)$ time cost.
Thus, all of the evaluations of $f$ are in a ball $B_{r,1}(\vec0)$ with
$r = \poly(a,k)$.

Let $L \subseteq \Z^k$ be a lattice which is the input for the short
vector problem $\uSVP^1_{a,2r}$.   Our goal is to construct two hiding
functions $f$ and $g$ with the following properties:
\begin{enumerate}
\item The function $f$ hides $L$ and is computable in polynomial
time.
\item The function $g$ hides the discrete line $H$ which is generated
by $L \cap B_{a,1}(\vec0)$, and need not be computable in polynomial time.
Instead, it is provided by an oracle.
\item The functions $f$ and $g$ agree on $B_{r,1}(\vec0)$.
\end{enumerate}
Given these three properties (which are illustrated in \Fig{f:dingbats}),
Algorithm $\sZ$ must give the same answer for the hiding functions $f$
and $g$ if its time limit restricts it to queries in $B_{r,1}(\vec0)$.
Given the $\poly(a,k)$ running time that is sufficient for $\sZ$ to find
the lattice $H$, this answer must be correct for the hiding function $g$.
It won't be correct for the hiding function $f$ when $L \ne H$, because $\sZ$
won't have enough time to confirm that any vector is in $L \setminus H$.
Nonetheless, $\sZ$ computes a shortest non-zero vector $\vv \in L$, given
that $\vv$ and $-\vv$ are the two free generators of the discrete line $H$.

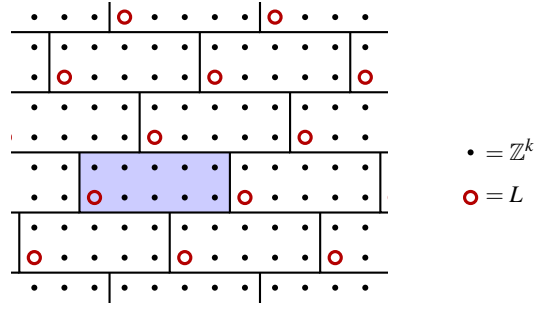
\begin{figure}
\begin{tikzpicture}[thick,scale=.4]
\fill (15,5) circle (.1);
\draw (15.2,5.1) node[right] {$= \Z^k$};
\draw[very thick,darkred] (15,3.5) circle (.2);
\draw (15.2,3.6) node[right] {$= L$};
\clip (-.25,0) rectangle (12.25,10);
\fill[lightblue] (2,3) rectangle (7,5);
\foreach \x/\y in {
    -2/-1,3/-1,8/-1,    -5/1,0/1,5/1,10/1,
    -3/3,2/3,7/3,12/3,  -1/5,4/5,9/5,
    -4/7,1/7,6/7,11/7,  -2/9,3/9,8/9} {
    \draw (\x,{\y+2}) -- (\x,\y) -- ({\x+5},\y);
    \draw[very thick,darkred] ({\x+.5},{\y+.5}) circle (.2);
    \foreach \a/\b in {1/0,2/0,3/0,4/0,0/1,1/1,2/1,3/1,4/1} {
        \fill ({\x+\a+.5},{\y+\b+.5}) circle (.1); } }
\end{tikzpicture}
\caption{The rectangular fundamental domain $B$ (shaded) and the
    brick tiling that arises from a basis matrix $M$ of a lattice $L
    \subseteq \Z^k$ in Hermite normal form.}
\label{f:brick} \end{figure}

To make the hiding function $f$, we use the same method of canonical
representatives described in \Sec{s:rational} and used again in \Sec{s:free}.
Let $M$ be a \emph{basis matrix} of $L \subseteq \Z^k$ (meaning a matrix
whose columns are a basis of $L$).  Kannan and Bachem \cite{KB:smith}
establish an algorithm in $\ccP$ to find the Hermite normal form of $M$.
Since the column Hermite normal form of $M$ is another basis matrix of $L$,
we can assume without loss of generality that $M$ is in the form.  In this
form, $M$ has an echelon structure with pivot entries $m_1,m_2,\ldots,m_k
> 0$.  One property of Hermite normal form is that the rectangular box
\[ B = [0,m_1)_\Z {\times} [0,m_2)_\Z \times \cdots
    {\times} [0,m_k)_\Z \subseteq \Z^k \]
is a set of canonical representatives of the cosets of $L$ in $\Z^k$.
As illustrated in \Fig{f:brick}, the $L$-translates of $B$ tile $\Z^k$ in
a brickwork pattern.  Using the echelon structure of $M$, there is also an
efficient algorithm (that works by induction on $k$) to put any $\vx \in
\Z^k$ in the form $\vy+\vz$ with $\vy \in B$ and $\vz \in L$.  In other
words, the canonical equivalent $\vy =f(\vx)$ can be computed quickly.

Finally, to make the hiding function $g$, we note that if $\vx,\vy \in
B_{r,1}(\vec0)$, then $\norm{\vx-\vy}_1 \le 2r$ by the triangle inequality.
The promised gap property of the lattice $L$ tells us that $f(\vx) =
f(\vy)$ in this case if and only if $\vx-\vy \in H$.  If we interpret
$f$ as a function only on $B_{r,1}(\vec0)$, it therefore has a $H$-periodic
extension $g$ to $\Z^k$.  If $L \ne H$, then this extension requires fresh
values of $g$ that need not have any particular meaning, for one reason
because Algorithm $\sZ$ would never access them anyway.  The extension of $g$
might not be computable in polynomial time either.  (One way to calculate
$g$ is to first find the shortest vector in $L$, which is backwards from
using Algorithm $\sZ$ to solve this problem.)  Nonetheless, $g$ satisfies
the promise to Algorithm $\sZ$, and serves to show that $\sZ$ would find
the shortest vector in $L$ if given the hiding function $f$ instead.
\end{proof}

\subsection{Reduction to a lattice learning problem}
\label{ss:learning}

We have set up \Thm{th:svp} as a problem about a lattice $L \subseteq
\Z^k$ whose intersection with an allowed ball $B_{r,1}(\vec0)$ is
a discrete line $H \subseteq \Z^k$.   We conjecture that there are
similar but more complicated reductions even when $L \cap B_{r,1}(\vec0)$
generates all of $L$.   If so, then our specific proof of \Thm{th:svp}
cannot work as described, because the two hiding functions $f$ and $g$
must agree.  However, if we encrypt the values of the hiding function $f$
(or even provably obfuscate them with an oracle), then we would expect
an algorithm for HSP to learn something about repeated values of $f$ on
$B_{r,1}(\vec0)$. We would therefore expect such an algorithm to detect
vectors in $L$ of length at most $2r$.

In contrast with \Thm{th:svp}, which is a hardness reduction from the
lattice problem $\uSVP$ to HSP on $\Z^k$ with pseudo-polynomial oracle cost,
we outline a partial algorithm in $\BQP$ that reduces this case of HSP to
a certain lattice learning problem.  Let $f:\Z^k \to X$ be a function that
hides a full-rank lattice $H \subseteq \Z^k$.  We can imitate the Shor--Kitaev
algorithm as described in \Sec{s:abelian}.  In our version of this algorithm,
we calculate the hiding function $f$ on an approximation of a Gaussian state
\[ \ket{\psi_{GC}} \approxprop \sum_{\vx \in \Z^k}
    \exp(-\pi \norm{\vx}_2^2/S^2) \ket{\vx}, \]
then measure the Fourier mode $\vy$ of the input register.  The lattice
$H \subseteq \Z^k$ has a matching dual group $H^\# \subseteq (\R/\Z)^k$,
which in this case is a finite group such that $\abs{H^\#} = [\Z^k:H]$.
The measured result can be interpreted as a vector
\[ \vy_1 = \vy_0 + \vu \in (\R/\Z)^k, \]
where $\vy_0 \in H^\#$, and $\vu$ is an approximately Gaussian noise term
with standard deviation proportional to $S$.  If the cost of evaluating $f$
is pseudo-polynomial, then the difficulty is that $S$ is polynomial in the
problem complexity, rather than exponential as in a conventional version
of the Shor--Kitaev algorithm.

We can lift $H^\#$ to a lattice $H^\circ \subseteq \R^k$, namely the
reciprocal lattice of $H$.  (See \Sec{ss:dual}.)  We can likewise lift
each measured sample $\vy_1 \in (\R/\Z)^k$ to a vector $\tvy_1 \in \R^k$.
Our partial Shor--Kitaev algorithm produces points in $\R^k$ that are
displaced from an unknown lattice $H^\circ$ by approximately Gaussian noise
on a polynomial scale.  If we could learn $H^\circ$ from these noisy samples,
then we could calculate its reciprocal lattice $H = H^{\circ\circ}$ and thus
solve HSP.  However, learning an entire lattice from noisy samples could
be significantly more difficult than Regev's LWE problem \cite{Regev:lwe},
which is about learning a single modular vector from noisy information.
It also seems significantly more difficult than the close vector problem,
where an unknown point in a \emph{known} lattice is displaced by noise.

\section{Proof of \Thm{th:abelian}: Infinite-index AHSP}
\label{s:abelian}

\subsection{The algorithm}
\label{ss:ahsp}

Let $f:\Z^k \to X$ be a function that hides a lattice $H \subseteq \Z^k$
which has some rank $0 \le \ell \le k$.  We assume that $H$ has a $k \times
\ell$ basis matrix $M$ with bit complexity $n$.  By Kannan--Bachem
\cite{KB:smith}, we can assume that $M$ is in Hermite normal form.  Our goal
is to compute $M$ in quantum polynomial time in $n$.

We begin with a description of the algorithm.

\begin{algorithm}{A}  Choose four positive integer parameters $Q =
1/q$, $R = 1/r$, $S = 1/s$, and $T = 1/t$.
\begin{enumerate}
\item Using either the Grover--Rudolph or the Kitaev--Webb algorithm
\cite{GR:prob,KW:wave}, construct $k$ copies of a 1-dimensional approximate
Gaussian state $\ket{\phi_{GC}}$ to form a $k$-dimensional approximate
Gaussian state on a discrete, centered cube in $\Z^k$ of size $Q$ or $Q-1$:
\[ \ket{\psi_{GC}} = \ket{\phi_{GC}}^{\tensor k} \propto \back\hback
    \sum_{\substack{\vx \in \Z^k \\ \norm{\vx}_\infty < Q/2}}
    \back\hback \exp(-\pi s^2 \norm{\vx}_2^2) \ket{\vx} \]

\item Let $U_f$ be the unitary dilation of the hiding function $f$
restricted to the discrete cube $\norm{\vx}_\infty < Q/2$.  Apply $U_f$
to the state $\ket{\psi_{GC}}$ to obtain:
\[ U_f\ket{\psi_{GC}} \propto \back\hback
    \sum_{\substack{\vx \in \Z^k \\ \norm{\vx}_\infty < Q/2}}
    \back\hback \exp(-\pi s^2 \norm{\vx}_2^2) \ket{\vx,f(\vx)}. \]
Discard the output register to obtain a mixed state $\rho_{H,s}$
on the input register.

\item Using the embedding of the discrete interval $\abs{x} < Q/2$
into $\Z/Q$, apply the Fourier operator $F$ for the group $(\Z/Q)^k$
to the state $\rho_{H,s}$.  Measure the state $F\rho_{H,s}F^\dag$ in the
computational basis to obtain a Fourier mode $\vy \in (\Z/Q)^k$, and let
\[ \vy_1 = q\vy \in q(\Z/Q)^k \subseteq (\R/\Z)^k. \]

\item Let $L \subseteq \Q^{k+1} \subseteq \R^{k+1}$ be the lattice generated
by the basis
\begin{eq}{e:lbasis} E = [\ve_1,\ve_2,\ldots,\ve_k,(\tvy_1,t)], \end{eq}
where $\ve_j$ is the $j$th standard basis vector and $\tvy_1$ is a lift
of $\vy_1$ to $\R^k$.  Apply the LLL algorithm \cite{LLL:factoring} to
the basis $E$ to obtain an ordered LLL basis
\[ B = [\vb_1,\vb_2,\ldots,\vb_{k+1}] \]
of $L$ (and therefore $\R^{k+1}$). Let $\ell$ be the smallest integer
such that the first $k+1-\ell$ vectors of the basis $B$ have length
$\norm{\vb_j}_2 \le r$, and let $B_1$ be the submatrix whose columns are
these vectors.

\item Find a nonsingular $(k+1-\ell) \times (k+1-\ell)$ submatrix $B_2$
of $B_1$ that includes the last row.  Reorder the coordinates so
that $B_1$ uses the last $k+1+\ell$ rows, and let
\[ B_3 = B_1 B_2^{-1} =
    \begin{bmatrix} B_4 \\ I_{k-\ell} \end{bmatrix} \oplus [1] \]
Using continued fractions, adjust each entry of the $\ell \times (k-\ell)$
matrix $B_4$ to obtain a matrix $A_4$ with rational entries with denominators
bounded by $R$.

\item If the algorithm is successful up to this stage, then $H_\R$ is
spanned by the columns of $k \times \ell$ matrix
\[ A_5 = \begin{bmatrix} I_\ell \\ -A_4^T \end{bmatrix}. \]
Using Kannan--Bachem \cite{KB:smith}, calculate the Smith normal form
$D = VA_5W$, where $V$ and $W$ are invertible integer matrices and $X$
is a diagonal rational matrix with $k+1-\ell$ vanishing rows.  Then
\[ A_6 = V^{-1}\begin{bmatrix} I_\ell \\ 0 \end{bmatrix}
    = [\va_1,\va_2,\ldots,\va_\ell] \]
is an integral basis for the lattice $H_1 = H_\R \cap \Z^k$.

\item Assuming success up to this stage, we have computed a lattice
$H_1$ that contains the unknown lattice $H$, both with rank $\ell$.
Find $A_6^{-1}(H)$ as the hidden subgroup of the function $f(A_6(\vx))$
using the Shor--Kitaev algorithm.  Finally, apply $A_6$ to obtain an
integral basis of $H$, and use the Kannan--Bachem algorithm for Hermite
normal form to find $M$.
\end{enumerate}
\label{a:abelian} \end{algorithm}

We assume the following scales for the parameters in \Alg{a:abelian}:
\[ 1 \ll R \ll T \ll S \ll Q = \exp(\poly(n)). \]
The parameters have the following interpretations:
\begin{itemize}
\item $q = 1/Q$ is the digital precision of $\vy_1 \in (\R/\Z)^k$.  This is
a source of error in the measured data point $\vy_1$.
\item $s = 1/S$ is the Gaussian noise scale of $\vy_1$, which is another
source of error in this data point.
\item $t = 1/T$ is the flattening of the lattice $L$, which
determines the scale of the LLL algorithm.
\item $r = 1/R$ is a lower bound for the feature length of the dual group
$H^\# \subseteq (\R/\Z)^k$ (defined in \Sec{ss:dual}), which approximately
contains $\vy_1$.
\item $R$ is also an upper bound on any denominator that can arise
in the denoised matrix $A_1$ in step 5.
\end{itemize}

Note that each step of \Alg{a:abelian} uses rational numbers in $\Q$ to
approximate real numbers in $\R$, since finite computers cannot calculate
with arbitrary real numbers.  \Alg{a:abelian} is designed so that if $Q$,
$R$, $S$, and $T$ are all of order $\exp(\poly(n))$, then so is every
numerator and denominator that arises in the calculations.

The main result of this subsection is the following.

\begin{theorem} The parameters $Q$, $R$, $S$, and $T$ can all be chosen
so that they are of order $\exp(\poly(n))$, and otherwise so that
\Alg{a:abelian} runs in $\BQP$ and succeeds with good probability.
\label{th:qrst} \end{theorem}

\Thm{th:qrst} plainly implies \Thm{th:abelian}. We will discuss the
motivation for \Alg{a:abelian} in \Sec{ss:amotive}, and then prove
\Thm{th:qrst} in Sections~\ref{ss:dual} to \ref{ss:lllrat}.

\subsection{Variations and simulations}

Although a direct classical simulation of the quantum work in steps
1--3 of \Alg{a:abelian} would take exponential time, the estimates in
\Sec{ss:fourier} show that we can classically produce random samples of
$\vy_1$ as in step 3, if we know the hidden subgroup $H$ and if $Q \gg
S$.  Meanwhile, steps 4--6 are all classical and can thus be directly be
simulated.  Such a simulation in this regime would illustrate that Theorems
\ref{th:abelian} and \ref{th:qrst} are valid and that \Alg{a:abelian}
works, but it might not otherwise reveal anything new.

We believe that \Alg{a:abelian} would work at least as well if we replaced
the initial quantum state $\ket{\psi_{GC}}$ in step 1 with the constant
pure state $\ket{(\Z/Q)^k}$, which is the standard initial state in the
Shor--Kitaev algorithm.  We use the Gaussian state $\ket{\psi_{GC}}$ of
width $S$ only to help prove Theorems \ref{th:abelian} and \ref{th:qrst}.
By contrast, the parameters $Q \gg T \gg R \gg 1$ are all essential to the
function of \Alg{a:abelian}.  It would be interesting to try to simulate
\Alg{a:abelian} with $\ket{(\Z/Q)^k}$ as the initial state in step 1,
or to derive analytic bounds on the performance of this simplified algorithm.

Steps 4--6 of \Alg{a:abelian} can be generalized to make use of $m$
samples from steps 1--3, to produce a lattice $L \subseteq \R^{k+m}$.
With the use of multiple samples, steps 4--6 can be extended to assimilate
the Shor--Kitaev algorithm, so that step 7 is not necessary as a separate
stage.  Heuristic estimates suggest that \Alg{a:abelian} would be more
efficient in both quantum space and time with multiple samples.  Austin
Tran \cite{Tran:thesis} has implemented a simulator of steps 4--6 that
supports the heuristic estimates, even when $H = H_1$ and the Shor--Kitaev
stage is not needed.  It seems interesting that \Alg{a:abelian} can find
$H_1$ in quantum polynomial time using only a single sample, even though
heuristics and simulations suggest using multiple samples, and even though
the Shor--Kitaev algorithm to find $H$ from $H_1$ may strictly require
multiple samples.

\subsection{Motivation and outline}
\label{ss:amotive}

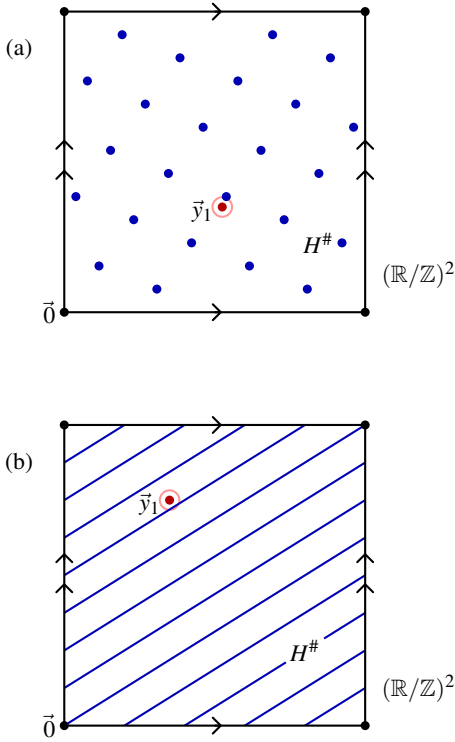
\begin{figure}
\begin{tikzpicture}[thick]
\draw[red!40!white] (2.1,1.4) circle (.13);
\fill[darkred] (2.1,1.4) circle (.06);
\draw (2.1,1.35) node[left] {$\vy_1$};
\foreach \n in {1,2,...,12} {
    \fill[darkblue] ({Mod(\n*16/13,4)},{\n*4/13}) circle (.06);
    \fill[darkblue] ({Mod(\n*16/13+2,4)},{\n*4/13}) circle (.06); }
\draw (0,0) rectangle (4,4);
\draw[-angle 90] (2.0,0) -- (2.1,0); \draw[-angle 90] (2.0,4) -- (2.1,4);
\draw[-angle 90] (0,1.8) -- (0,1.9); \draw[-angle 90] (0,2.2) -- (0,2.3);
\draw[-angle 90] (4,1.8) -- (4,1.9); \draw[-angle 90] (4,2.2) -- (4,2.3);
\fill (0,0) circle (.06) node[left] {$\vec0$}; \fill (4,0) circle (.06);
\fill (0,4) circle (.06); \fill (4,4) circle (.06);
\draw ({48/13},{12/13}) node[left] {$H^\#$};
\draw (4.1,0.5) node[right] {$(\R/\Z)^2$};
\draw (-.6,3.5) node {(a)};
\tikzset{yshift=-5.5cm};
\draw[red!40!white] (1.4,3) circle (.13);
\fill[darkred] (1.4,3) circle (.06);
\draw (1.4,2.95) node[left] {$\vy_1$};
\begin{scope}
\clip (0,0) rectangle (4,4);
\foreach \y in {-2,-1.5,...,4} {
    \draw[darkblue] (0,\y) -- (4,{\y+2.5}); }
\end{scope}
\draw (0,0) rectangle (4,4);
\draw[-angle 90] (2.0,0) -- (2.1,0); \draw[-angle 90] (2.0,4) -- (2.1,4);
\draw[-angle 90] (0,1.8) -- (0,1.9); \draw[-angle 90] (0,2.2) -- (0,2.3);
\draw[-angle 90] (4,1.8) -- (4,1.9); \draw[-angle 90] (4,2.2) -- (4,2.3);
\fill (0,0) circle (.06) node[left] {$\vec0$}; \fill (4,0) circle (.06);
\fill (0,4) circle (.06); \fill (4,4) circle (.06);
\fill[white] (3.2,1) circle (.3);
\draw (3.2,1) node {$H^\#$};
\draw (4.1,0.5) node[right] {$(\R/\Z)^2$};
\draw (-.6,3.5) node {(b)};
\end{tikzpicture}
\caption{Two examples of $H^\# \subseteq (\R/\Z)^2$ with an approximate
    sample $\vy_1$.  In (a), $H$ has rank 2 and $H^\#$ is finite.  In (b),
    $H$ has rank 1 and $H^\#$ is 1-dimensional.}
\label{f:hsharp} \end{figure}

Given $\vy \in (\R/\Z)^k$, we let $\tvy$ be its unique representative
in the half-open cube $(-1/2,1/2]^k$.  We also define
\[ \norm{\vy}_2 \defeq \norm{\tvy}_2. \]

It is standard and simplest to set the QFT radix $Q$ in \Alg{a:abelian} to
be a power of 2.  However, there is also an efficient quantum algorithm
for any radix \cite{Kitaev:stab,HH:fourier}.  The specific choice of $Q$
does not matter in our context, other than that we need $Q \gg S$ and
$\norm{Q}_\bit = \poly(n)$.

As before, let $f:\Z^k \to X$ be a function that hides a subgroup $H$.
As discussed in \Sec{ss:dual}, $H$ has a dual group $H^\#$, defined as the
set of vectors $\vy \in (\R/\Z)^k$ whose inner products with vectors $\vx
\in H$ vanish as elements of $\R/\Z$.  In general, if $Q \gg S \gg 1$,
then step 3 of \Alg{a:abelian} produces a vector
\[ \vy_1 = \vy_0 + \vu \in (\R/\Z)^k, \]
where $\vy_0 \in H^\#$ is uniformly random, and $\vu$ is a noise term
such that
\[ \norm{\vu}_2 = O\bigl(s\sqrt{k}\bigr) \]
with good probability.

If $H$ has full rank $k$, then the Shor--Kitaev algorithm works as follows.
$H^\#$ is a finite group of cardinality
\[ \abs{H^\#} = [\Z^k:H] = \exp(\poly(n)), \]
where $n$ is the bit complexity of a generating matrix of $H$.
If $S = \Omega(k\abs{H^\#}^2)$, then
\[ \norm{\vu}_2 = o\bigl(\frac{1}{\abs{H^\#}^2}\bigr) \]
with good probability.  Each coordinate of the ideal point $\vy_0$ is a
rational number $y_{0,j} = a_j/b_j$ such that $b_j$ divides $\abs{H^\#}$.
In this regime, we can calculate each $y_{0,j}$ exactly from the noisy
coordinate $y_{1,j}$ using continued fractions, as in the original Shor's
algorithm.  With good probability, a set of
\[ O(\log(\abs{H^\#})) = \poly(n) \]
random samples $\{\vy_0^{(\lambda)}\}$ is a generating set of $H^\#$.
We can then produce a generating matrix $A$ of $H$ from a generating matrix
$B$ of $H^\#$, via the Smith normal form (SNF) of $B$.  (See the end of
\Sec{ss:lllrat} for the concept of the SNF of a rational rather than
an integer matrix.)

If $H$ has some lower rank $\ell < k$, then $H^\#$ makes a pattern
of $(k-\ell)$-dimensional stripes in the torus $(\R/\Z)^k$, as in
\Fig{f:hsharp}.  In this case, typically no one coordinate of $\vy_1$
carries any useful information regardless of the bound on the noise $\vu$.
If we knew the direction of the stripes, then we could get useful information
from $\vy_1$ in the orthogonal direction.  However, this would be a circular
strategy, because the tangent directions carry the same information as
the directions of the vectors in $H$, which is exactly the hard part
of \Thm{th:abelian} beyond Shor--Kitaev.  More precisely, let $H_\R$ be
the real vector space spanned by $H$ as defined in \Alg{a:abelian}, and let
$H_\R^\perp$ be its orthogonal space.  Then $H_\R^\perp$ is the tangent
space to $H^\#$ and determines $H_\R = H_\R^{\perp\perp}$.

Our idea is to find multiples $\lambda\vy_0$ of $\vy_0$ whose distance to
$\vec0$ in the torus $(\R/\Z)^k$ is smaller than the minimum spacing between
stripes of $H^\#$.  The shortest lift $\widetilde{\lambda\vy_0}$ of such
a vector to $\R^k$ is then an element of $H_\R$.  The connected subgroup
$H_1^\#$ of $H^\#$ is a torus, albeit one with a complicated geometry.
If $\vy_0$ is chosen generically, then its multiples are dense in $H_1^\#$.
In particular, these multiples eventually yield a basis of $H_\R^\perp$
that consists of short vectors of the form $\widetilde{\lambda\vy_0}$.

We cannot do an exhaustive search for good values of $\lambda$ in the given
exponential range.  Instead, we lift the dense orbit $\{\lambda\vy_0\}$ to
a discrete orbit $\{(\lambda\vy_0,\lambda t)\}$ in $(\R/\Z)^k \oplus \R$.
If $t$ is exponentially small, then the lift is exponentially compressed in
the new direction, and the two orbits track each other for exponentially
many steps.  We can further take the universal cover of $(\R/\Z)^k$ to
convert this discrete orbit to a lattice $L_0$ in $\R^{k+1}$ generated
by $\Z^k$ and $(\widetilde{\vy_0},t)$. If $t$ is small enough, but still
only $\poly(\exp(n))$ small, then the LLL algorithm would find a basis of
$H_\R^\perp \oplus \R$ that consists of short vectors of $L_0$.

In reality, we do not have access to $\vy_0$, rather we have access to its
noisy counterpart $\vy_1 = \vy_0 + \vu$.  Thus, we instead apply the LLL
algorithm to the lattice $L$ generated by $\Z^k$ and $(\widetilde{\vy_1},t)$.
If $\lambda$ is of order $\exp(\poly(n))$, and if the rescaled noise
$\lambda\norm{\vu}_2$ is small enough, then the LLL algorithm applied to
the lattice $L$ yields an adequate approximate basis $B_1$ of $H_\R^\perp
\oplus \R$.  However, the matrix entries of $B_1$ have continuous degrees
of freedom just as $\vy_0$ and $\vy_1$ do.   Step 5 computes another matrix
$B_3$ which is essentially the reduced column echelon form of $B_1$, and
which has no continuous freedom as a description of $H_\R^\perp \oplus \R$.
We can then remove the noise from the entries of $B_3$ using continued
fractions to precisely learn $H_\R^\perp$.

\subsection{Pontryagin duality}
\label{ss:dual}

In this subsection and the next one, we review some elements Fourier
theory for abelian Lie groups, including the concepts of Pontryagin
duality and tempered distributions.   For background on these topics,
see for example Rudin \cite[Ch.~1]{Rudin:fourier} and H\"ormander
\cite[Ch.~7]{Hormander:partial1}.

If $A$ is a locally compact abelian group, then it has a \emph{Pontryagin
dual} $\hA$, which is defined as the set of continuous homomorphisms from
$A$ to the unit complex numbers:
\[ \hA \defeq \{f:A \to S^1 \subseteq \C\}. \]
The group law on $\hA$ is given by pointwise multiplication of these
homomorphisms.  The following additional structures are harder to define,
but also exist canonically:
\begin{enumerate}
\item A locally compact topology on $\hA$, constructed so that
$\widehat{\widehat{A\,}} = A$.
\item Matching Haar measures on $A$ and $\hA$ that yield
Hilbert spaces $L^2(A)$ and $L^2(\hA)$.
\item A unitary Fourier operator
\begin{eq}{e:fourier} F = F_A:L^2(A) \to L^2(\hA) \end{eq}
such that $F_A^\dag$ is the complex conjugate of $F_{\hA}$.
\end{enumerate}

\begin{examples} If $A$ is a discrete group, then its Hilbert space
$L^2(A)$ is the space of square-summable functions on $A$ as a set, and
is more commonly written $\ell^2(A)$. If $A$ is finite, then $\ell^2(A)$
is the vector space of all functions on $A$ and is also written as $\C[A]$.

In our proof of \Thm{th:abelian}, we will need the following standard
cases of Pontryagin duality:
\begin{align*}
\widehat{\Z^k} &= (\R/\Z)^k & (\R/\Z)^k &= \widehat{\Z^k} \\
\widehat{\R^k} &= \R^k & \widehat{(\Z/Q)^k} &= (\Z/Q)^k.
\end{align*}
\end{examples}

We also summarize some useful properties and related definitions:
\begin{enumerate}
\item  If $A$ and $B$ are two locally compact abelian groups, then
\[ \widehat{A \times B} = \hA \times \hB \qquad
    F_{A \times B} = F_A \tensor F_B. \]
\item If $B \subseteq A$ is a closed subgroup of $A$, then
\[ B^\# \defeq \widehat{A/B} \subseteq \widehat{A}\]
is a closed subgroup of $\hA$, and $\hB \cong \hA/B^\#$ is the
corresponding quotient.
\item Since $\widehat{\widehat{A\,}} = A$, we can relate $A$ and $\hA$
symmetrically with a continuous circle-valued inner product
\[ A \times \hA \to S^1 \cong \R/\Z. \]
If $B \subseteq A$ is a closed subgroup, then $B^\# \subseteq \hA$ is its
Pontryagin orthogonal with respect to this inner product, and $(B^\#)^\#
= B$.
\end{enumerate}

If $H \subseteq \Z^k$ is a hidden subgroup, then $H^\# \subseteq (\R/\Z)^k$
is key to understanding \Alg{a:abelian}.  Concretely, if $\vx \in \Z^k$
and $\vy \in (\R/\Z)^k$, then their dot product $\vx \cdot \vy$ is a
well-defined element of $\R/\Z$.  By definition, $H^\#$ is the set of $\vy
\in (\R/\Z)^k$ such that $\vx \cdot \vy = 0$ for all $\vx \in H$.

To understand the structure of $H^\#$, we define several other associated
groups and vector spaces.  Let $\ell$ be the rank of $H$, let
\[ H_\R \defeq H \tensor \R \subseteq \R^k\]
be the $\ell$-dimensional vector space spanned by elements of $H$.  We can
think of $H_\R$ as a self-dual vector space, using the inner product and
Euclidean geometry that it inherits from $\R^k$.  The set of integer points
\[ H_1 \defeq H_\R \cap \Z^k. \]
is also the set of elements $\vx \in \Z^k$ such that $a\vx \in H$ for some
non-zero integer $a$.  The quotient $H_1/H$ is a finite group, and there
are (non-canonical) group isomorphisms
\begin{eq}{e:noncanon} \Z^k/H \cong (\Z^k/H_1) \oplus (H_1/H)
    \cong \Z^{k-\ell} \oplus (H_1/H).
\end{eq}
It follows that $H_1^\#$ is a $(k-\ell)$-dimensional torus, although
its Euclidean geometry can be much more complicated than that of the
square torus $(\R/\Z)^{k-\ell}$.  Also, $H^\#$ is the disjoint union of
$\abs{H_1/H}$ cosets of its connected subgroup $H_1^\#$.

Since we can also view $H \subseteq \Z^k \subseteq \R^k$ as a closed
subgroup of $\R^k$, we write $H^\bullet$ (rather than $H^\#$) for its
Pontryagin orthogonal in $\R^k$.  Explicitly,
\[ H^\bullet \defeq \widehat{\R^k/H} =
    \{\vy \in \R^k | \forall \vx \in \R^k, \vx \cdot \vy \in \Z\}. \]
$H^\bullet$ is also the lift of $H^\#$ from $(\R/\Z)^k$ to $\R^k$.  Since
$H_\R$ is a vector space, it follows that $H_\R^\bullet = H_\R^\perp$,the
$(k-\ell)$-dimensional perpendicular space to $H_\R$.  $H_\R^\perp$ is
also the connected subgroup of $H^\bullet$ and the tangent space to $H^\#$.

Finally, if $H \subseteq \Z^k \subseteq \R^k$, then it is a maximal
rank lattice in the vector space $H_\R$, and the intersection
\[ H^\circ \defeq H^\bullet \cap H_\R. \]
is its \emph{reciprocal lattice}.  We can more directly define $H^\circ$
as the set of vectors $\vy \in H_\R$ such that $\vx \cdot \vy \in \Z$
for all $\vx \in H$.  Note that if $H$ has maximal rank $k$, then $H_\R =
\R^k$ and $H^\circ = H^\bullet$.

\begin{figure}
\begin{tikzpicture}[thick,scale=1.25]
\useasboundingbox (0,0) rectangle (4,4);
\draw[darkgreen,dashed,fill=lightgreen] (2,2) circle (.465);
\begin{scope}
\clip (0,0) rectangle (4,4);
\foreach \y in {-4,-3,...,6} {
    \draw[darkblue] (0,{\y*4/7}) -- (4,{(\y*4+20)/7}); }
\end{scope}
\draw (0,0) rectangle (4,4);
\draw[-angle 90] (2.0,0) -- (2.1,0); \draw[-angle 90] (2.0,4) -- (2.1,4);
\draw[-angle 90] (0,1.85) -- (0,1.95); \draw[-angle 90] (0,2.15) -- (0,2.2);
\draw[-angle 90] (4,1.85) -- (4,1.95); \draw[-angle 90] (4,2.15) -- (4,2.2);
\draw[darkred,<->] (1,3.4) -- (3,0.6);
\fill (2,2) circle (.06);
\draw[darkred,-] (2.7,2.5) -- ({2.7-20/74},{2.5+28/74});
\draw[very thick,darkred,-] (2.63,2.45) -- (2.77,2.55);
\draw[very thick,darkred,-]
    ({2.63-20/74},{2.45+28/74}) -- ({2.77-20/74},{2.55+28/74});
\fill[white] ({2-50/74},{2+70/74}) circle (.18);
\draw ({2-50/74},{2+70/74}) node {$H_\R$};
\draw (1.95,1.75) node {$\vec0$};
\draw (2.715,2.775) node {$r_0$};
\fill[white] (3.2,{12/7}) circle (.25);
\draw (3.2,{12/7}) node {$H^\#$};
\draw (4.1,0.5) node[right] {$(\R/\Z)^2$};
\end{tikzpicture}
\caption{The relationship among the real space $H_\R$ spanned by $H$, the
    dual $H^\#$ sampled in \Alg{a:abelian}, and the feature length $r_0$
    of $H^\#$.}
\label{f:feature} \end{figure}
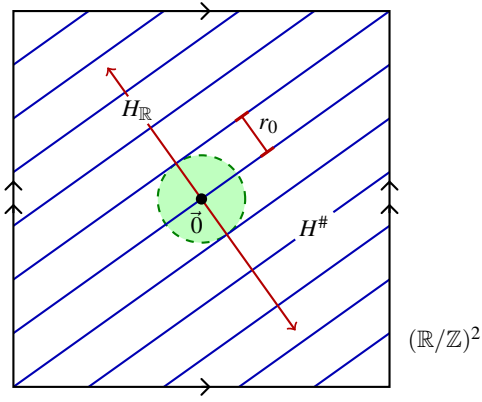

Given $H \subseteq \Z^k$, we define two constants related to the geometry
of $H$ and $H^\#$.  The first constant is the \emph{Gram determinant}
of $H$, defined as
\begin{eq}{e:gram} \Delta \defeq \det(M^T M) \end{eq}
for any generating matrix $M$ of $H$.  (It is not hard to see that
$\Delta$ is independent of the choice of $M$.)  The second constant is
the \emph{feature length}  of $H^\#$ as a subgroup of $(\R/\Z)^k$.
By definition, the feature length $r_0 = \sup r$ is the supremal radius of
an open round ball $B_{r,2}(\vec0) \subseteq (\R/\Z)^k$ centered at $\vec0$,
such that $B_{r,2}(\vec0) \cap H^\#$ is a single $(k-\ell)$-dimensional
disk. \Fig{f:feature} shows an example, and illustrates that the feature
length is also the length of the shortest non-zero vector in $H^\circ$.

\begin{lemma} If $\Delta$ and $r_0$ are defined from $H$ as above, then:
\[ \Vol H^\# = \Vol H_\R/H = \sqrt{\Delta}, \qquad
    H^\circ \subseteq \frac1{\Delta}\Z^k, \qquad
    r_0 \ge \frac1{\Delta}. \]
\eatline \label{l:feature} \end{lemma}

\begin{proof} Let $M$ be a matrix whose columns generate $H$, and let $V =
\Vol H_\R/H$.  The volume $V$ is known as the lattice volume or lattice
determinant of $H$, and one standard formula for it is $V = \sqrt{\Delta}$.
Using the fact that $M^TM$ is the matrix of dot products of the basis $M$,
one can check that $(M^TM)^{-1}M$ is a generating matrix for $H^\circ$.
This implies the standard fact that the lattice volume of $H^\circ$ is
$1/V$, \ie, that reciprocal lattices have reciprocal lattice volumes.

For the rest of the first equation, we can use the orthogonal direct sum
\[ H^\bullet = H^\circ \oplus H_\R^\perp. \]
Since lattice volume is the reciprocal of lattice density, $H^\bullet$ thus
has $(k-\ell)$-dimensional volume density $\Vol H_\R/H$ per unit volume
of $\R^k$.  In particular, since the quotient $(\R/\Z)^k$ has unit volume,
$H^\#$ has volume $V$.

For the second equation, observe that $H \subseteq H^\circ$ since vectors
in $H$ have integer dot products with each other.  It follows that
$\abs{H^\circ/H} = V^2 = \Delta$.  Applying Lagrange's theorem to the
finite group $H^\circ/H$, we obtain that $\Delta H^\circ \subseteq H$,
equivalently that
\[ H^\circ \subseteq \frac1{\Delta}H \subseteq \frac1{\Delta}\Z^k. \]
For an alternate proof of this step, $\Delta(M^TM)^{-1}$ is an integer
matrix, namely the adjugate matrix of $M^TM$, which again shows that
$\Delta H^\circ \subseteq H$.

For the last equation, since the lattice distance of $\frac1{\Delta}\Z^k$
is $1/\Delta$, the lattice distance $r_0$ of $H^\circ \subseteq
\frac1{\Delta}\Z^k$ is at least $1/\Delta$.
\end{proof}

\subsection{Tempered distributions and Fourier transforms}
\label{ss:tempered}

It will be convenient to use generalized states on $\R^k$ that are not in
the Hilbert space $L^2(\R^k)$, such as Dirac delta functions and periodic
functions.  In physics, such entities are known as ``non-normalizable wave
functions".  In rigorous mathematics, a widely useful type of generalized
state is known as a tempered distribution, introduced by Laurent Schwartz.
We will extend all four of the Fourier transforms listed above to tempered
distributions on their domains.  We will also interpret all four as special
cases of the Fourier transform operator on tempered distributions on $\R^k$.

A function $\psi(\vx)$ on $\R^k$ is \emph{Schwartz} when it has all
derivatives, or $\psi(\vx) \in C^\infty(\R^k)$, and every derivative
decays superpolynomially as $\norm{\vx} \to \infty$.  For example, Gaussian
functions on $\R^k$ are Schwartz.  The vector space of Schwartz functions
on $\R^k$ is denoted $\cS(\R^k)$.  It has a natural topology which arises
as the limit (or mutual refinement) of an infinite sequence of Hilbert
space topologies.

The concept of a Schwartz function extends to functions $\psi:A \to
\C$ for other choices of $A$ as follows:
\begin{enumerate}
\item If $A = (\R/\Z)^k$, then $\psi(\vx)$ is Schwartz when
it is $C^\infty$.
\item If $A = \Z^k$, then $\psi(\vx)$ is Schwartz when it decays
    superpolynomially as $\norm{\vx} \to \infty$.
\item If $A = (\Z/Q)^k$, then every function is Schwartz.
\end{enumerate}
A key property of any Schwartz function is that its Fourier transform is
also Schwartz.  In other words, there is a Fourier operator
\[ F:\cS(A) \to \cS(\hA). \]

If $A$ is one of $\R^k$, $\Z^k$, $(\R/\Z)^k$, or $(\Z/Q)^k$, then a
\emph{tempered distribution} (or TD) is a linear functional
\[ \xi:\cS(A) \to \C \]
which is continuous with respect to the topology on $\cS(A)$.  At a rigorous
level, the input to $\xi$ is a Schwartz function $\psi$.  For example, if
$\xi = \delta_{\vx}$ is the unit atomic measure (or Dirac delta function)
at $\vx$, then it is by definition the evaluation map $\delta_{\vx}(\psi)
= \psi(\vx)$.  In Dirac notation, $\xi$ is an abstract state with a
well-defined inner product $\braket{\psi|\xi}$ with every sufficiently tame
state $\ket{\psi}$.  Tempered distributions are also written as functions
interpreted through their sums or integrals.  For instance, in the case
$A = \R^k$,
\[  \braket{\psi|\xi} \defeq \int_{\R^k}
    \overline{\psi(\vx)}\xi(\vx)\;d\vx. \]
The space of TDs is denoted $\cS(A)^*$, since it is the topological dual
of $\cS(A)$.

\begin{remark} Our case-by-case definition of Schwartz functions
and TDs amounts to special cases of a general definition for
arbitrary locally compact, abelian groups which is due to Bruhat
\cite{Osborne:schwartz,Bruhat:distrib}.  Our treatment describes the
Schwartz--Bruhat theory when $A$ (and therefore $\hA$ and every closed $B
\subseteq A$) is a finite-dimensional abelian Lie group or discrete group
whose group of components is finitely generated.  Every group of this type
is isomorphic to a direct product of copies of the four basic cases $\R$,
$\Z$, $S^1 \cong \R/\Z$, and $\Z/Q$.
\end{remark}

An important intermediate class between Schwartz functions and TDs is the
space $\cP(A)$ of \emph{polynomially growing} functions.  A smooth function
$\psi \in C^\infty(\R^k)$ is polynomially growing if each derivative of
$\psi(\vx)$ is $\poly(\vx)$, and we let $\cP(\R^k)$ be the space of these.
In the other three cases of interest to us:
\begin{align*}
\cP(\Z^k) &\defeq \cS(\Z^k)^* \\
\cP((\R/\Z)^k) &\defeq \cS((\R/\Z)^k) = C^\infty((\R/\Z)^k) \\
\cP((\Z/Q)^k) &\defeq \cS((\Z/Q)^k) = \C[(\Z/Q)^k].
\end{align*}

We state the following standard facts without proof.

\begin{proposition} If $A \in \{\R^k, \Z^k, (\R/\Z)^k, (\Z/Q)^k\}$,
then the spaces $\cS(A)$, $\cP(A)$, and $\cS(A)^*$ have these properties.
\begin{enumerate}
\item They nest as follows:
\[ \begin{matrix} \cS(A) & \subseteq & \cP(A) & \\[.5ex]
\reflectbox{\rotatebox{90}{$\supseteq$}} & &
    \reflectbox{\rotatebox{90}{$\supseteq$}} \\
L^2(A) & \subseteq & \cS(A)^* \end{matrix} \]
\item The Fourier operator $F$ in \eqref{e:fourier} extends to an operator
\[ F:\cS(A)^* \to \cS(\hA)^*. \]
\item if $\phi \in \cP(A)$ and $\xi \in \cS(A)^*$, then the
product $\phi\xi$ and the convolution $\hphi * \xi$ are both TDs
in $\cS(A)^*$.
\end{enumerate}
In addition, the TD spaces embed in each other as follows:
\begin{enumerate} \setcounter{enumi}{4}
\item $\cS(\Z^k)^*$ embeds in $\cS(\R^k)^*$ using
Dirac delta functions.
\item $\cS((\R/\Z)^k)^*$ embeds in $\cS(\R^k)^*$ by lifting
each $\xi \in \cS((\R/\Z)^k)^*$ to a periodic TD on $\R^k$.
\item $\cS((\Z/Q)^k)^*$ embeds in $\cS((\R/\Z)^k)^*$ using
Dirac delta functions with spacing $1/Q$.
\item $\cS((\Z/Q)^k)^*$ embeds in $\cS(\Z^k)^*$ by lifting each $\xi \in
\cS((\Z/Q)^k)^*$ to a periodic TD on $\Z^k$.
\end{enumerate}
\label{p:tempered} \end{proposition}

Using the notation $\ket{\hpsi} = F\ket{\psi}$, we assume the following
explicit forms for the Fourier transform in the four cases $\R^k$, $\Z^k$,
$(\R/\Z)^k$, and $(\Z/Q)^k$:
\begin{enumerate}
\item If $\psi \in \cS(\R^k)^*$, then
\[ \hpsi(\vy) = \int_{\R^k} \exp(-2\pi i\vx \cdot \vy)\psi(\vx)\;d\vx. \]
\item If $\psi \in \cS((\R/\Z)^k)^*$, then
\[ \hpsi(\vy) = \int_{(\R/\Z)^k} \back\back \exp(-2\pi i\vx \cdot \vy)
    \psi(\vx)\;d\vx. \]
\item If $\psi \in \cS(\Z^k)^*$, then
\[ \hpsi(\vy) = \sum_{\vx \in \Z^k} \exp(-2\pi i\vx \cdot \vy) \psi(\vx). \]
\item If $\psi \in \cS((\Z/Q)^k)^*$, then
\[ \hpsi(\vy) = Q^{-k/2}  \back\hback \sum_{\vx \in (\Z/Q)^k} \back
    \exp\big(\frac{-2\pi i\vx \cdot \vy}Q\big)\psi(\vx). \]
\end{enumerate}
One advantage of these conventions is that the Fourier operator $F$ acting
on TDs commutes with the embeddings in cases 5 and 6 in \Prop{p:tempered}.
($F$ commutes with the embeddings in cases 7 and 8 up to a factor
of $Q^{k/2}$.)  Another advantage is that multiplication of functions
transforms to convolution without any extra factor in the first three cases.
More precisely, we recall the following proposition without proof.

\begin{proposition}  If $A \in \{\R^k, \Z^k, (\R/\Z)^k\}$, if $\psi \in
\cP(A)$ is polynomially growing, and if $\xi \in \cS(A)^*$ is a TD, then
\[ \widehat{\psi\xi} = \hpsi * \hxi. \]
\eatline \label{p:prodconv} \end{proposition}

If $A$ lies in any of the four cases of \Prop{p:tempered} and $B \subseteq
A$ is a closed subgroup, then we can interpret Haar measure $\delta_B$ on
$B$ as a measure and a TD on $A$.  If $A$ is $\R^k$ or $(\R/\Z)^k$ and $B$
is $\ell$-dimensional, then we normalize $\delta_B$ by defining it to be
$\ell$-dimensional Euclidean volume in the continuous directions of $B$.
For example, if $A = \R^k$ and $B = \R^\ell$, then $\delta_B$ is the constant
function 1 (or Lebesgue measure) in the first $\ell$ coordinates, times a
Dirac delta at $\vec0$ in the other $k-\ell$ coordinates.  We extend the
notation to Haar measure $\delta_{B+\vx}$ on a coset $B+\vx \subseteq A$.

Finally, we state the Fourier transform of $\delta_B$ without proof.
(The proof is similar to the proof of the first part of \Lem{l:feature}.)
To state the result, we define a certain relative volume $V(A,B)$
if $A$ is one of $\R^k$, $(\R/\Z)^k$, or $\Z^k$ and
$B \subseteq A$ is a closed subgroup.  If $A$ is $\R^k$ or $\Z^k$, then
recall that $B_\R$ is the linear span of $B$ in $\R^k$, and let
\[ V(A,B) \defeq \Vol B_\R/B. \]
If $A = (\R/\Z)^k$, then let
\[ V(A,B) \defeq \Vol A/B. \]
Note that since there are several choices for the dimensions of $B_\R$
and $B$, we use $\ell$-dimensional volume in these formulas for several
choices of $\ell$.

\begin{proposition} If $A$ is one of $\R^k$, $\Z^k$, or $(\R/\Z)^k$ and
$B \subseteq A$ is a closed subgroup, then
\[ \hdelta_B = \frac{\delta_{\widehat{A/B}}}{V(A,B)} \in \cS(\hA)^*. \]
\eatline \label{p:haarfourier} \end{proposition}

\subsection{The Fourier stage of \Alg{a:abelian}}
\label{ss:fourier}

The goal of this subsection is to prove a key assertion at the beginning
of \Sec{ss:amotive}.

\begin{lemma} Suppose that $Q \ge S^2$ and that $S \ge \max(\Delta^2,C)$,
where $C > 0$ is a suitable absolute constant.  Then the vector $\vy_1$
produced by step 3 of \Alg{a:abelian} equals (in distribution)
\[ \vy_1 = \vy_0 + \vu \in (\R/\Z)^k, \]
where $\vy_0 \in H^\#$ is uniformly random and $\vu$ is a noise term
such that
\[ \Pr\big[\norm{\vu}_2 \le \sqrt{k}s\big] \ge \frac34. \]
\eatline \label{l:unoise} \end{lemma}

To prove \Lem{l:unoise}, we will combine the definitions of TD's and
Fourier operators with Propositions~\ref{p:tempered},~\ref{p:prodconv},
and~\ref{p:haarfourier} as in the following preliminary example.  First,
one can check that in $\R^k$, the Gaussian
\[ \gamma_{\R^k,s}(\vx) = \exp(-\pi s^2 \norm{\vx}_2^2) \]
transforms to the Gaussian
\[ \widehat{\gamma_{\R^k,s}}(\vy) = s^{-k} \gamma_{\R^k,1/s}(\vy)
    = s^{-k} \exp(-\pi s^{-2} \norm{\vy}_2^2). \]
The equation
\[ \widehat{\gamma_{\R^k,s}\delta_{\Z^k}}
    = \gamma_{\R^k,1/s} * \delta_{\Z^k} \]
then tells us that the Fourier transform of the discrete
Gaussian
\[ \gamma_{\Z^k,s}(\vx) = \exp(-\pi s^2 \norm{\vx}_2^2) \]
is a sum of Gaussians,
\begin{eq}{e:gsum} \widehat{\gamma_{\Z^k,s}}(\vy) = \sum_{\vz \in \Z^k}
    s^{-k} \exp(-\pi s^{-2}\norm{\tvy+\vz}_2^2), \end{eq}
where in the sum we use any lift of $\vy \in (\R/\Z)^k$
to $\tvy \in \R^k$.

The following lemma enables us to discretize a multidimensional Gaussian
and/or cut off its tail without much change to its norm or its Fourier
transform.

\begin{lemma} The following estimates hold as $S \to \infty$, and uniformly
in the other parameters when $S \ge k$, $Q \ge S^2$, and $\vy \in
[-\frac12,\frac12]^k$:
\begin{align}
\sum_{\vx \in \Z^k} \exp(-\pi \norm{\vx}_2^2/S^2)
    &= S^k + \exp(-\Omega(S^2)) \label{e:gnorm} \\
\sum_{\substack{\vx \in \Z^k \\ \norm{\vx}_\infty \ge Q/2}} \back\hback
    \exp(-\pi \norm{\vx}_2^2/S^2) &= \exp(-\Omega(S^2)) \label{e:gtail} \\
\sum_{\vz \in \Z^k \setminus \{\vec0\}} \back\hback
    \exp(-\pi S^2\norm{\vy+\vz}_2^2) &= \exp(-\Omega(S^2)). \label{e:glattice}
\end{align}
\eatline \label{l:gaussest} \end{lemma}

Following \Lem{l:gaussest}, we assume in \Alg{a:abelian} that $k \le S$ and
that $Q \ge S^2$.  We also assume in the algorithm and the proof that $Q$
is even, although this is not essential.

\begin{proof} We start with an estimate in one dimension that holds for
all real $S,y>0$:
\begin{multline}
\sum_{z=0}^\infty \exp(-\pi S^2(z+y)^2)
    < \sum_{z=0}^\infty \exp(-\pi S^2(y^2+2yz)) \\
    = \frac{\exp(-\pi S^2 y^2)}{1-\exp(-2\pi S^2 y)}
    < \frac{2\exp(-\pi S^2 y^2)}{\min(1,S^2 y)}.
\label{e:tail} \end{multline}
To prove estimate \eqref{e:gnorm}, let
\[ t_1 = \sum_{\vx \in \Z^k} \exp(-\pi \norm{\vx}_2^2/S^2). \]
We evaluate \equ{e:gsum} with $\tvy = \vec0$ and $s = 1/S$ to obtain
\[ t_1 = \sum_{\vz \in \Z^k} S^k \exp(-\pi S^2 \norm{\vz}_2^2) \\
    = S^k \left(\sum_{z \in \Z} \exp(-\pi S^2 z^2)\right)^k. \]
We apply \equ{e:tail} with $y=1$ to the 1-dimensional sum
in the last expression for $t_1$, to obtain
\[ t_1 = S^k \big(1+O(\exp(-\pi S^2))\big)^k = S^k +\exp(-\Omega(S^2)). \]
The second equality is valid when $k = o(\exp(\pi S^2))$,
and certainly $k \le S$ suffices.

To prove estimate \eqref{e:gtail}, let
\[ t_2 = \back \sum_{\substack{\vx \in \Z^k \\
    \norm{\vx}_\infty \ge Q/2}} \back \exp(-\pi \norm{\vx}_2^2/S^2). \]
By the union bound, $t_2$ is bounded by the total of the sums when any
one coordinate $x_j$ satisfies $\abs{x_j} \ge Q/2$.  Since the coordinates
are all equivalent, we get
\[ t_2 < 2k\bigg(\sum_{\vx \in \Z^{k-1}} \hback
    \exp(-\pi \norm{\vx}_2^2/S^2)\bigg) \bigg(\sum_{x_k = Q/2}^\infty
    \hback \exp(-\pi x_k^2/S^2)\bigg). \]
The first sum in this product is the quantity $t_1$ in $k-1$ rather than $k$
dimensions. For the second sum in the product, we apply \equ{e:tail} after
setting $y=Q/2$ and substituting $s = 1/S$ for $S$.  Recalling also that
$Q \ge S^2$, we obtain:
\begin{align*}
t_2 &< 2k(S^{k-1} + \exp(-\Omega(S^2)) O(\exp(-\pi Q^2/4S^2)) \\
    &= \exp(-\Omega(S^2)).
\end{align*}

To prove estimate \eqref{e:glattice}, let
\[ t_3 = \sum_{\vz \in \Z^k} \back
    \exp(-\pi S^2\norm{\vy+\vz}_2^2). \]
For each $1 \le j \le k$, we apply \equ{e:tail} in the two cases $y =
1 \pm y_j$ to obtain:
\[ \sum_{z \in \Z} \exp(-\pi S^2 (y_j + z)^2) = \exp(-\pi S^2 y_j)
    + \exp(-\Omega(S^2)). \]
Then, recalling that $k \le S$:
\begin{align*}
t_3 &= \prod_{j=1}^k \sum_{z \in \Z} \exp(-\pi S^2 (y_j + z)^2) \\
    &= \prod_{j=1}^k \big(\exp(-\pi S^2 y_j) + \exp(-\Omega(S^2))\big) \\
    &= \exp(-\pi S^2\norm{\vy}_2^2) + \exp(-\Omega(S^2)),
\end{align*}
which is equivalent to the desired sum since the formula for $t_3$ includes
the term with $\vz = \vec0$.
\end{proof}

We will approximate \Alg{a:abelian} using two operators $P$ and $\hP$
that are unbounded as operators on Hilbert spaces.  To deal with this,
we will work with alternate norms in which $P$ and $\hP$ are bounded,
and with restricted states that are bounded in these alternate norms.
To later compare with standard Hilbert 2-norms, we will need these standard
inequalities, where in the first case $A$ can be any discrete set:
\begin{align*}
\norm{\psi}_2 &\le \norm{\psi}_1 &&\text{if $\psi:A \to \C$} \\
\norm{\psi}_2 &\le \norm{\psi}_\infty &&\text{if $\psi:(\R/\Z)^k \to \C$} \\
\norm{\psi}_2 &\le Q^{k/2}\norm{\psi}_\infty &&\text{if\ }
    \psi:(\Z/Q)^k \to \C.
\end{align*}
In proving \Thm{th:abelian}, we will take $A$ to be $(\Z/Q)^k$, $\Z^k$,
and $\Z^k \times X$, where $X$ is the (discrete) target of the hiding
function $f$ as in diagram \eqref{e:hsp}.  We will use these Banach spaces
that correspond to these norms:
\begin{gather*}
\ell^1(\Z^k) \subseteq \ell^2(\Z^k) \qquad
\ell^1(\Z^k \times X) \subseteq \ell^2(\Z^k \times X) \\
C((\R/\Z)^k) \subseteq L^2((\R/\Z)^k) \\
\ell^1((\Z/Q)^k) = \ell^2((\Z/Q)^k) = \ell^\infty((\Z/Q)^k) = \C[(\Z/Q)^k]
\end{gather*}
Here $C((\R/\Z)^k)$ denotes the continuous functions on $(\R/\Z)^k$,
which is a Banach space using the $\infty$-norm.

We will also need the mixed norm
\begin{align*}
\norm{\psi}_{\infty, 1} \defeq \sum_{a \in X} \norm{\psi(\cdot,a)}_\infty
\end{align*}
when $\psi$ is either a function on $(\Z/Q)^k \times X$ or
a continuous function on $(\R/\Z)^k \times X$.  These
inequalities are routine, albeit not entirely standard:
\begin{align}
\norm{\psi}_2 &\le Q^{k/2} \norm{\psi}_{\infty,1}
    &&\text{if $\psi:(\Z/Q)^k \times X \to \C$} \label{e:qmixed} \\
\norm{\psi}_2 &\le \norm{\psi}_{\infty,1}
    &&\text{if $\psi:(\R/\Z)^k \times X \to \C$}
\end{align}

Finally, we will use operator norms between Banach spaces.   Recall that
if $O:\cX \to \cY$ is a bounded linear operator between two Banach spaces
$\cX$ and $\cY$, then it has a norm
\[ \norm{O}_{\cX \to \cY} \defeq
    \sup_{\norm{\psi}_\cX = 1} \norm{O\psi}_\cY. \]
If $\cX$ and $\cY$ have $p$-norms, then we abbreviate the notation for the
corresponding operator norm using these values of $p$, as in the statement
of \Lem{l:pffp}.

We now define the operators $P$ and $\hP$ as maps
\begin{align*}
P&:\ell^1(\Z^k) \to \ell^1((\Z/Q)^k) \\
\hP&:C((\R/\Z)^k) \to \ell^\infty((\Z/Q)^k)
\end{align*}
given by the formulas:
\begin{align*}
(P\psi)(\vx) &\defeq \sum_{\vz \in \Z^k} \psi(\tvx + Q\vz) \\
(\hP\psi)(\vy) &\defeq \psi(q\vy),
\end{align*}
where $\tvx$ is any lift of $\vx$ from $(\Z/Q)^k$ to $\Z^k$.  As \Lem{l:pffp}
below implies, the Fourier operators $F$ acting on $\Z^k$ and $(\Z/Q)^k$
also have the following alternate domains and targets:
\begin{align*}
F_{\Z^k}&:\ell^1(\Z^k) \to C((\R/\Z)^k) \\
F_{(\Z/Q)^k}&:\ell^1((\Z/Q)^k) \to \ell^\infty((\Z/Q)^k)
\end{align*}

The following lemma is standard in functional analysis, and we omit
the proof.

\begin{lemma} The operators $F_{\Z^k}$, $F_{(\Z/Q)^k}$,
$P$, and $\hP$ have the following norms:
\begin{align*}
\norm{F_{\Z^k}}_{1 \to \infty} &= 1 &
    \norm{F_{(\Z/Q)^k}}_{1 \to \infty} &= Q^{-k/2} \\
\norm{P}_{1 \to 1} &= 1 & \norm{\hP}_{\infty \to \infty} &= 1.
\end{align*}
The four operators together satisfy the relation
\[ \hP F_{\Z^k} = Q^{k/2} F_{(\Z/Q)^k}P. \]
\eatline \label{l:pffp} \end{lemma}

We approximate steps 1--3 in \Alg{a:abelian} with a discrete Gaussian
input defined on all of $\Z^k$.  Explicitly, recall that step 1 creates
the normalized state
\[ \ket{\psi_{GC}} \propto \back\hback
    \sum_{\substack{\vx \in \Z^k \\ \norm{\vx}_\infty < Q/2}}
    \back\hback \exp(-\pi s^2 \norm{\vx}_2^2) \ket{\vx}, \]
and let the approximation be the normalized state
\[ \ket{\psi_G} \propto \sum_{\vx \in \Z^k}
    \exp(-\pi s^2 \norm{\vx}_2^2) \ket{\vx}. \]
Then step 2 of \Alg{a:abelian} begins with the dilation $U_f$, which
satisfies $\norm{U_f}_{1 \to 1} = 1$ as a map
\[ U_f:\ell^1(\Z^k) \to \ell^1(\Z^k \times X). \]
The rest of step 2 discards the output register, but this action is optional,
since what is meant is that the algorithm never uses the output register.
For the time being, we keep the output register to help prove estimates.

\begin{lemma} Let $P$, $\hP$, and any Fourier operator $F$ also denote
the operators $P \tensor I_X$, $\hP \tensor I_X$, and $F \tensor I_X$,
where $I_X$ is the identity acting on $\ell^2(X)$.  If $S \ge k$ and $Q
\ge S^2$, then
\[ \norm{Q^{-k/2} \hP F_{\Z^k}U_f\psi_G - F_{(\Z/Q)^k}
    PU_f\psi_{GC}}_2 = \exp(-\Omega(S^2)). \]
\eatline \label{l:pffpgc} \end{lemma}

Conceptually, \Lem{l:pffpgc} says that we get a good approximation to
steps 1--3 of \Alg{a:abelian} if we first apply the infinite Fourier operator
$F_{\Z^k}$ to the state $U_f\ket{\psi_G}$, and then specialize from a state
on $(\R/\Z)^k$ to a state on $q(\Z/Q)^k$.  The actual algorithm works in
the other order:  It first interprets the state $\ket{\psi_{GC}}$ as a
state on $(\Z/Q)^k$, and then applies the Fourier operator $F_{(\Z/Q)^k}$.

\begin{proof} We first bound the distance between
$\ket{\psi_G}$ and $\ket{\psi_{GC}}$. \Equ{e:gnorm} implies
\[ \psi_G(\vx) = 2^{k/4} s^{k/2} (1+\exp(-\Omega(S^2)) \sum_{\vx \in \Z^k}
    \exp(-\pi s^2 \norm{\vx}_2^2), \]
\ie, that $2^{k/4} s^{k/2}$ is nearly perfect quantum normalization for
the discrete Gaussian state $\ket{\psi_G}$.  Given that this factor is
$\exp(O(S^2))$, \equ{e:gtail} then says that
\[ \norm{\psi_G - \psi_{GC}}_1 = \exp(-\Omega(S^2)). \]

Since $\norm{U_f}_{1 \to 1} = 1$, we then obtain
\[ \norm{U_f\psi_G - U_f\psi_{GC}}_1 = \exp(-\Omega(S^2)), \]
where here we interpret $U_f\psi_{GC}$ as a state on
$\Z^k \times X$ whose support in the input register lies in the discrete
cube $\norm{\vx}_\infty < Q/2$.

Step 3 reinterprets the input register as $(\Z/Q)^k$.  The discrete
cube injects into this discrete torus, but $\Z^k$ as a whole does not.
Nonetheless, we can approximate the reinterpretation
by applying the operator $P$ to obtain
\[ \norm{PU_f\psi_G - PU_f\psi_{GC}}_1 = \exp(-\Omega(S^2)). \]
Step 3 then applies the QFT for $(\Z/Q)^k$ to the input register.
Since $\norm{F_{(\Z/Q)^k}} = Q^{-k/2}$ by \Lem{l:pffp}, we obtain:
\begin{multline*}
\norm{F_{(\Z/Q)^k}PU_f\psi_G - F_{(\Z/Q)^k}PU_f\psi_{GC}}_{\infty, 1} \\
    = Q^{-k/2}\exp(-\Omega(S^2)).
\end{multline*}
\Lem{l:pffp} and \equ{e:qmixed} then imply
\begin{multline*}
\norm{Q^{-k/2} \hP F_{\Z^k}U_f\psi_G - F_{(\Z/Q)^k}
    PU_f\psi_{GC}}_2 \\
\le \norm{\hP F_{\Z^k}U_f\psi_G - Q^{k/2} F_{(\Z/Q)^k}
    PU_f\psi_{GC}}_{\infty, 1} \\ = \exp(-\Omega(S^2)),
\end{multline*}
as desired.
\end{proof}

By \Lem{l:pffpgc}, the state $Q^{-k/2} \hP F_{\Z^k}U_f\ket{\psi_G}$, and more
importantly its measurement in the computational basis, approximates the
result of steps 1--3 of \Alg{a:abelian}.  We can divide the construction
of this state into two stages.  Compared to an ideal process that
produces a uniformly random point $\vy \in H^\#$, roughly speaking the
state $F_{\Z^k}U_f\ket{\psi_G}$ produces a Gaussian noise term, while
the operator $Q^{-k/2} \hP$ then produces a rasterization noise term.
Our goal is to bound both terms.

As a first step, we measure the output register for the state
$U_f\ket{\psi_G}$ to obtain a Gaussian-damped coset state
\begin{eq}{e:gcos} \ket{\psi_{H+\vv}} \propto \sum_{\vx \in H+\vv}
    \exp(-\pi s^2 \norm{\vx}_2^2) \ket{\vx}. \end{eq}
Like other hidden subgroup algorithms, \Alg{a:abelian} never uses the output
register; we measure it only to simplify the analysis.  In our case,
we will obtain uniform statistical bounds for steps 3--6 of \Alg{a:abelian}
for every possible coset state $\ket{\psi_{H+\vv}}$.  This uniform outcome
is an advantage of the idealized Gaussian state $\ket{\psi_G}$ over the
truncated state $\ket{\psi_{GC}}$.

Using \Prop{p:tempered}, we extend the state $\ket{\psi_{H+\vv}}$ to a
tempered distribution $\xi_{H+\vv}$ on $\R^k$ defined as
\[ \xi_{H+\vv}(\vx) \propto \exp(-\pi s^2 \norm{\vx}_2^2)
    \delta_{H+\vv}(\vx). \]
To analyze the state $\ket{\hpsi_{H+\vv}}$ on $(\R/\Z)^k$, we analyze its
TD lift $\hxi_{H+\vv}$ on $\R^k$.  We first use the orthogonal direct sum
\[ \R^k = H_\R \oplus H_\R^\perp \]
to express $\hxi_{H+\vv}(\vy)$ in separate variables
\[ \vy = \vy_\spar + \vy_\perp \qquad
    \vy \in \R^k,\ \vy_\spar \in H_\R,\ \vy_\perp \in H_\R^\perp. \]

\begin{lemma} If $\xi_{H+\vv}$ is defined as above, then
\begin{multline*}
\hxi_{H+\vv}(\vy) \propto \tau(\vy) \defeq
    \exp(-2 \pi i \vv \cdot \vy_\perp) \\
    \sum_{\vz \in H^\circ} \exp(-\pi S^2 \norm{\vy_\spar-\vz}_2^2
    -2\pi i \vv\cdot \vz).
\end{multline*}
\eatline \label{l:xihat} \end{lemma}

\Lem{l:xihat} says that the Fourier transform of the Gaussian-damped coset
function $\xi_{H+\vv}$ generalizes the sum of Gaussians in \eqref{e:gsum}.
In the parallel direction $H_\R$, $\xi_{H+\vv}$ is a sum of Gaussians
translated by the reciprocal lattice $H^\circ$, while in the perpendicular
direction $H_\R^\perp$, $\xi_{H+\vv}$ is essentially trivial. In both
directions, there is also a plane wave factor due to the fact that $\xi_H$
is translated by $\vv$.

\begin{proof}
By \Prop{p:prodconv},
\[ \hxi_{H+\vv}(\vy) \propto \exp(-\pi S^2 \norm{\vy}_2^2)
    * \hdelta_{H+\vv}(\vy). \]
Note also that:
\begin{align*}
\delta_{H+\vv} &= \delta_H * \delta_{\vv} \\
\hdelta_{\vv}(\vy) &= \exp(-2\pi i \vv \cdot \vy) \\
\hdelta_H(\vy) &= \delta_{H^\bullet}(\vy) = \delta_{H^\circ}(\vy_\spar).
\end{align*}
By \Prop{p:haarfourier} and by separation of variables,
\begin{align*}
\hxi_{H+\vv}(\vy) &\propto \big(\exp(-\pi S^2 \norm{\vy_\spar}_2^2)
    \exp(-\pi S^2 \norm{\vy_\perp}_2^2) \big) \\
    &\quad * \big(\hdelta_{H^\circ}(\vy_\spar)
    \exp(-2\pi i \vv \cdot \vy_\spar)
    \exp(-2\pi i \vv \cdot \vy_\perp) \big).
\end{align*}
Since the factors are in separate variables $\vy_\spar$ and $\vy_\perp$,
we can swap convolution with multiplication to obtain
\begin{align}
\hxi_{H+\vv}(\vy) &\propto \big(\exp(-\pi S^2 \norm{\vy_\spar}_2^2) *
    (\hdelta_{H^\circ}(\vy_\spar)\exp(-2 \pi i \vv \cdot \vy_\spar))
    \big) \nonumber \\
    &\quad \big( \exp(-\pi S^2 \norm{\vy_\perp}_2^2)
    * \exp(-2 \pi i \vv \cdot \vy_\perp) \big), \label{e:sepvar}
\end{align}
where by abuse of notation, we write
\[ (\psi_1 * \psi_2)(\vy) = \psi_1(\vy) * \psi_2(\vy). \]
The second factor in \eqref{e:sepvar} is proportional to $\exp(-2 \pi i
\vv \cdot \vy_\perp)$, while $\hdelta_{H^\circ}$ is a sum of delta functions
on the lattice $H^\circ$.  Thus, we obtain
\begin{align*}
\hxi_{H+\vv}(\vy) &\propto
    \exp(-2 \pi i \vv \cdot \vy_\perp) \\
    &\quad \sum_{\vz \in H^\circ} \exp(-\pi S^2 \norm{\vy_\spar-\vz}_2^2
    -2\pi i \vv\cdot \vz),
\end{align*}
as desired. \end{proof}

The function $\tau(\vy)$ in \Lem{l:xihat} descends to a function
\[ \tau:(\R/\Z)^k \to \C \]
since $\xi_{H+\vv}$ is supported on $\Z^k$, and since $\hxi_{H+\vv}$ is
the lift of $\hpsi_{H+\vv}$ from $(\R/\Z)^k$ to $\R^k$.  More explicitly,
if $\vy \in (\R/\Z)^k$ and $\tvy \in \R^k$ is any lift, then we define
\[ \tau(\vy) \defeq \tau(\tvy). \]

\begin{lemma} If $\vy \in (\R/\Z)^k$ and $\vy_0 \in H^\#$
is a closest point to $\vy$, then uniformly when $S \ge k\Delta^2$,
\begin{gather}
\abs{\tau(\vy)} = \exp(-\pi S^2 \norm{\vy-\vy_0}_2^2)
    + \exp(-\Omega(S)) \label{e:tauabs} \\
\int_{(\R/\Z)^k} \abs{\tau(\vy)}^2\;d\vy = 2^{(\ell-k)/2}S^{\ell-k}
    \sqrt{\Delta} + \exp(-\Omega(S)). \label{e:taunorm}
\end{gather}
\eatline \label{l:tau} \end{lemma}

\begin{proof} To prove \equ{e:tauabs}, we work upstairs with $\vy \in
\R^k$ and $\vy_0 \in H^\bullet$.  We can factor $\tau(\vy)$ by separation
of variables:
\begin{align}
\tau_\spar(\vy_\spar) &\defeq
    \sum_{\vz \in H^\circ} \exp(-\pi S^2 \norm{\vy_\spar-\vz}_2^2
    -2\pi i \vv\cdot \vz) \label{e:tauspar} \\
\tau_\perp(\vy_\perp) &\defeq \exp(-2 \pi i \vv \cdot \vy_\perp) \nonumber \\
\tau(\vy) &= \tau_\spar(\vy_\spar)\tau_\perp(\vy_\perp) \nonumber \\
\abs{\tau(\vy)} &= \abs{\tau_\spar(\vy_\spar)}. \nonumber
\end{align}
Since the desired quantity equals $\abs{\tau_\spar(\vy_\spar)}$, we
analyze the function $\tau_\spar$.   By translational invariance, we can
assume that the closest point (or tied for closest) is $\vy_0 = \vec0$,
which then implies that $\vy = \vy_\perp$.  With these simplifications,
\equ{e:tauabs} asserts that
\begin{align*}
\abs{\tau_\spar(\vy_\spar)} &=
    \exp(-\pi S^2 \norm{\vy_\spar}_2^2) + \exp(-\Omega(S^2/\Delta^2)) \\
    &= \exp(-\pi S^2 \norm{\vy_\spar}_2^2) + \exp(-\Omega(S)).
\end{align*}
Here and below, we replace $\Omega(S^2/\Delta^2)$ by $\Omega(S)$
in the error, which we can do given that $S \ge \Delta^2$.

We consider two cases, depending on the value of $\norm{\vy_\spar}_\infty$.
If $\norm{\vy_\spar}_\infty > 1/2\Delta$, then first note
that the error term swallows the other term:
\[  \exp(-\pi S^2 \norm{\vy_\spar}_2^2) = \exp(-\Omega(S)). \]
We apply the triangle inequality to \equ{e:tauspar} to obtain
\begin{eq}{e:triang1} \abs{\tau_\spar(\vy_\spar)} \le \sum_{\vz \in H^\circ}
    \exp(-\pi S^2 \norm{\vy_\spar-\vz}_2^2) \end{eq}
Recall from \Lem{l:feature} that $H^\circ \subseteq \frac1{\Delta}\Z^k$.  We
can thus substitute $\vx = \Delta\vy_\perp$ and compare \eqref{e:triang1}
to \equ{e:glattice} to obtain
\begin{eq}{e:tautail} \abs{\tau_\spar(\vy_\spar)} = \exp(-\Omega(S)),
\end{eq}
as desired.  If $\norm{\vy_\spar}_\infty \le 1/2\Delta$, then
we apply the triangle inequality to \equ{e:tauspar} after
subtracting the $\vz = \vec0$ term to obtain
\[ \abs{\tau_\spar(\vy_\spar) -
    \exp(-\pi S^2 \norm{\vy_\spar}_2^2)} \le \back\hback
    \sum_{\vz \in H^\circ \setminus \{\vec0\}} \back\hback
    \exp(-\pi S^2 \norm{\vy_\spar-\vz}_2^2). \]
We can again substitute $\vx = \Delta\vy_\perp$ and compare \eqref{e:triang1}
to \equ{e:glattice} to obtain
\[ \abs{\tau_\spar(\vy_\spar) - \exp(-\pi S^2 \norm{\vy_\spar}_2^2)}
    = \exp(-\Omega(S)), \]
as desired.

It remains to prove \equ{e:taunorm}.  Given $\delta > 0$, let
\[ H_{\R,\delta} \defeq H_\R \cap [-\delta,\delta]^k \]
We first approximate
\begin{eq}{e:hgbox}
\int_{H_{\R,1/2\Delta}} \back\back\back
    \exp(-2\pi S^2 \norm{\vy}_2^2)\;d\vy \\
    = 2^{(\ell-k)/2}S^{\ell-k} + \exp(-\Omega(S)). \end{eq}
To prove \equ{e:hgbox}, we start with its exact counterpart
\[ \int_{H_\R} \back \exp(-\pi 2S^2 \norm{\vy}_2^2)\;d\vy
    = 2^{(\ell-k)/2}S^{\ell-k}. \]
We then establish the approximation \eqref{e:hgbox} using the same proof as
the proof of \equ{e:gtail} in \Lem{l:gaussest}, skipping some details.
The approximation follows by cutting off half spaces from $H_\R$ to make
$H_{\R,1/2\Delta}$.  Each half space integral is small,
\[ \int_{H_\R,\ \abs{y_j} \ge 1/2\Delta} \back\back
    \exp(-\pi 2S^2 \norm{\vy}_2^2)\;d\vy
    = \exp(-\Omega(S)), \]
and the estimate follows from the union bound.  To finish the proof of
\equ{e:taunorm}, we can interpret the set $H_{\R,1/2\Delta}$ as a subset of
$(\R/\Z)^k$.  It is perpendicular to $H^\#$; and again since $H^\circ
\subseteq \frac1{\Delta}\Z^k$, every point
\[ \vy \in H^\# + H_{\R,1/2\Delta} \subseteq (\R/\Z)^k \]
is uniquely expressible as
\[ \vy = \vy_\perp + \vy_\spar \qquad
    \vy_\perp \in H^\#,\ \vy_\spar \in H_{\R,1/2\Delta}. \]
Recall from \Lem{l:feature} that $\Vol H^\# = \sqrt{\Delta}$.  We thus
obtain:
\[ \int_{H^\# + H_{\R,1/2\Delta}} \back\back\back |\tau(\vy)|^2\;d\vy
    = 2^{(\ell-k)/2}S^{\ell-k}\sqrt{\Delta} + \exp(-\Omega(S)). \]
\Equ{e:tautail} also tells us that
\[ |\tau(\vy)|^2 = \exp(-\Omega(S)) \]
on the complement
\[ \vy \in (\R/\Z)^k \setminus (H_{\R,1/2\Delta} + H^\#). \]
Since $(\R/\Z)^k$ has unit volume, \equ{e:taunorm} follows.
\end{proof}

\Lem{l:tau} implies that the normalized state $\ket{\hpsi_{H+\vv}}$
is given by
\begin{multline*}
\hpsi_{H+\vv}(\vy) = 2^{\ell/4}S^{\ell/2}\Delta^{-1/4}
    \exp(-\pi S^2 \norm{\vy-\vy_0}_2^2) \\ + \exp(-\Omega(S)).
\end{multline*}
Using \Lem{l:tau} a second time, we obtain:
\begin{multline}
\abs{\hpsi_{H+\vv}(\vy)}^2 = 2^{\ell/2}S^{\ell}\Delta^{-1/2}
    \exp(-\pi 2 S^2 \norm{\vy-\vy_0}_2^2) \\ + \exp(-\Omega(S)).
\label{e:u2noise} \end{multline}
Setting $\vy_2 = \vy$, \equ{e:u2noise} says that
\[ \mu_1(\vy_2) \defeq \abs{\hpsi_{H+\vv}(\vy_2)}^2 \]
is the probability density of a random variable $\vy_2$ which consists of
a random $\vy_0 \in H^\#$ plus an orthogonal, approximately Gaussian error
\[ \vu_2 \defeq \vy_2-\vy_0. \]

To complete our use of \Lem{l:pffpgc}, we apply the operator $\hP$ to the
state $\ket{\hpsi_{H+\vv}}$.  The result is the same function
\[ (\hP\hpsi_{H+\vv})(\vy_1) = \hpsi_{H+\vv}(\vy_1), \]
except restricted to $\vy_1 \in q(\Z/Q)^k$.  In particular,
the measurement of $q^{k/2}\hP\ket{\hpsi_{H+\vv}}$ has the distribution
\[ \mu_2 \defeq q^{k/2}\mu_1|_{q(\Z/Q)^k}. \]
(Note that $q^{k/2}\hP\ket{\hpsi_{H+\vv}}$ and $\mu_2$ are only approximately
normalized.  We can still prove probability estimates despite this.)

In pursuit of \Lem{l:unoise}, we compare $\mu_2$ to another measure
$\mu_3$ produced by rasterizing a vector $\vy_2 \in (\R/\Z)^k$ drawn
with density $\mu_1$.   Let $\vy_1$ be the closest point to $\vy_2$
$\in q(\Z/Q)^k$ (ignoring the set of measure zero with more
than one closest point).  If we set
\[ \vu_1 \defeq \vy_1 - \vy_2, \]
then evidently
\[ \norm{\vu_1}_2 \le \sqrt{k}\norm{\vu_1}_\infty \le \frac{\sqrt{k}q}2. \]
Let $\mu_3(\vy_1)$ be the distribution of $\vy_1$ given by this process.

\begin{lemma} Under the hypotheses of \Lem{l:unoise}, if $\mu_2$ and $\mu_3$
are measures on $q(\Z/Q)^k$ defined as above, then
\[ \norm{\mu_3 - \mu_2}_1 = O(s^2). \]
\eatline \label{l:yquant} \end{lemma}

We prove \Lem{l:yquant} with a general lemma about the accuracy of the
midpoint rule for integrals.  If $\phi:A \to \R$ is a function on $\R^k$ with
bounded second derivatives, then its second derivative matrix (or Hessian)
$\vnabla\vnabla\phi$ has a mixed norm that can be defined as follows:
\begin{eq}{e:hessnorm}
\norm{\vnabla\vnabla\phi}_{2\to 2,\infty}
    \defeq \sup_{\vx \in A,\norm{\vt}_2 = 1}
    \abs{\vt\cdot (\vnabla\vnabla\phi)(\vx)\cdot \vt}.
\end{eq}

\begin{lemma} If
\[ \phi:[-q/2,q/2]^k \to \R \]
is $C^2$, then
\[ \int_{[-q/2,q/2]^k} \phi(\vx)\;d\vx = q^k\phi(\vec0) + O(kq^{k+2}). \]
\label{l:midpoint} \end{lemma}

\begin{proof} We can replace $\phi(\vx)$ by
\[ \phi(\vx) - \phi(\vec0) - \vx \cdot (\vnabla \phi)(\vx) \]
to assume that $\phi(\vec0)$ and $\vnabla\phi(\vec0)$ both vanish.
Integrating in a straight line from $\vec0$ to $\vx$, \equ{e:hessnorm}
then tells us that
\[ \abs{\phi(\vx)} \le \frac12\norm{\vnabla\vnabla\phi}_{2\to 2,\infty}
    \norm{\vx}^2_2. \]
The conclusion then follows from the integral
\[ \int_{[-q/2,q/2]^k} \norm{\vx}^2_2\;d\vx
    = \frac{kq^{k+2}}{12}. \qedhere \]
\end{proof}

\begin{lemma} The Hessian of an isotropic Gaussian on $\R^k$ has norm
\[ \big\lVert\vnabla\vnabla\exp(-\pi S^2
    \norm{\vx}_2^2)\big\rVert_{2 \to 2} = 2\pi S^2. \]
\eatline \label{l:ghess} \end{lemma}

\begin{proof} We can rescale $\vx$ to assume that $S = 1/\sqrt{\pi}$. The
1-dimensional case $k=1$ is a routine calculus problem that the absolute
value of the second derivative of the Gaussian,
\[ \abs{\exp(-x^2)''} = \abs{4x^2-2}\exp(-x^2), \]
is maximized at $x=0$.  (Fun exercise: The same is true of any even
derivative of a univariate Gaussian.)  In dimension $k>1$, the maximand
of \equ{e:hessnorm}, which in this case is
\[ \big\lvert\vt\cdot (\vnabla\vnabla \exp(-\norm{\vx}_2^2))
    \cdot \vt\big\rvert, \]
can be interpreted as part of a 1-dimensional restriction of
$\exp(-\norm{\vx}_2^2)$ to the straight line $L$ that passes through $\vx$
and is tangent to $\vt$.  The restriction to this line $L$ is a univariate
Gaussian that is equivalent to $\exp(-x^2)$, except multiplied by a constant
$0 < c < 1$ when $L$ does not pass through the origin.  If $L$ does pass
through the origin, then maximization along $L$ is exactly the $k=1$ case.
\end{proof}

\begin{proof}[Proof of \Lem{l:yquant}] We apply \Lem{l:midpoint}
to the shifted cube $\vy_1 + [-q/2,q/2]^k$ together with the
\Lem{l:ghess} in $\ell$ dimensions (parallel to $H_\R$)
to the non-error term
\[ 2^{\ell/2}S^{\ell}\Delta^{-1/2} \exp(-\pi 2 S^2 \norm{\vy-\vy_0}_2^2) \]
in \equ{e:u2noise}.  This yields the estimate
\[ \mu_3(\vy_1) = \mu_2(\vy_1) + O(q^{k+2}S^2) + q^k\exp(-\Omega(S)). \]
Summing over $(\Z/Q)^k$, we obtain
\[ \norm{\mu_3-\mu_2}_1 = O(q^2S^2) + \exp(-\Omega(S)) = O(S^{-2}), \]
as desired.
\end{proof}

\begin{lemma} If $\vu \in \R^\ell$ is a random variable
with density proportional to $\exp(-\pi S^2 \norm{\vu}_2^2)$, then
\[ \Pr[\norm{\vu}_2 \ge s_0] \le \frac{\ell}{2\pi S^2 s_0^2}\]
\eatline \label{l:ugtail} \end{lemma}

\begin{proof} Each coordinate $u_j$ is an independent random variable with
density proportional to $\exp(-\pi S^2 u_j^2)$.  Thus,
\[ \Ex[u_j^2] = \frac{1}{2 \pi S^2} \qquad
    \Ex[\norm{\vu}_2^2] = \frac{\ell}{2 \pi S^2}. \]
Since $\norm{\vu}_2^2$ is a non-negative random variable, the lemma
follows from Markov's inequality.
\end{proof}

To conclude the proof of \Lem{l:unoise}, let
\[ \vu = \vy_1 - \vy_0 = \vu_1 + \vu_2 \]
be the total error in $\vy_1$, where $\vy_0 \in H^\#$ is uniformly random,
$\vu_2$ has distribution \eqref{e:u2noise}, and $\vu_1$ is the rasterization
error described by \Lem{l:yquant}.  Then
\begin{eq}{e:u1u2noise} \norm{\vu_1}_2 = O(\sqrt{k}q)
    = O(\sqrt{k}s^2). \end{eq}
We can apply \Lem{l:ugtail} to $\vu_2$ with
\[ s_0 = \sqrt{k}s + \norm{\vu_1}_2 \]
to obtain
\[ \Pr\big[\norm{\vu_2}_2 \ge \sqrt{k}s+ \norm{\vu_1}_2\big]
    \le \frac{1}{2\pi} + O(s). \]
Then by the triangle inequality,
\begin{eq}{e:utail}
\Pr\big[\norm{\vu}_2 \ge \sqrt{k}s \big] \le \frac{1}{2\pi} + O(s).
\end{eq}
Meanwhile, the difference between this probabilistic model and the true
result of steps 1--3 of \Alg{a:abelian} is given by the error terms in
Lemmas~\ref{l:pffpgc}, \ref{l:tau}, and \ref{l:yquant}; the total is
\[ \exp(-\Omega(S^2)) + \exp(-\Omega(S)) + O(s^2) = O(s^2). \]
Thus, \equ{e:utail} also holds for the true distribution produced in
\Alg{a:abelian}, and \Lem{l:unoise} follows.

\subsection{Dense orbits}
\label{ss:dense}

Recall from \Sec{ss:amotive} that $\tvy \in (-1/2,1/2]^k$ is the canonical
lift of $\vy \in (\R/\Z)^k$, and that $\norm{\vy}_2 = \norm{\tvy}_2$.

As mentioned in \Sec{ss:amotive}, the set of positive multiples (or the
forward orbit) $\{\lambda\vy_0\}_{\lambda \in \Z_+}$ of a randomly chosen
\[ \vy_0 \in H^\# \subseteq (\R/\Z)^k \]
is dense in $H_1^\#$ with probability 1.  Thus, for any $\epsilon >
0$, the orbit includes $k-\ell$ points in the ball $B_{\eps,2}(\vec0)$
that lift to a basis of $H_\R^\perp \subseteq \R^k$.  The discrete orbit
$\{(\lambda\vy_0,\lambda)\}_{\lambda \in \Z_+}$ also includes a subset
that lies in $B_{\eps,2}(\vec0) \oplus \R$ and that lifts to a basis
\[ (\widetilde{\lambda_1\vy_0},\lambda_1),
    (\widetilde{\lambda_2\vy_0},\lambda_2), \ldots,
    (\widetilde{\lambda_{k+1-\ell}\vy_0},\lambda_{k+1-\ell}) \]
of $H_\R^\perp \oplus \R \subseteq \R^{k+1}$.  The goal of this subsection is
to prove an effective version of this standard result, moreover one that is
stable when we add noise to $\vy_0$.  Recall the Gram determinant $\Delta$
defined by \equ{e:gram}.  We also assume for convenience that $\ell < k$,
since $\ell = k$ is the case handled by the Shor--Kitaev algorithm.

\begin{lemma} Let $C > 0$ be a suitable absolute constant, let $R_1 \ge
\max(\Delta,2^k,C)$ be an integer, let $\Lambda = R_1^{Ck}$, let $\vy_0
\in H^\# \subseteq (\R/\Z)^k$ be uniformly random, and let $\vy_1 \in
(\R/\Z)^k$ be any vector such that
\[ \norm{\vy_1 - \vy_0}_2 \le \frac{1}\Lambda. \]
Then with probability at least $3/4$, there are integers
\[ 1 \le \lambda_1, \lambda_2, \ldots, \lambda_{k+1-\ell} \le \Lambda \]
such that $\norm{\widetilde{\lambda_j\vy_1}} \le \frac1{R_1}$
for every $j$, and such that
\[ (\widetilde{\lambda_1\vy_1},\lambda_1),
    (\widetilde{\lambda_2\vy_1},\lambda_2), \ldots,
    (\widetilde{\lambda_{k+1-\ell}\vy_1},\lambda_{k+1-\ell}) \]
is a linearly independent set in $\R^{k+1}$.
\label{l:basis} \end{lemma}

As a first step in the proof of \Lem{l:basis}, recall that $H^\#/H^\#_1$ is a
finite group.  If $\vy_0$ is uniformly random in $H^\#$, then $[H_1:H]\vy_0$
is uniformly random in $H_1^\#$.  Let $\Delta_1$ be the Gram determinant
of $H_1$.  Note that
\begin{eq}{e:h1bound} R_1 \ge \Delta \ge \Delta_1 \ge 1 \qquad
    [H_1:H] = \sqrt{\Delta/\Delta_1} \le \sqrt{\Delta}, \end{eq}
and that $C$ is arbitrary and can be increased.  We can thus boost $\Lambda$
by a factor of $[H_1:H]$ and reduce the lemma to the case when $H = H_1$.
Equivalently, we can assume that $H^\#$ is connected.

The following lemma is more precise than necessary for our purposes,
but the author enjoyed this version as an exercise in geometry and
linear algebra.

\begin{lemma} Let
\[ \vv_1,\vv_2,\ldots,\vv_{\ell+1} \in \R^\ell \]
be the vertices of a regular simplex centered at $\vec0$ with circumradius 1.
For each $j$, let $\vw_j \in B_{1/\ell}(\vv_j)$, and let $\lambda_j > 0$.
Then
\[ (\vw_1,\lambda_1),(\vw_2,\lambda_2),\ldots,(\vw_{\ell+1},\lambda_\ell) \]
is a basis of $\R^{\ell+1}$.  Moreover, the radius $1/\ell$ is the largest
possible to guarantee that these vectors are a basis.
\label{l:simplex} \end{lemma}

\begin{figure}
\begin{tikzpicture}[thick,scale=.8]
\draw[dashed,darkgreen,fill=lightgreen] (90:2) circle (1);
\draw[dashed,darkgreen,fill=lightgreen] (210:2) circle (1);
\draw[dashed,darkgreen,fill=lightgreen] (330:2) circle (1);
\draw[darkblue] (90:2) -- (210:2) -- (330:2) -- cycle;
\fill (330:2) circle (.075) node[below] {$\vv_1$};
\fill (330:2) +(.1,.35) circle (.075) node[above] {$\vw_1$};
\fill (90:2) circle (.075) node[above] {$\vv_2$};
\fill (90:2) +(-.6,0) circle (.075) node[below] {$\vw_2$};
\fill (210:2) circle (.075) node[below] {$\vv_3$};
\fill (210:2) +(.1,.5) circle (.075) node[left] {$\vw_3$};
\fill (0,0) circle (.075); \draw (240:.35) node {$\vec0$};
\draw (330:2) -- +(30:1); \draw (330:2) +(0:.6) node {$\frac12$};
\draw (0,0) -- (90:2); \draw (.2,.7) node {$1$};
\draw (300:3) -- (120:3) node[above] {$Z$};
\draw[->] (0,0) -- (30:2); \draw (20:1.5) node {$\vz$};
\draw (-3.5,2) node {(a)};
\tikzset{yshift=-6cm};
\draw[darkblue] (30:4) -- (120:2.309) -- (210:4) -- (300:2.309) -- cycle;
\draw[very thick,darkgreen] (0,0) circle (2);
\draw[<->] (270:2) -- (90:2); \draw (102:1.25) node {$\vv_2$};
\draw (250:1.25) node {$-\vv_2$};
\draw[<->] (330:2) -- (150:2); \draw (318:1.25) node {$\vv_1$};
\draw (165:1.35) node {$-\vv_1$};
\draw[->] (0,0) -- (30:2); \draw (20:1.5) node {$\vz$};
\fill (0,0) circle (.075); \draw (210:.35) node {$\vec0$};
\fill (30:4) circle (.075); \draw[darkred] (30:4) circle (.2);
\draw (-3.25,-1) node {$X$};
\draw (-3.5,2) node {(b)};
\end{tikzpicture}
\caption{Two steps in the proof of \Lem{l:simplex}.  Diagram (a) illustrates
    in dimension $\ell=2$ that no line $Z$ through $\vec0$ contains two of
    the three vectors $\vw_1$, $\vw_2$, $\vw_3$.  Diagram (b) illustrates
    that in any dimension, the worst-case hyperplane $Z$ is perpendicular to
    a $\vz$ that points toward a furthest corner of the parallelepiped $X$.}
\label{f:simplex} \end{figure}
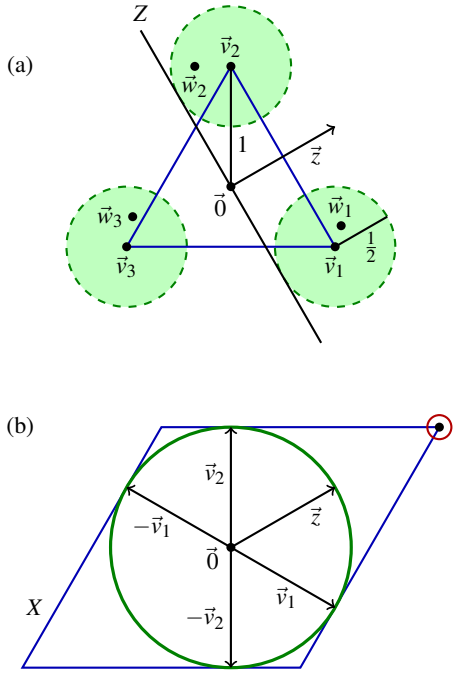

\begin{proof} Let $r > 0$ and suppose that each $\vw_j \in B_r(\vv_j)$.
We also order the vectors so that $\vv_1,\vv_2,\ldots,\vv_\ell$ is a positive
basis of $\R^\ell$.  We first consider the special case in which $\vw_j =
\vv_j$ for all $j$.   In this case, each of the $\ell+1$ terms
in the cofactor expansion of the determinant
\[ \det \begin{bmatrix} \vv_1 & \vv_2 & \cdots & \vv_{\ell+1} \\
    \lambda_1 & \lambda_2 & \cdots & \lambda_{\ell+1} \end{bmatrix} \]
is a positive number, so the determinant is positive and the
columns are a positive basis.  If the determinant vanished when we
replace each $\vv_j$ by $\vw_j$, then at least one of the cofactors would
have to be non-positive, without loss of generality the first cofactor
\[ \det \begin{bmatrix} \vw_1 & \vw_2 & \cdots & \vw_\ell \end{bmatrix}. \]
By continuity and the intermediate value theorem, if each $\vw_j \in
B_r(\vv_j)$ and this cofactor can take a negative value, then it can also
vanish.  To prove the lemma, it suffices to show that $r$
is small enough that the cofactors cannot vanish, equivalently that the
vectors $\vw_1,\vw_2,\ldots,\vw_\ell$ cannot lie in a hyperplane $Z \subseteq
\R^\ell$ that passes through $\vec0$.

Conversely, if $\vw_1,\ldots,\vw_\ell$ can lie in a hyperplane $Z \ni
\vec0$, then they can be moved to the same side of $Z$ as $\vw_{\ell+1}$,
since the vectors are constrained by open conditions.   In this case,
if $\vz$ is normal to $Z$ and on the same side, then we can let $\lambda_j =
\vz \cdot \vw_j$.  Then
\[ (\vz,-1) \cdot (\vw_j,\lambda_j) = 0 \]
for all $j$, so that the vectors $\{(\vw_j,\lambda_j)\}$ are not a basis of
$\R^{\ell+1}$.   In conclusion, the lemma is equivalent to finding the
largest $r$ such that $\vw_1,\ldots,\vw_\ell$ cannot lie in a hyperplane
$Z \ni \vec0$ when each $\vw_j \in B_r(\vv_j)$.   This is equivalent to
the smallest $r$ such that $\vw_1,\ldots,\vw_\ell$ can lie in $Z$ when
each $\vw_j$ is in the closed ball $\overline{B_r(\vv_j)}$.  The lemma
asserts that $r = 1/\ell$ is the optimal value.

At this point, we switch from optimizing over $r$ to optimizing over the
hyperplane $Z$ which is defined by a unit normal $\vz$.  Given $\vz$,
the least value of $r$ such that each $\overline{B_r(\vv_j)}$ reaches $Z$ is
\[ r = \max_{j \le \ell} \abs{\vz \cdot \vv_j}. \]
The optimum $\vz$ is one that points toward a furthest corner of the
polytope $X$ defined by the equations $\abs{\vz \cdot \vv_j} \le 1$
for all $j \le \ell$.  $X$ is a cube stretched along its long diagonal
parallel to $\vv_{\ell+1}$, so $\vz = \pm \vv_{\ell+1}$.  Thus
\[ \abs{\vv_{\ell+1} \cdot \vv_j} = \frac1\ell \]
is the optimal value of $r$, as desired.

\Fig{f:simplex} illustrates two steps in the argument.
\end{proof}

Now let
\[ \Lambda \defeq 2R_1(k+1)\Lambda_1 \qquad \Lambda = R_1^{C_1k}, \]
where $C_1$ (and therefore $\Lambda_1$) is an integer which is not yet
specified, and let
\[ r_1 \defeq \frac{1}{R_1} \qquad r_2 \defeq \frac{r_1}{2(k+1)}
    \qquad Y_0 \defeq \{\lambda\vy_0\}_{1 \le \lambda \le \Lambda_1}. \]

\begin{corollary}[of \Lem{l:simplex}]  If $\vy_0$ is chosen so that its
orbit $Y$ intersects every ball $B_{r_2}(\vv)$ with $\vv \in H^\#$, then
the conclusions of \Lem{l:basis} hold.
\label{c:simplex} \end{corollary}

\begin{proof} Let
\[ \tY_0 \defeq \{\widetilde{\lambda\vy_0}\}_{1 \le \lambda \le \Lambda_1}
    \subseteq H_\R^\perp \subseteq \R^k \]
be the lift of $Y_0$ from $(\R/\Z)^k$
to $\R^k$.  We will also need the orbit
\[ \tY_1 \defeq \{\widetilde{\lambda\vy_1}^\perp\}_{1
    \le \lambda \le \Lambda_1} \subseteq H_\R^\perp, \]
where $\widetilde{\lambda\vy_1}^\perp$ is the orthogonal projection of
$\widetilde{\lambda\vy_1}$ to $H_\R^\perp$.  By hypothesis,
\[ \norm{\vy_1 - \vy_0}_2 \le \frac{1}\Lambda
    = \frac1{2R_1(k+1)\Lambda_1}. \]
Therefore when $1 \le \lambda \le \Lambda_1$,
\[ \norm{\widetilde{\lambda\vy_1} - \widetilde{\lambda\vy_0}}_2 \le r_2, \]
except when $\lambda\vy_0$ is near the boundary of the cube $[-1/2,1/2]^k$.
Finally $\widetilde{\lambda\vy_0}$ lies in $H^\perp \oplus \R$, while
$\widetilde{\lambda\vy_1}$ only gets closer if we project it to $H^\perp
\oplus \R$.

We apply \Lem{l:simplex} by substituting $k-\ell$ for $\ell$ and $H_\R^\perp$
for $\R^{k-\ell}$.  In addition, we rescale the geometry in \Lem{l:simplex}
by a factor of
\begin{eq}{e:rescale} \frac{k-\ell}{k+1-\ell}r_1. \end{eq}
Before rescaling, the vectors $\vv_j$ listed in \Lem{l:simplex}
had norm 1, so that after rescaling, they have norm given by
\eqref{e:rescale}.  If $\tY_0$ intersects each ball $B_{r_2}(\vv_j)$
with $\vv_j \in B_{r_1}(\vec0) \subseteq H_\R^\perp$, then $\tY_1$
intersects each ball $B_{2r_2}(\vv_j)$ by the triangle inequality.
Given the value of $r_2$, \Lem{l:simplex} then implies that
$\{\widetilde{(\lambda_j\vy_1}^\perp,\lambda_j)\}$ is a basis of $H_\R^\perp
\oplus \R$. Therefore $\{\widetilde{(\lambda_j\vy_1},\lambda_j)\}$
is a linearly independent set, since a linear projection of a linearly
dependent set would also have been linearly dependent.  We also learn from
the triangle inequality that
\[ \norm{\widetilde{\lambda_j\vy_1}-\vv_j}_2 \le r_2
    + \norm{\widetilde{\lambda_j\vy_0}-\vv_j}+2 \le 2r_2, \]
and thus that
\[ \norm{\widetilde{\lambda_j\vy_1}}_2 \le 2r_2
    + \frac{k-\ell}{k+1-\ell}r_1 \le r_1. \qedhere \]
\end{proof}

To conclude the proof of \Lem{l:basis}, we will prove that the hypothesis
of \Cor{c:simplex} (that $Y_0$ intersects every ball $B_{r_2}(\vv)$)
holds with probability at least 3/4 when $\Lambda$ and $\Lambda_1$ are
large enough but not too large, $\vy_0$ is randomly chosen, and $\vy_1$ is
adversarially chosen.  More precisely, we will give a proof by contradiction
that the hypothesis of \Cor{c:simplex} probably holds if $C_1$ is a large
enough constant.  Given that $\Lambda = 2R_1(k+1)\Lambda_1$ and that $R_1
\ge 2^k$, we can then set $\Lambda = R_1^{Ck}$ with a constant $C > C_1$.

Let
\[ r_3 = r_1^4 \le r_2^2 \qquad R_3 = \frac{1}{r_3}, \]
and fix some $\vv \in H^\#$.  Consider the function
\begin{eq}{e:sigdef} \sigma(\vy) \defeq \sum_{\vz \in \Z^k}
    \exp(-\pi R_3^2\norm{\tvy-\tvv+\vz}_2^2), \end{eq}
where we choose some lifts $\tvv,\tvy \in \R^k$ of $\vv,\vy \in (\R/\Z)^k$.
The specific choices of these lifts do not affect the value of $\sigma(\vy)$;
but note that $\sigma$ implicitly depends on the fixed choice of $\vv$.

Henceforth let $Y = Y_0$, and let $\delta_Y$ be the TD with a unit measure
at each point in $Y$.  Using a finitary version of Weyl's criterion for
equidistribution of orbits, we will show that $Y$ contains a point in
$B_{r_2}(\vv)$.  Informally, we will show that the average of $\sigma$
over $Y$ is close to its average over $H^\#$:
\begin{eq}{e:wapprox} \frac{\braket{\delta_Y|\sigma}}{\Lambda} \sim
    \frac{\braket{\delta_{H^\#}|\sigma}}{\sqrt{\Delta}}. \end{eq}
We establish \Lem{l:basis} with the following three bounds
associated with \equ{e:wapprox}:
\begin{enumerate}
\item An upper bound on the left side of \equ{e:wapprox} when
$Y$ is disjoint from $B_{r_2}(\vv)$.
\item A lower bound on the right side of \equ{e:wapprox}.
\item An upper bound on the difference between the left and right
sides of \equ{e:wapprox} that holds with probability at least 3/4
and that makes the bounds on the two sides conflict.
\end{enumerate}

\Sec{ss:fourier} already has the main work for the upper bound on the left
side of \equ{e:wapprox}.  Suppose that $\vy \notin B_{r_2}(\vv)$, and
choose lifts $\tvv$ and $\tvy$ such that $\tvy-\tvv \in [-1/2,1/2]^k$.
Then we can apply \equ{e:glattice} to \equ{e:sigdef}, with $\tvy-\tvv$
substituted for $\vy$ and $R_3$ substituted for $S$.  We obtain:
\[ \sigma(\vy) \le \exp(-\pi r_2^2/r_3^2)
    + \exp(-\Omega(R_3^2))
    = \exp(-\Omega(R_3)). \]
If we average this estimate over $Y$, we obtain
\begin{eq}{e:wupper} \frac{\braket{\delta_Y|\sigma}}{\Lambda_1}
    = \exp(-\Omega(R_3)). \end{eq}
To obtain a lower bound on the right side of \equ{e:wapprox}, we calculate
\begin{align}
\frac{\braket{\delta_{H^\#}|\sigma}}{\sqrt{\Delta}} &=
    \int_{H^\#} \frac{\sigma(\vy)}{\sqrt{\Delta}} \;d\vy \nonumber \\
    &\ge \int_{H_\R^\perp} \frac{\exp(-\pi R_3^2 \norm{\vy-\vv}_2^2)}
    {\sqrt{\Delta}}\;d\vy
    = \frac{r_3^{k-\ell}}{\sqrt{\Delta}}. \label{e:wlower}
\end{align}

The hardest step, and the last one remaining, is to bound the difference
between the two sides of \eqref{e:wapprox} with good probability.  We first
apply a Fourier transform to each function and TD:
\begin{eq}{e:pairings} \braket{\delta_Y|\sigma} = \braket{\hdelta_Y|\htau}
    \qquad \braket{\delta_{H^\#}|\sigma} = \braket{\hdelta_{H^\#}|\htau}.
\end{eq}
By \equ{e:gsum} and \Prop{p:prodconv},
\[ \htau(\vx) = r_3^k \exp(-\pi r_3^2 \norm{\vx}_2^2 -2\pi i\vx \cdot \vv) \]
for $\vx \in \Z^k$.  Then by equations~\eqref{e:gsum} and
\eqref{e:glattice},
\begin{align}
\norm{\htau}_1
    &= \sum_{\vx \in \Z^k} r_3^k \exp(-\pi r_3^2 \norm{\vx}_2^2) \nonumber \\
    &= \sum_{\vz \in \Z^k} \exp(-\pi R_3^2 \norm{\vz}_2^2) \nonumber \\
    &= 1 + \exp(-\Omega(R_3^2)). \label{e:hgnorm}
\end{align}
Meanwhile,
\begin{align}
\frac{\hdelta_{H^\#}(\vx)}{\sqrt{\Delta}} &= \delta_H(\vx)
    = \begin{cases} 1 & \text{if $\vx \in H$} \\
    0 & \text{if $\vx \notin H$} \end{cases} \nonumber \\
\frac{\hdelta_Y(\vx)}{\Lambda_1}
    &= \sum_{\lambda=1}^{\Lambda_1} \frac{\exp(-2\pi i
        \lambda \vx \cdot \vy_2)}{\Lambda_1}. \label{e:hdely}
\end{align}
We learn from this that $\hdelta_Y/\Lambda_1$ and
$\hdelta_{H^\#}/\sqrt{\Delta}$ agree on $H \subseteq \Z^k$, which is also
the support of $\hdelta_{H^\#}$.  We also learn that
\[ \frac{\abs{\hdelta_Y(\vx)}}{\Lambda_1} \le 1 \]
for all $\vx \in \Z^k$.

To establish an upper bound between the sides of \eqref{e:wapprox}, the
comparison of $\hdelta_Y$ and $\hdelta_{H^\#}$ yields
\begin{eq}{e:weylerr}
\left\lvert \frac{\braket{\hdelta_Y|\htau}}{\Lambda_1} -
\frac{\braket{\hdelta_{H^\#}|\htau}}{\sqrt{\Delta}} \right\rvert
\le \sum_{\substack{\vx \in \Z^k \setminus H}}
\back \frac{\abs{\hdelta_Y(\vx)}\;\abs{\htau}}{\Lambda_1}. \end{eq}
We split the sum in \equ{e:weylerr} into two parts, according to whether
$\norm{\vx}_\infty \ge R_3^2/2$ or $\norm{\vx}_\infty < R_3^2/2$.  The first
part can be bounded by applying \equ{e:gtail}:
\begin{eq}{e:xlarge} \sum_{\substack{\vx \in \Z^k \setminus H \\
    \norm{\vx}_\infty \ge R_3^2/2}} \back \back
    \frac{\abs{\hdelta_Y(\vx)}\;\abs{\htau(\vx)}}{\Lambda_1} \le
    \back \sum_{\substack{\vx \in \Z^k \setminus H \\
    \norm{\vx}_\infty \ge R_3^2/2}}
    \back \abs{\htau(\vx)} = \exp(-\Omega(R_3^2)). \end{eq}

Bounding the other part of \equ{e:weylerr} is the probabilistic part of
\Lem{l:basis}, and the part that most resembles the standard Weyl trick.
We set
\[ \zeta(\vx) \defeq \exp(-2\pi i \vx \cdot \vy_2). \]
Since \equ{e:hdely} is a finite geometric series, we obtain
\[ \abs{\hdelta_Y(\vx)} = \left\lvert\frac{1-\zeta(\vx)^{\Lambda_1}}
    {1-\zeta(\vx)}\right\rvert \le \frac{2}{\abs{1-\zeta(\vx)}}. \]
If $\vx \notin H$, then $\zeta(\vx)$ is a uniformly random element of
the unit circle $S^1 \subseteq \C$.  (This remark about the behavior of
$\zeta(\vx)$ requires the hypothesis that $H = H_1$.  For general $H$,
we can only say this about $\zeta(\vx)$ for $\vx \notin H_1$.)  Therefore
\[ \Pr\left[\frac{1}{\abs{1-\zeta(\vx)}} \ge 8R_3^{2k} \right]
    < \frac{r_3^{2k}}{4}. \]
Applying the union bound over all $\vx \notin H$ with $\norm{\vx}_\infty <
R_3^2/2$, we get
\begin{eq}{e:halfprob}
\Pr\Bigg[\max_{\substack{\vx \in \Z^k \setminus H \\
    \norm{\vx}_\infty < R_3^2/2}} \back \frac{1}{\abs{1-\zeta(\vx)}}
    \le 8R_3^{2k} \Bigg] > \frac{3}{4}. \end{eq}
Assume from here that the conditional event in \equ{e:halfprob}
holds.  If so, then \equ{e:xlarge} tells us that
\begin{eq}{e:xsmall}
\sum_{\substack{\vx \in \Z^k \setminus H \\ \norm{\vx}_\infty < R_3^2/2}}
    \back \back \frac{\abs{\hdelta_Y(\vx)}\;\abs{\htau(\vx)}}{\Lambda_1}
    \le \frac{16R_3^{2k}}{\Lambda_1} + \exp(-\Omega(R_3^2)). \end{eq}
By combining equations \eqref{e:pairings}, \eqref{e:weylerr},
\eqref{e:xlarge}, and \eqref{e:xsmall}, we obtain the rigorous version of
\equ{e:wapprox} that we wanted, namely that
\begin{eq}{e:wbound} \left\lvert \frac{\braket{\hdelta_Y|\htau}}{\Lambda_1} -
\frac{\braket{\hdelta_{H^\#}|\htau}}{\sqrt{\Delta}} \right\rvert
\le \frac{16R_3^{2k}}{\Lambda_1} + \exp(-\Omega(R_3^2)) \end{eq}
with good probability.

Finally, we combine equations \eqref{e:wupper}, \eqref{e:wlower}, and
\eqref{e:wbound}, and we recall that
\[ R_3 = R_1^4 \qquad R_1 \ge \max(C,\Delta) > C_1. \]
We obtain:
\begin{align*}
\frac{r_3^{k-\ell}}{\sqrt{\Delta}} &\le \frac{16R_3^{2k}}{\Lambda_1}
    + \exp(-\Omega(R_3)) \\
\Lambda_1 &\le 16\sqrt{\Delta}R_3^{3k-\ell} + R^{k-\ell}\exp(-\Omega(R_3)) \\
    &\le 16R_1^{12k-4\ell+1/2} + \exp(-\Omega(R_1^4)).
\end{align*}
If $C_1$ is large enough, this contradicts the formula $\Lambda_1 =
R_1^{C_1k}$.  In turn, this confirms the hypothesis of \Cor{c:simplex} with
probability at least $3/4$, and thus concludes the proof of \Lem{l:basis}.

\subsection{LLL and rational approximation}
\label{ss:lllrat}

In this subsection, we will analyze steps 4--7 of \Alg{a:abelian}.  Then we
will combine various key estimates to prove \Thm{th:qrst}, which as stated
means that we will set and justify the parameters $Q$, $R$, $S$, and $T$
to make \Alg{a:abelian} work and run in $\BQP$.  The relevant estimates
are stated as Lemmas~\ref{l:feature}, \ref{l:unoise}, and \ref{l:basis}
in previous subsections, and additional lemmas and inequalities in this
subsection.  Along the way, we will also estimate the constants $R_1$ and
$\Lambda$ that are involved in \Lem{l:basis} and that connect $R$ and $T$.

Note that all of these parameters can only depend on the bit complexity $n
= \norm{M}_\bit$ of the output $M$ of the algorithm.  For concreteness, we
assume that the entries of $M$ have $d$ binary digits in total, and that with
some notation for the signs and separator symbols of the entries, $n > d$.

Step 4 of \Alg{a:abelian} applies the LLL algorithm \cite{LLL:factoring}
to the lattice $L$ described in that step.  $L$ is given as a lattice in
$\Q^{k+1} \subseteq \R^{k+1}$; more precisely,
\[ L \subseteq \frac{1}{\lcm(T,Q)}\Z^{k+1}. \]
If $T,Q = \exp(\poly(n))$, then LLL runs in $\ccP$.  We will
also need the following fundamental properties of LLL.

\begin{theorem}[Lenstra, Lenstra, and Lovasz
    {\cite[Prop.~1.12]{LLL:factoring}}]
There is an algorithm in $\ccP$ that takes as input a rank $k$ lattice $L
\subseteq \Q^k$ with basis matrix $M$, and produces a short basis
\[ B = [\vb_1,\vb_2,\ldots,\vb_k] \]
of $L$.  If $\vx_1,\vx_2,\ldots,\vx_j \in L$ are any $m$ linearly
independent vectors in $L$, then
\[ \norm{\vb_j}_2 \le 2^{k-1}
    \max(\norm{\vx_1}_2,\norm{\vx_2}_2,\ldots,\norm{\vx_j}_2). \]
\label{th:lll} \end{theorem} \eatline

The parameters $R$ and $T$ must both be chosen suitably for step 4 to
make effective use of LLL.  We first consider $R$.  Step 4 keeps all of
the vectors $\vb_j$ in the basis $B$ such that $\norm{\vb_j}_2 < r = 1/R$.
The point of this restriction is to keep $\vb_j$ when its denoised version
lies in $H_\R^\perp \oplus \R$.  Thus we can take $2r \le r_0$, where $r_0$
is the feature length of $H^\#$ discussed in \Sec{ss:dual} and illustrated
in \Fig{f:feature}, and the factor of 2 is a constant factor noise allowance.
In turn, \Lem{l:feature} says that
\[ r_0 \le \frac1{\Delta} = \frac1{\det(M^TM)}, \]
where $\Delta$ is the Gram determinant of $L$ defined in \eqref{e:gram}.
If $\vm_j$ is the $j$th column of $M$, then the Hadamard bound tells us that
\[ \Delta \le \norm{\vm_1}_2^2\;\norm{\vm_2}_2^2
    \cdots \norm{\vm_\ell}_2^2. \]
Suppose that the entries of $\vm_j$ have $d_j$ digits, and with a total of
$d$ digits across all columns.   Then $\norm{\vm_j}_2 \le 2^{d_j}$, so that
\begin{eq}{e:dbound} \Delta \le 2^{2d} < 2^{2n}, \end{eq}
and we can let
\begin{eq}{e:rnbound} R \ge 2^{2n+1}. \end{eq}

While we need to make $R$ large enough that step 4 only keeps vectors
$\vb_j \approxin H_\R^\perp \oplus \R$, we also need to make $T$ large
enough that step 4 retains an approximate basis of $H_\R^\perp \oplus \R$.
\Lem{l:basis} establishes that $L$ includes a basis of $H_\R^\perp \oplus \R$
whose vectors have norm at most $r_1 = 1/R_1$ in the $H_\R^\perp$ summands.
If we let
\begin{eq}{e:tbound} T \ge \Lambda R_1 \end{eq}
where $\Lambda$ is taken from \Lem{l:basis}, then $L$ includes a basis of
$H_\R^\perp \oplus \R$ whose vectors have norm at most $r_1$ in both
summands, and thus have norm at most $\sqrt{2}r_1$.  Since \Thm{th:lll}
requires an extra factor of $2^k$, it suffices to let
\begin{eq}{e:r1rbound} R_1 \ge 2^{k+1} R. \end{eq}

In step 4, we also need to bound the noise term $\vu$ in the sample
$\vy_1$ as discussed in \Lem{l:unoise}.  If we compare \Lem{l:unoise}
to \Lem{l:basis}, it suffices to satisfy the latter lemma by letting
\begin{eq}{e:sbound1} S \ge \Lambda\sqrt{k}. \end{eq}
This bound shows that LLL will $k+1-\ell$ vectors of norm at most $r$,
but we also want to know that any such vector is a vector in $H_\R^\perp
\oplus \R$ plus a noise term. In other words, we want to know that there
is little enough noise in $\vb_j$ with $\norm{\vb_j}_2 \le r$ that it must
be displaced from the stripe through the origin in \Fig{f:feature} and
not any parallel stripe.

The basis $E$ in \eqref{e:lbasis}, has only one vector with a non-zero
last component, namely $(\tvy_1,t)$.  Therefore each LLL vector $\vb_j$
with $\norm{\vb_j}_2 \le r$ uses $(\tvy_1,t)$ with a coefficient $\lambda_j$
with $\abs{\lambda_j} \le rT$.  Using the fact that $r$ is at most half of
the feature length $r_0$, and again using \Lem{l:unoise} the corresponding
noiseless vector $\lambda_j\vy_0 \in H^\#$ lies on the stripe through the
origin with good probability if
\begin{eq}{e:sbound2}
rTs\sqrt{k} < r \quad \iff \quad S > T\sqrt{k}.
\end{eq}

The first part of step 5 finds a nonsingular square submatrix $B_2$ of
$B_1$ to then form $B_3 = B_1 B_2^{-1}$.  The matrix $B_3$ is essentially
the reduced column echelon form of $B_1$ and can be found efficiently.
We will need an upper bound on the operator norm $\norm{B_2^{-1}}_{2 \to
2}$ in order to calculate a noise bound and then another bound on $S$
in the second part of step 5.

\begin{lemma} $B_1$ has a nonsingular $(k+1-\ell) \times (k+1-\ell)$
submatrix $B_2$ that contains the last row and that can be found in $\ccP$.
If $R \ge \sqrt{k+1}$ and $T > 1$, then
\[ \abs{\det(B_2)} \ge t \qquad \norm{B_2^{-1}}_{2 \to 2} \le T. \]
\label{l:b2inv} \end{lemma} \eatline

\begin{proof} Recall that the defining basis \eqref{e:lbasis} of $L$ consists
of the standard basis of $\R^k$ plus one other vector $(\tvy_1,t)$.  $L$
therefore has lattice volume $t$.  Since $r < 1$, the columns of $B_1$
must use the last coordinate.  Thus the vectors
\begin{eq}{e:span} \vb_1,\vb_2,\ldots,\vb_{k+1-\ell},
    \ve_1,\ve_2,\ldots,\ve_k \end{eq}
must span $\R^{k+1}$ (over $\R$).  We can make a basis of $\R^{k+1}$
by taking those vectors in \eqref{e:span} that are linearly independent
from the previous vectors; we let $B_4$ be the matrix of these vectors.
The vectors in $B_1$ itself (since they come from an LLL basis of $L$),
so $B_1$ is a submatrix of $B_4$.  After reordering the basis of $\R^k$,
the matrix $B_4$ has the form
\[ B_4 = \begin{bmatrix} B_5 & I_\ell \\ B_2 & 0 \end{bmatrix}. \]
Moreover, the columns of $B_4$ all lie in $L$ and generate either
$L$ or a sublattice, so
\[ \abs{\det(B_2)} = \abs{\det(B_4)} \ge t. \]
This proves that we can find $B_2$ in $\ccP$, and it establishes
the lower bound on its determinant.

For the upper bound on $\norm{B_2^{-1}}_{2 \to 2}$, the fact that
$\norm{\vb_j}_2 \le r$ for each $1 \le j \le k+1-\ell$ tells us that
$\norm{B_1}_{1 \to 2} \le r$ and therefore that
\begin{multline}
\norm{B_2}_{2 \to 2} \le \norm{B_1}_{2 \to 2} \\
    \le \sqrt{k+1-\ell}\norm{B_1}_{1 \to 2} \le r\sqrt{k+1-\ell} \le 1.
\label{e:b1b2bound} \end{multline}
Let
\[ 1 \ge \sigma_1 \ge \sigma_2 \ge \cdots \ge \sigma_{k+1-\ell} > 0 \]
be the singular values of $B_2$, and recall that
\[ \prod_j \sigma_j = \abs{\det(B_2)} \ge t. \]
Therefore the smallest singular value of $B_2$ satisfies
\[ \sigma_{k+1-\ell} \ge t. \]
Since the singular values of $B_2^{-1}$ are
\[ \sigma_{k+1-\ell}^{-1} \ge \sigma_{k-\ell}^{-1}
    \ge \cdots \ge \sigma_1^{-1} > 0, \]
we learn that
\[ \norm{B_2^{-1}}_{2 \to 2} = \sigma_{k+1-\ell}^{-1} \le T, \]
as desired.
\end{proof}

To complete step 5, we need noise bounds for each matrix in the equation $B_3
= B_1 B_2^{-1}$.  For this purpose, let $A_1$ be the noiseless version of the
matrix $B_1$, in which the component $\lambda_j\vy_1$ in each term $\vb_j$
is replaced with $\lambda_j\vy_0$.  Likewise let $A_2$ be the submatrix
of $A_1$ in the same position as $B_1$ lies in $B_2$, and let $A_3 =
A_1 A_2^{-1}$.  We also denote the differences as matrix noise terms:
\begin{align*}
B_1 &= A_1 + U_1 & B_2 &= A_2 + U_2 \\
B_3 &= A_3 + U_3 & B_2^{-1} &= A_2^{-1} + V_2.
\end{align*}

\begin{lemma} If
\[ R \ge k+1 \qquad T > 1 \qquad S \ge 2T^2 \]
and the hypotheses of \Lem{l:unoise} hold, then
\[ \norm{U_3}_{2 \to 2} < 4sT^3 \]
with probability at least $3/4$.
\label{l:u3} \end{lemma}

\begin{proof} By \Lem{l:unoise},
\[ \norm{U_1}_{1 \to 2} \le rT\sqrt{k}s \]
with good probability. Therefore
\begin{multline}
\norm{U_2}_{2 \to 2} \le \norm{U_1}_{2 \to 2} \le
    \sqrt{k+1-\ell}\norm{U_1}_{1 \to 2} \\ \le \sqrt{(k+1-\ell)k}rsT < sT,
\label{e:u1u2bound} \end{multline}
where the last step uses the hypothesis $R \ge k+1$.

In the remainder of the proof, we omit the subscript from the matrix norm:
\[ \norm{X} \defeq \norm{X}_{2 \to 2}. \]
Using the hypothesis $S \ge 2T^2$ and \Lem{l:b2inv}, we learn that
\begin{eq}{e:b2u2}
\norm{B_2^{-1}U_2} \le \norm{B_2^{-1}}\;\norm{U_2} < sT^2 \le \frac12.
\end{eq}
Using this, we next obtain a bound on $\norm{V_2}$ using a power
series method.

If we write
\[ A_2^{-1} = (B_2-U_2)^{-1} = \big(B_2(I-B_2^{-1}U_2)\big)^{-1}, \]
then by \eqref{e:b2u2},
\[ A_2^{-1} = (I + B_2^{-1}U_2 + (B_2^{-1}U_2)^2 + \cdots)B_2^{-1} \]
is a convergent matrix power series.  Moreover
\begin{align}
\norm{V_2} &\le \sum_{j=1}^\infty
        \norm{B_2^{-1}U_2}^j\;\norm{B_2^{-1}} \nonumber \\
    &\le \norm{B_2^{-1}U_2}\left(\sum_{j=0}^\infty
        \frac1{2^j}\right)\;\norm{B_2^{-1}} \nonumber \\
    &= 2\norm{B_2^{-1}U_2}\;\norm{B_2^{-1}} < 2sT^3. \label{e:v2bound}
\end{align}

Finally
\begin{align}
B_3 &= B_1B_2^{-1} = (A_1 + U_1)(A_2^{-1} + V_2) \nonumber \\
    &= A_1A_2^{-1} + U_1(A_2^{-1}+V_2) + (A_1+U_1)V_2 -U_1V_2 \nonumber \\
    &= A_1A_2^{-1} + U_1B_2^{-1} + B_1V_2-U_1V_2. \label{e:b3}
\end{align}
Combining \equ{e:b3} with the inequalities \eqref{e:b1b2bound},
\eqref{e:u1u2bound}, \eqref{e:v2bound}, $S \ge 2T^2$, and $T > 1$,
we obtain
\begin{align*}
\norm{U_3} &= \norm{B_3-A_1A_2^{-1}} \\
    &\le \norm{U_1}\;\norm{B_2^{-1}} + \norm{B_1}\;\norm{V_2}
        + \norm{U_1}\;\norm{V_2} \\
    &< sT^2 + 2sT^3 + 2s^2T^4 < 4sT^3,
\end{align*}
as desired.
\end{proof}

Following step 5, we reorder the $k$ coordinates of $\Z^k$ so that
$B_2$ occupies the last $k+1-\ell$ rows of $B_1$.  It follows that
\[ A_3 = \begin{bmatrix} A_4 \\ I_{k-\ell} \end{bmatrix} \oplus [1]. \qquad
B_3 = \begin{bmatrix} B_4 \\ I_{k-\ell} \end{bmatrix} \oplus [1]. \]
\Lem{l:u3} implies that each entry of $A_4$ differs from the
corresponding entry of $B_4$ by less than $4sT^3$.

The remainder of step 5 removes the noise $U_3$ from $B_3$ in each entry
separately using an estimate of Legendre:

\begin{theorem}[Legendre {\cite[Thm.~184]{HW:numbers}}] If $x \in \R$
and $a/b \in \Q$ satisfy
\[ \left\vert x-\frac{a}{b}\right\vert < \frac1{2b^2}, \]
then $a/b$ is a convergent in the continued fraction expansion of $x$.
\label{th:legendre} \end{theorem}

If we use the continued fraction algorithm and \Thm{th:legendre} to find
$A_4$ from $B_4$, then we need an upper bound on the denominators that
arise in $A_4$.  Borrowing from step 6, the columns of
\[ A_5 = \begin{bmatrix} I_\ell \\ -A_4^T \end{bmatrix} \]
are a basis of the kernel of $A_3^T$ in $\R^k \subseteq \R^{k+1}$.
Since the columns of $A_3$ are an $\R$-basis of $H_\R^\perp \oplus \R$
in the success case, the columns of $A_5$ are then a basis of $H_\R$.
The columns of the hidden matrix $M$ are another basis of $H_\R$.  Thus $A_5$
is the reduced column echelon form (RCEF) of $M$.  The matrix
\[ M_2 = M\;M_1^{-1} = \begin{bmatrix} I_\ell \\ M_3 \end{bmatrix}, \]
where $M_1$ is the initial $\ell \times \ell$ submatrix of $M$, is also
the RCEF of $M$.  Thus
\[ A_4 = -M_3^T. \]
By the adjugate formula for the inverse of a matrix, the denominators in
$M_1^{-1}$ and therefore in $M_2$ and $M_3$ all divide $\abs{\det(M_1)}$.
By the same Hadamard bound argument that provides the upper bound
\eqref{e:dbound}, and assuming \eqref{e:rnbound}, we learn that
\[ \abs{\det(M_1)} < 2^n < R. \]
It is thus valid to use $R$ as the denominator cutoff in the continued
fraction algorithm.   Moreover, if
\begin{eq}{e:rst} 4sT^3 \le r^2 \quad \iff \quad S \ge 4R^2T^3, \end{eq}
then \Lem{l:u3} yields the hypothesis of \Thm{th:legendre}.

To conclude the analysis of step 5 and the constants that govern
\Alg{a:abelian}, we draw together the hypotheses of Lemmas~\ref{l:feature},
\ref{l:unoise}, \ref{l:basis}, \ref{l:b2inv}, and \ref{l:u3}, together
with the inequalities \eqref{e:dbound}--\eqref{e:sbound2}, and \eqref{e:rst}:
\begin{gather*}
\Delta < 2^{2n} \qquad R \ge 2^{2n+1} \qquad
    R \ge \sqrt{k+1} \qquad R \ge k+1 \\
R_1 \ge \max(\Delta,2^k,C) \qquad R_1 \ge 2^{k+1} R
    \qquad \Lambda = R_1^{Ck} \\
T \ge \Lambda R_1 \qquad T > 1 \qquad S \ge \max(\Delta^2,C) \qquad S \ge
    \Lambda \sqrt{k} \\
S > T \sqrt{k} \qquad S \ge 2T^2 \qquad S \ge 4R^2T^3 \qquad Q \ge S^2.
\end{gather*}
(Recall that $C$ is a universal constant large enough to make everything
work.)  The inequalities can all be satisfied if we take, in order:
\begin{gather*}
R = 2^{\Theta(n)} \qquad R_1 = 2^{\Theta(k)}R = 2^{\Theta(nk)} \\
\Lambda = R_1^{\Theta(k)} = 2^{\Theta(nk^2)} \qquad
T = \Lambda R_1 = 2^{\Theta(nk^2)} \\
S = 4R^2T^3 = 2^{\Theta(nk^2)} \qquad Q = S^2 = 2^{\Theta(nk^2)}.
\end{gather*}
This establishes the first half of \Thm{th:qrst}, that
\[ Q,R,S,T = \exp(\poly(n)). \]

We have also confirmed the second half of \Thm{th:qrst} through step 5,
given that steps 1--3 run in $\BQP$, steps 4 and 5 run in $\ccP$, and
\Lem{l:unoise} guarantees success with good probability.  Step 6 also runs
in $\ccP$ using the Kannan--Bachem theorem that the Smith normal form (SNF)
of an integer matrix can be computed in $\ccP$.  Finally step 7 runs in
$\BQP$ because it is the Shor--Kitaev algorithm.

The only remaining task is that we need a certain generalization of Smith
normal form and the Kannan--Bachem result.  In the standard version, if $A
\in M(k \times \ell,\Z)$ is a $k \times \ell$ matrix with entries in $\Z$,
then its SNF
\[ X  = VAW \in M(k \times \ell,\Z) \]
is a non-negative diagonal matrix such that each diagonal entry divides
the next one, and such that $V \in \GL(k,\Z)$ and $W \in \GL(\ell,\Z)$.
The same definition works equally well when $A,X \in M(k \times \ell,\Q)$,
as long as $V$ and $W$ are still integer matrices.  This generalization of
Smith normal form is equivalent to the standard one, because we can rescale
$A$ to make it an integer matrix, and rescale $X$ by the same ratio.  Also,
just as the SNF of an integer matrix generalizes the GCD of two integers,
the SNF of a rational matrix generalizes the fact that any two rational
numbers have a GCD if we interpret $\Q$ as a $\Z$-module. For example,
\[ \gcd{}_\Z\bigl(\frac13,\frac34\bigr) = \frac1{12}. \]

\section{\Thm{th:shift}: AHShP and VAHSP}
\label{s:ahshp}

In this section, we outline the proof of \Thm{th:shift} by giving
the corresponding algorithm, which we call a collimation sieve.
In Sections~\ref{ss:phase}-\ref{ss:measure}, we will first discuss ideas
in the algorithm that are based on earlier hidden shift algorithms as well
as \Sec{s:abelian}.  The algorithm itself is in \Sec{ss:sieve}.

Throughout this section, let $k \ge 1$ be an integer, and let $H \subseteq
\Z^k$ be a visible subgroup of $\Z^k$ given by a generator matrix $M$.
Let $\vs \in \Z^k/H$ be a shift hidden by a pair of injective hiding
functions
\[ f_0,f_1:\Z^k/H \into X. \]
It will be convenient to combine these two hiding functions into one
$H$-periodic hiding function
\begin{eq}{e:combhide}
    f:\Z^k/H \times \{0,1\} \to X \qquad f(\vx,a) = f_a(\vx). \end{eq}
We assume that the hidden shift is represented by a vector $\vs \in \Z^k$
such that
\begin{eq}{e:vsbound} \norm{\vs}_\infty \le 2^{t-1}, \end{eq}
so that the $k$ components of $\vs$ together have bit complexity at most
$n = kt$.  We allow the computed answer to be any element of the coset
$\vs+H$, not necessarily one that satisfies \eqref{e:vsbound}, despite
our assumption that such a representative exists.  Note that finding a
low-norm representative $\vs$ is the close vector problem (CVP) for the
lattice $H$, which can be solved exactly in time $2^{O(k)}$ for the norm
$\norm{\vs}_2$ \cite{MV:voronoi} and within a factor of $1+\eps$ for any
$\eps > 1$ for the norm $\norm{\vs}_\infty$ \cite{BN:sampling}.

In \Sec{s:abelian}, $s = 1/S$ is used to denote the reciprocal width
of an approximate Gaussian state $\ket{\psi_{GC}}$.  Since we are now
using $s$ for a hidden shift, for the moment we let $G$ be the width of
$\ket{\psi_{GC}}$, and we let $g = 1/G$.  We also keep the QFT radix $Q$
and its reciprocal $q = 1/Q$, and we keep the relation $Q \ge G^2$.

\subsection{Phase vectors}
\label{ss:phase}

A \emph{phase vector} is a quantum state of the form
\begin{eq}{e:phase} \ket{\phi_Y} \propto \sum_{a=0}^{\ell-1}
    \exp(2\pi i \vy_a \cdot \vs)\ket{a}, \end{eq}
where the \emph{phase multipliers} $\vy_a \in H^\# \subseteq (\R/\Z)^k$
are stored classically as a list
\[ Y = (\vy_0,\vy_1,\ldots,\vy_{\ell-1}). \]
(In some cases, the multiplier list $Y$ is stored in a compressed form.)
Following Peikert \cite{Peikert:csidh} and earlier work by the author
\cite{K:dhsp2}, phase vectors are the main data structure in our algorithm.
In this subsection, we give some of their basic properties.

If we add a constant $\vz$ to each $\vy_a$, the phase vector $\ket{\phi_Y}$
changes by a global phase and thus represents the same quantum state
$\ketbra{\phi_Y}$.  We are thus free to translate $Y$ by any fixed $\vz \in
(\R/\Z)^k$ without materially altering $\ket{\phi_Y}$.

The author's first algorithm \cite{K:dhsp} used qubit phase vectors (or
\emph{phase qubits})
\begin{eq}{e:qphase} \ket{\phi_{\vy}} = \frac{\ket{0}
    + \exp(2\pi i \vy \cdot \vs)\ket{1}}{\sqrt2} \end{eq}
with a single nonzero multiplier $\vy$.  In the algorithm for \Thm{th:shift},
phase qubits arise at the beginning and end of the collimation sieve.
At the beginning, they are created from the hiding function (\Sec{ss:create})
with an approximately random value of $\vy$.  The sieve itself uses longer
phase vectors in its intermediate stages.  At the end, it produces a qubit
phase vector whose multiplier $\vy$ approximates a desired target value
(\Sec{ss:sieve}).

Let $\ket{\phi_Y}$ and $\ket{\phi_Z}$ be two phase vectors with multipliers
\[ Y = \vy_0,\vy_1,\ldots,\vy_{\ell-1}
    \qquad Z = \vz_0,\vz_1,\ldots,\vz_{j-1}. \]
Their direct sum $\ket{\phi_Y} \oplus \ket{\phi_Z}$ is a phase vector with
multiplier list $(Y,Z)$, meaning here the concatenation of $Y$ followed
by $Z$.  Their tensor product $\ket{\phi_Y} \tensor \ket{\phi_Z}$ is also a
phase vector, with multipliers $\vy_a+\vz_b$ for all $0 \le a < \ell$ and
$0 \le b < j$.  We can think of the multiplier list of $\ket{\phi_Y}
\tensor \ket{\phi_Z}$ as a convolution $Y*Z$.

If $\ket{\phi_Y}$ is a phase vector of length $\ell$, we can make a shorter
phase vector by measuring the value of any function $c:[0,\ell)_\Z \to
C$, where $C$ is an unstructured set.  We can restore the standard form
\eqref{e:phase} by enumerating the solutions to the equation $c(a) = x$,
where $x$ is the measured value.

We will combine these primitives to make Algorithms \ref{a:shsieve} and
\ref{a:shfinal} in \Sec{ss:sieve}.  One primitive that we still clearly
need is phase vector creation, which we describe in the next subsection.

\subsection{Creating phase vectors}
\label{ss:create}

We can follow the first three steps of \Alg{a:abelian} to approximate
a phase qubit $\ket{\phi_{\vy}}$ as in \eqref{e:qphase}, where
$\vy \in H^\#$ is approximately uniformly distributed.  We can produce an
approximate version of $\ket{\phi_{\vy}}$ by evaluating the hiding function
$f$ in \eqref{e:combhide} and then applying a QFT measurement.

As in \Sec{ss:ahsp}, let $U_f$ be the unitary dilation of $f$ restricted
to the region $\norm{\vx}_\infty < Q/2$, which is in turn interpreted as
a subset of the finite group $(\Z/Q)^k$.  $U_f$ has two input registers,
a $(\Z/Q)^k$ register and a qubit register, and an output register that can
hold an element of $X$ (or a finite subset if $X$ is infinite).  Since $U_f$
is a unitary dilation, as usual it retains the input along with the output.

To approximately create $\ket{\phi_{\vy}}$, we first make
\[ U_f(\ket{\psi_{GC}} \tensor \ket{+}), \]
where $\ket{\psi_{GC}}$ is the approximate Gaussian state in \Alg{a:abelian}
and
\[ \ket{+} \propto \ket{0} + \ket{1} \]
is the standard plus state. As in \Alg{a:abelian}, we discard the output
register, and we measure the input $(\Z/Q)^k$ register in the Fourier basis.
Let
\[ \vy_1 \in q(\Z/Q)^k \subseteq (\R/\Z)^k \]
be the rescaled result of the Fourier measurement, and let $\vy_2$ be the
point in $H^\#$ closest to $\vy_1$.  Both to find $\vy_2$ and then later
to bound statistical error, we can use a counterpart to \Lem{l:unoise}
with the probability $3/4$ replaced by a very high probability (and with
$s$ replaced by $g$ in this section).

Calculating $\vy_2$ from $\vy_1$ amounts to CVP as follows.  We can choose
a lift $\tvy_1 \in \R^k$ of $\vy_1$, and then equivalently find a closest
point $\tvy_2 \in H^\bullet$.  We can find the closest point $\vz \in
H^\perp$ with linear algebra, and then the closest point to $\tvy_1-\vz
\in H_\R$ in the lattice $H^\circ = H_\R \cap H^\bullet$ is $\tvy_2-\vz$.
If the Gaussian width $G$ is large enough and $Q \ge G^2$ (and given a
suitable counterpart to \Lem{l:unoise}), we can use the LLL algorithm.

The posterior state in the qubit register is in general a mixed state
$\rho_{\vy_1}$, because the output register of $U_f$ was discarded
rather than measured.  Again if $G$ is large enough and $Q \ge
G^2$, then $\rho_{\vy_1}$ is a good approximation to the phase vector
$\ket{\phi_{\vy_2}}$.  We can then replace one by the other when analyzing
the algorithms in \Sec{ss:sieve}.  The distribution of $\vy_2 \in H^\#$
is also close enough to uniform that we can assume a uniform distribution
in this analysis.

The difference between the actual measured state $\rho_{\vy_1}$ and an
ideal phase vector state $\ket{\phi_{\vy}}$ can be described in three steps.

First, we can replace the state $\ket{\psi_{GC}}$ with the state
$\ket{\psi_G}$, and replace the Fourier operator on $(\Z/Q)^k$ with the
Fourier operator
\[ F:\ell^2(\Z^k) \to L^2((\R/\Z)^k). \]
The measured Fourier mode is now $\vy_3 \in (\R/\Z)^k$, the posterior qubit
state is a slightly different mixed state $\sigma_{\vy_3}$, and we
can let $\vy_4 \in H^\#$ be the closest point to $\vy_3$.  A calculation
similar to the proof of \Lem{l:pffpgc} yields the trace distance relation
\[ d(\rho_{\vy_1},\sigma_{\vy_3}) = O(g^2). \]
Moreover, by a Hessian calculation similar to the proof of \Lem{l:yquant}
and \equ{e:u1u2noise}, we obtain:
\begin{align*}
d(\vy_2,\vy_4) &= O(\sqrt{k}g^2) \\
d(\ket{\phi_{\vy_2}},\ket{\phi_{\vy_4}}) &= O(\sqrt{k}g^2\norm{\vs}_2).
\end{align*}
Finally, $\vy_4$ is exactly uniformly random on $H^\#$.

Second, if we measure the value of $U_f$ instead of discarding it, then
before taking the Fourier transform the full input register has a state
\begin{multline*}
\ket{\psi_{\vv,\vs}} = \frac{1}{\sqrt{2\braket{\psi_G|\psi_G}}}
    \biggl(\sum_{\vx \in H+\vv} \exp(-\pi g^2 \norm{\vx}_2^2) \ket{\vx,0} \\
    + \sum_{\vx \in H+\vv+\vs} \exp(-\pi g^2 \norm{\vx}_2^2)
    \ket{\vx,1}\biggr).
\end{multline*}
This state is similar to \eqref{e:gcos}, except that here we normalize
each $\ket{\psi_{\vv,\vs}}$ using the total normalization of all terms in
$\ket{\psi_G}$.  These two sums generally differ in their total norm and
their Fourier spectra.  Hence, after the Fourier measurement of $\vy_3$
the coefficients of $\ket{0}$ and $\ket{1}$ in the posterior qubit state
typically have unequal amplitudes.  The error from this discrepancy is
$O(g\norm{\vs}_2)$, ultimately because the Gaussian function $\exp(-\pi
g^2 \norm{\vx}_2^2)$ has Lipschitz constant $O(g)$.

Third, when $\vy_3$ is replaced by the closest point $\vy_4 \in H^\#$,
then $\ket{\phi_{\vy_3}}$ changes to $\ket{\phi_{\vy_4}}$.  The difference
$\vu = \vy_4-\vy_3$ has the same behavior as the noise term $\vu_2$
in \eqref{e:u1u2noise}, and the distance between $\ket{\phi_{\vy_1}}$
and $\ket{\phi_{\vy}}$ is $O(\vu \cdot \vs)$.  On average this is also
$O(g\norm{\vs}_2)$.  If $G \ge \sqrt{k}$, then the average distance between
$\rho_{\vy_1}$ and $\ket{\phi_{\vy}}$ is $O(g\norm{\vs}_2)$.

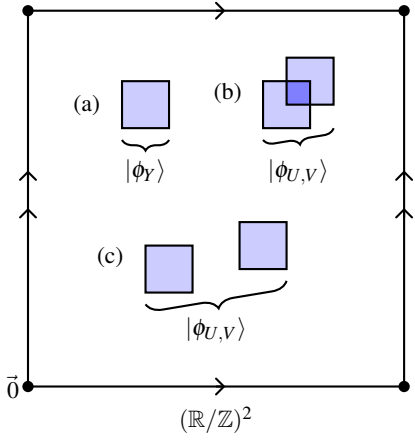
\begin{figure}
\begin{tikzpicture}[thick,scale=1.25]
\begin{scope}[shift={(1.25,1)}]
\draw[fill=lightblue] (0,0) rectangle ++(.5,.5);
\draw[fill=lightblue] (1,.25) rectangle ++(.5,.5);
\draw[decoration={brace,mirror,amplitude=1.5ex},decorate]
    (0,-.125) -- node[below=1.5ex] {$\ket{\phi_{U,V}}$} ++(1.5,.25);
\draw (-.375,.375) node {(c)};
\end{scope}
\begin{scope}[shift={(2.5,2.75)}]
\fill[lightblue] (0,0) rectangle ++(.5,.5);
\fill[lightblue] (.25,.25) rectangle ++(.5,.5);
\fill[halfblue] (.25,.25) rectangle ++(.25,.25);
\draw (0,0) rectangle ++(.5,.5);
\draw (.25,.25) rectangle ++(.5,.5);
\draw[decoration={brace,mirror,amplitude=1ex},decorate]
    (0,-.125) -- node[below=1.5ex] {$\ket{\phi_{U,V}}$} ++(.75,.125);
\draw (-.375,.375) node {(b)};
\end{scope}
\begin{scope}[shift={(1,2.75)}]
\draw[black,fill=lightblue] (0,0) rectangle ++(.5,.5);
\draw[decoration={brace,mirror,amplitude=1ex},decorate]
    (0,-.125) -- node[below=1ex] {$\ket{\phi_Y}$} ++(.5,0);
\draw (-.375,.25) node {(a)};
\end{scope}
\draw (0,0) rectangle (4,4);
\draw[-angle 90] (2.0,0) -- (2.1,0); \draw[-angle 90] (2.0,4) -- (2.1,4);
\draw[-angle 90] (0,1.8) -- (0,1.9); \draw[-angle 90] (0,2.2) -- (0,2.3);
\draw[-angle 90] (4,1.8) -- (4,1.9); \draw[-angle 90] (4,2.2) -- (4,2.3);
\fill (0,0) circle (.06) node[left] {$\vec0$}; \fill (4,0) circle (.06);
\fill (0,4) circle (.06); \fill (4,4) circle (.06);
\draw (2,0) node[below=1ex] {$(\R/\Z)^2$};
\end{tikzpicture}
\caption{Three possibilities for a collimated phase vector, with (a)
    one spot, (b) two overlapping spots, or (c) two disjoint spots.}
\label{f:spots}
\end{figure}

If $G \ge \sqrt{k}$, then the total of all three sources of error is
$O(g\norm{\vs}_2)$. Using the bound \eqref{e:vsbound},
\[ \norm{\vs}_2 \le \sqrt{k}\norm{\vs}_\infty \le \sqrt{k}2^t. \]
The algorithms in \Sec{ss:sieve} will need to create
$2^{O(\sqrt{n})}$ phase qubits (where $\norm{\vs}_\bit \le n$
and $\norm{M}_\bit \le h$), so it suffices to let
\begin{eq}{e:gqbound} Q \ge G^2 \qquad G =
    2^{O(\sqrt{n})}\norm{\vs}_2 = 2^{O(\sqrt{n}+t)} = 2^{O(n)}. \end{eq}
We can choose $G$ at this scale so that each sampled state $\rho_{\vy_1}$
is close enough to $\ket{\phi_{\vy}}$ that the noise makes no difference
for the algorithm.  The same estimates also imply that we can take $\vy_2
\in H^\#$ to be uniformly random, since this is exactly true of $\vy_4$.
Finally, \eqref{e:gqbound} is generous enough to use the LLL algorithm to
find $\vy_2$.

\subsection{Collimation}
\label{ss:collim}

Let
\[ R = [-r,r]^k + \vv \subseteq (\R/\Z)^k \]
be a cubical region centered at some vector $\vv \in (\R/\Z)^k$.  Then a
phase vector $\ket{\phi_Y}$ is \emph{$R$-collimated} means that every
multiplier $\vy_a \in Y$ lies in $R$.  We also say that $\ket{\phi_Y}$ is
\emph{single-spot collimated} (or a \emph{single-spot phase vector} and that
$R$ is a \emph{collimation window}.  As a variation, if $\ket{\phi_U}$ is
$R$-collimated and $\ket{\phi_V}$ is $S$-collimated for two cubical regions
\[ R = [-r,r]^k + \vv \subseteq (\R/\Z)^k \qquad
    S = [-r,r]^k + \vw \subseteq (\R/\Z)^k \]
then the direct sum
\begin{eq}{e:2spot} \ket{\phi_{U,V}} = \ket{\phi_U} \oplus \ket{\phi_V}
\end{eq}
is \emph{$(R,S)$-collimated} or \emph{double-spot collimated}. The cubes
$R$ and $S$ in a double-spot collimation might or might not intersect,
as in \Fig{f:spots}. Recall also that the state $\ket{\phi_Y}$ is only
meaningful up to a global phase, and likewise the multiplier list $Y$ is only
considered up to a constant $\vz \in (\R/\Z)^k$.  Thus if $\ket{\phi_Y}$
is $R$-collimated, we are free to translate the multiplier list $Y$ and
the collimation window $R$ in tandem by $\vz$.

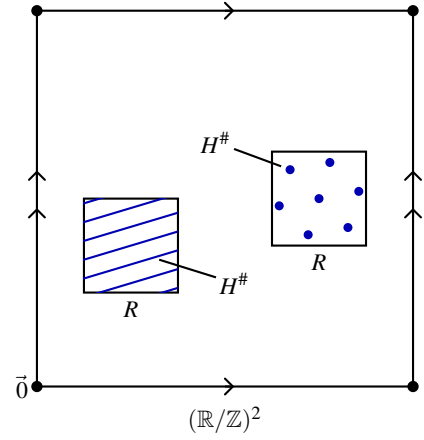
\begin{figure}
\begin{tikzpicture}[thick,scale=1.25]
\draw (0,0) rectangle (4,4);
\begin{scope}
\draw[clip] (.5,1) rectangle ++(1,1);
\foreach \y in {3,4,...,9} { \draw[darkblue] (0,{\y*.2}) -- ++(4,1.2); }
\end{scope}
\draw (2.1,1.1) node[inner sep=1pt] (Hsh1) {$H^\#$};
\draw (Hsh1) -- (1.3,1.35);
\draw (1,1) node[below] {$R$};
\draw (2.5,1.5) rectangle ++(1,1);
\foreach \n in {15,18,23,26,29,34,37} {
    \fill[darkblue] ({\n/26+2},{Mod(\n*19/13,4)}) circle (.048); }
\draw (1.9,2.6) node[inner sep=1pt] (Hsh2) {$H^\#$};
\draw ({18/26+2},{Mod(18*19/13,4)}) node[inner sep=3pt] (dot) {};
\draw (Hsh2) -- (dot);
\draw (3,1.5) node[below] {$R$};
\draw[-angle 90] (2.0,0) -- (2.1,0); \draw[-angle 90] (2.0,4) -- (2.1,4);
\draw[-angle 90] (0,1.8) -- (0,1.9); \draw[-angle 90] (0,2.2) -- (0,2.3);
\draw[-angle 90] (4,1.8) -- (4,1.9); \draw[-angle 90] (4,2.2) -- (4,2.3);
\fill (0,0) circle (.06) node[left] {$\vec0$}; \fill (4,0) circle (.06);
\fill (0,4) circle (.06); \fill (4,4) circle (.06);
\draw (2,0) node[below=1ex] {$(\R/\Z)^2$};
\end{tikzpicture}
\caption{An intersection of a collimation window $R$ with two possibilities
    for $H^\#$.}
\label{f:rhwindow}
\end{figure}

If $\ket{\phi_Y}$ is an unrestricted phase vector, then we will consider
it to be trivially single-spot collimated with $R = (\R/\Z)^k$; we will
also write $R = [-1/2,1/2]^k$.  We will likewise consider a direct sum
\eqref{e:2spot} of two unrestricted phase vectors $\ket{\phi_U}$ and
$\ket{\phi_V}$ to be trivially double-spot collimated.

\begin{remark} We use the term ``collimation'' by analogy with the concept
of collimated light in optics.  In earlier work, the author \cite{K:dhsp2}
defined collimation for phase multipliers $y_a \in \Z/(2^n)$ by requiring
that low bits vanish.  Peikert \cite{Peikert:csidh} introduced a variation
in which the multipliers $y_a \in \Z/N$ lie in an interval instead.
Our definition is a higher-dimensional generalization of Peikert's version
of collimation.  As Peikert noted, there are other useful collimation
windows.  In general, a sieve algorithm AHShP in an abelian group $A$
can make use of any subgroup or approximate subgroup of the Pontryagin
dual $\hA$ as a collimation window.
\end{remark}

Suppose that
\[ R = [-r,r]^k+\vv \subseteq (\R/\Z)^k \]
is a collimation window of a phase vector $\ket{\phi_Y}$ as defined above.
Since by definition every phase multiplier $\vy$ lies in $H^\# \subseteq
(\R/\Z)^k$, we can also say that the collimation window $\ket{\phi_Y}$
is the intersection
\[ R_H = ([-r,r]^k+\vv) \cap H^\#. \]
The region $R_H$ can be a complicated approximate subgroup of $H^\#$,
as indicated in \Fig{f:rhwindow}. It can also be a single point or the
empty set.

\begin{figure}
\begin{tikzpicture}[thick,scale=1.5]
\draw (0,0) rectangle (4,4);
\begin{scope}[shift={(.5,2.75)}]
\draw (-.25,.75) node {(a)};
\fill[lightblue] (0,0) rectangle ++(1,1);
\fill[darkred] (.25,.5) rectangle ++(.25,.25);
\foreach \t in {0,.25,...,1} {
    \draw (0,\t) -- ++(1,0); \draw (\t,0) -- ++(0,1); }
\draw (.375,.625) -- ++(1,0) node[right] {$\ket{\phi_Y}$};
\draw[decoration={brace,mirror,amplitude=1.5ex},decorate] (0,-.15)
    -- node[below=1.5ex] {$\ket{\phi_U} \tensor \ket{\phi_V}$} ++ (1,0);
\end{scope}
\begin{scope}[shift={(.5,.75)}]
\draw (-.25,.75) node {(b)};
\fill[lightblue] (0,0) rectangle ++(1,1) (2,.5) rectangle ++(1,1);
\fill[darkred] (.25,.5) rectangle ++(.25,.25) (2.25,1) rectangle ++(.25,.25);
\draw (.375,.625) -- ++(2,.5);
\draw (1.5,.675) node {$\ket{\phi_{Y,Z}}$};
\foreach \t in {0,.25,...,1} {
    \draw (0,\t) -- ++(1,0); \draw (\t,0) -- ++(0,1);
    \draw (2,{.5+\t}) -- ++(1,0); \draw ({2+\t},.5) -- ++(0,1); }
\draw[decoration={brace,mirror,amplitude=1.5ex},decorate] (0,-.2)
    -- node[below=1.5ex] {$\ket{\phi_{U,V}} \tensor \ket{\phi_W}$} ++(3,.5);
\end{scope}
\draw[-angle 90] (2.0,0) -- (2.1,0); \draw[-angle 90] (2.0,4) -- (2.1,4);
\draw[-angle 90] (0,1.8) -- (0,1.9); \draw[-angle 90] (0,2.2) -- (0,2.3);
\draw[-angle 90] (4,1.8) -- (4,1.9); \draw[-angle 90] (4,2.2) -- (4,2.3);
\fill (0,0) circle (.05) node[left] {$\vec0$}; \fill (4,0) circle (.05);
\fill (0,4) circle (.05); \fill (4,4) circle (.05);
\draw (2,0) node[below=1ex] {$(\R/\Z)^2$};
\end{tikzpicture}
\caption{Two examples of a collimation measurement of a tensor product of
    two phase vectors to produce a new phase vector with smaller support.
    In (a), a measurement of $\ket{\phi_U} \tensor \ket{\phi_V}$ produces
    $\ket{\phi_Y}$.  In (b), a measurement of $\ket{\phi_{U,V}} \tensor
    \ket{\phi_W}$ produces $\ket{\phi_{Y,Z}}$.}
\label{f:collim}
\end{figure}
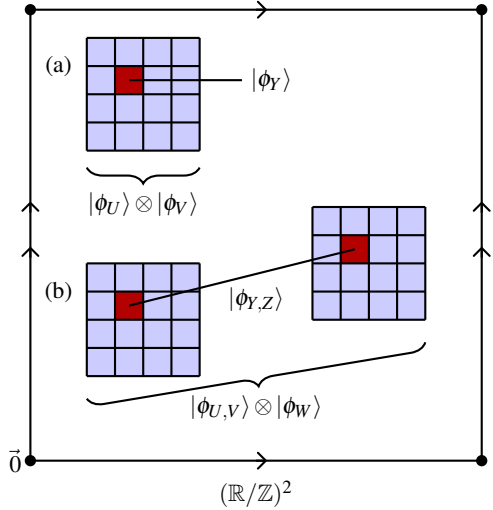

Finally in this subsection, we describe the sieve step that converts
two collimated phase vectors to a more sharply collimated phase vector.
Suppose first that $\ket{\phi_U}$ and $\ket{\phi_V}$ are both single-spot
collimated phase vectors of comparable but possibly unequal lengths,
and with the same collimation window $R = [-r,r]^k$.  Then their tensor
product $\ket{\phi_U} \tensor \ket{\phi_V}$ has collimation window $R+R =
[-2r,2r]^k$.  We can tile $R+R$ by $2^{km+k}$ subcubes that are each
translates of $[-2^{-m}r,2^{-m}r]^k$.  If $C$ is the set of centers of
these tiles, we obtain a function
\[ \vc:[0,\ell)_\Z {\times} [0,j)_\Z \to C, \]
where
\[ T = [-r2^{-m},r2^{-m}]^k + \vc(a,b)\]
is the tile that contains the phase multiplier $\vy_a+\vz_b$ of $\ket{\phi_U}
\tensor \ket{\phi_V}$.  If we use $\vc$ to measure $\ket{\phi_U} \tensor
\ket{\phi_V}$, then the posterior state $\ket{\phi_Y}$ is $T$-collimated.
\Fig{f:collim}(a) shows an example.

Second, suppose that $\ket{\phi_{U,V}}$ is double-spot $(R,S)$-collimated
with
\[ R = [-r,r]^k \qquad S = [-r,r]^k + \vy, \]
and that $\ket{\phi_W}$ is single-spot $R$-collimated.  Then we can
likewise tile the cubes $R+R$ and $R+S$ by subcubes that are translates
of $[-r2^{-m},r2^{-m}]^k$, and likewise measure $\ket{\phi_{U,V}} \tensor
\ket{\phi_W}$ to obtain a more sharply collimated double-spot phase vector
$\ket{\phi_{Y,Z}}$.  In this case, we measure which pair of tiles contain
a multiplier $\vy_a + \vz_b$, without measuring which of the two spots it
lies in.  Since the two spots are kept separate in the direct sum structure
\eqref{e:2spot}, we can define such a measurement even if the cubes $R+R$
and $R+S$ overlap, and even if the paired tiles overlap.  \Fig{f:collim}(b)
shows an example in the simpler case with disjoint cubes.

\subsection{Final measurement}
\label{ss:measure}

In this subsection, we describe how to learn the hidden shift $\vs$ by
measuring a phase vector $\ket{\psi_A}$, where $A \subseteq H^\#$ is a
suitably chosen finite subgroup of $H^\#$.  The technique is inspired by
a similar method due to Peikert \cite{Peikert:csidh}.

We recall the phase vector formula \eqref{e:phase}, except with the extra
property that the phase multiplier list is a finite group:
\[ \ket{\phi_A} \propto \sum_{\va \in A}
    \exp(2\pi i \va \cdot \vs)\ket{\va}. \]
Since $A \subseteq H^\# \subseteq (\R/\Z)^k$, it follows that $H \subseteq
A^\# \subseteq \Z^k$, and the quotient $\Z^k/A^\# \cong \hA$ is the
Pontryagin dual of $A$ (\Sec{ss:dual}).  $A$ will be small enough that we
can calculate a description of it in polynomial time from the matrix $M$
of $H$.  Given this, we can measure $\ket{\phi_A}$ in the Fourier basis on
$A$ in $\BQP$ \cite{Kitaev:stab,HH:fourier}.   The value of this measurement
is then $\vs$ as an element of $\hA \cong \Z^k/A^\#$, which is a quotient
of $\Z^k/H$.  We want $A$ to be a good enough approximation to $H^\#$
that we can learn $\vs \in \Z^k/H$ from its residue in $\Z^k/A^\#$.

The isomorphisms \eqref{e:noncanon} imply that $H^\#$ is (non-canonically)
isomorphic to $(H^\#/H_1^\#) \times H_1^\#$, in other words that $H_1^\#$
is a complemented subgroup of $H^\#$.  Let $A_1$ be any such complement,
and let $A_2 \subseteq H_1^\#$ be the set of all $\vy \in H_1^\#$ such
that $2^t\vy = \vec0$, where $t$ comes from the bound \eqref{e:vsbound}.
Finally let $A = A_1 + A_2$.  We can efficiently compute cyclic
generators of $A$ using the Smith normal form of the generator matrix $M$
of $H$.

A vector $\vx \in \Z^k$ lies in $A_2^\#$ if and only if $\vx = \vx_1
+ \vx_2$, where $\vx_1 \in H_1$ and $2^t\vert\vx_2$.  If $\vx \in A^\#
\subseteq A_2^\#$, it therefore also has this form.  Since $\norm{\vs}_\infty
\le 2^{t-1}$ by hypothesis, the shift $\vs \in \Z^k/H$ is then uniquely
determined by its projection to $\Z^k/A^\#$.  Learning $\vs \in \Z^k/H$
amounts to the close vector problem for the quotient group $A^\#/H_1$
as a lattice in the vector space $\R^k/H_\R$.  If we assume a moderately
larger value $e = t+O(k)$, we can compute $\vs$ using the LLL algorithm.

We have constructed a finite group $A \subseteq H^\#$ such that the phase
vector $\ket{\phi_A}$ yields the hidden shift $\vs$.  The other task is
to produce the state $\ket{\phi_A}$.  In \Sec{ss:sieve}, we will describe
a sieve algorithm to produce a close approximation to the phase qubit
$\ket{\phi_{\vy}}$ for any specified $\vy \in H^\#$.  Here we describe
how to combine these phase qubits to make $\ket{\phi_A}$.

We can first decompose $A$ as an internal direct sum of cyclic factors:
\[ A = D_1 \oplus D_2 \oplus \cdots \oplus D_\ell
    \qquad D_j \cong \Z/d_j. \]
This decomposition can be computed from the Smith normal form of $M$,
and it implies a tensor decomposition of $\ket{\phi_A}$
into simpler phase vectors:
\[ \ket{\phi_A} = \ket{\phi_{D_1}} \tensor \ket{\phi_{D_2}} \tensor \cdots
    \tensor \ket{\phi_{D_\ell}}. \]
To create each tensor factor, let $D = D_j$ be generated by $\vy \in H^\#$,
so that
\[ D = \{0,\vy,2\vy,\ldots,(d-1)\vy\} \cong \Z/d. \]
We can create $\ket{\phi_D}$ by first making the phase vector
\[ \ket{\phi_Y} =  (\vec0,\vy,2\vy,\ldots,(2^e-1)\vy), \]
where $e = \ceil{\log_2(d)}$ is the first power of 2 that is
at least $d$.  We can create $\ket{\phi_Y}$ as the tensor product 
\[ \ket{\phi_Y} = \ket{\phi_{\vy}} \otimes \ket{\phi_{2\vy}}
    \otimes \cdots \otimes \ket{\phi_{2^{e-1}\vy}}. \]
If $d = 2^e$, then $\ket{\phi_D} = \ket{\phi_Y}$ and we are done.  Otherwise,
we can obtain $\ket{\phi_D}$ from $\ket{\phi_Y}$ as the posterior state
of a boolean measurement that succeeds with probability $d/2^e > 1/2$.
If the measurement fails, then we can discard the posterior and make
another attempt to create $\ket{\phi_D}$.  (We can also recycle the
disappointing posterior state if $2^e \ge d+2$, but the noise estimates
for the performance of \Alg{a:shfinal} would be more complicated.)

We conclude this subsection with an estimate of how much we need to
collimate an approximation $\ket{\phi_{(\vv,\vw)}}$ to each qubit
$\ket{\phi_{\vy}}$ to yield an adequate approximation to the Fourier
state $\ket{\phi_A}$.  We first estimate how many qubits we will need.
We use at most $2\log_2(\abs{A_1})$ qubits for the $A_1$ factor, and
$(k-\ell)t \le kt$ qubits for the $A_2$ factor. To estimate $\abs{A_1}$,
we combine the estimates \eqref{e:h1bound} and \eqref{e:dbound} for the
Gram determinant \eqref{e:gram}, to obtain
\[ |A_1| = [H^\#:H_1^\#] = [H_1:H] < 2^h \qquad h = \norm{M}_\bit. \]
Meanwhile $A_2$ is the group of $2^t$-torsion points in $H_1^\#$, which is
a $(k-\ell)$-torus, so this factor requires $(k-\ell)t \le kt = n$ qubits.
In total, we need at most $n+2h$ collimated qubits to make $\ket{\phi_A}$.

As in \eqref{e:qphase}, we have 
\begin{align*}
\ket{\phi_{(\vv,\vw)}} &= \frac{\exp(2\pi i \vv \cdot \vs)\ket{0}
    + \exp(2\pi i \vw \cdot \vs)\ket{1}}{\sqrt2} \\
\ket{\phi_{\vy}} &= \frac{\ket{0}
    + \exp(2\pi i \vy \cdot \vs)\ket{1}}{\sqrt2}.
\end{align*}
Moreover, $\ket{\phi_{(\vv,\vw)}}$ will be $(R,S)$-collimated with
\begin{align*}
\vv \in R &= [-r,r]^k \subseteq (\R/\Z)^k \\
    \vw \in S &= [-r,r]^k+\vy \subseteq (\R/\Z)^k.
\end{align*}
From this, we obtain the trace distance estimate
\begin{multline*}
\frac{\tr\bigl(\bigl\vert\;\ketbra{\psi_{\vv,\vw}}-
    \ketbra{\psi_{\vy}}\;\bigr\vert\bigr)}2 \le
    \abs{(\vw-\vv)\cdot\vs - \vy\cdot\vw} \\ \le
    \norm{\vw-\vv-\vy}_\infty\norm{\vs}_\infty \le kr2^{t-1}.
\end{multline*}
Since we need at most $n+2h$ of these qubits to make $\ket{\phi_A}$,
we need
\begin{eq}{e:rktbound} r^{-1} \ge (n+2h)k2^{t+1} \end{eq}
to get within trace distance $1/2$ of $\ket{\phi_A}$ and thus learn $\vs$
with probability at least $1/2$.

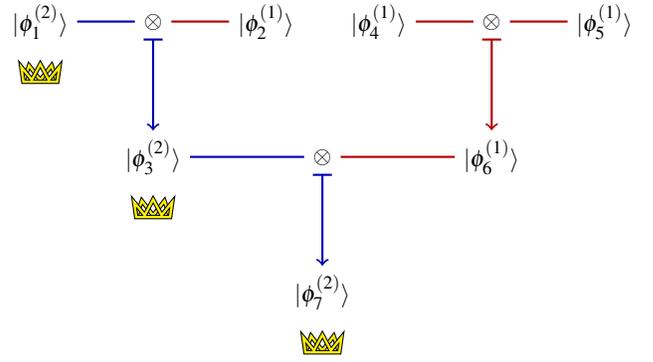
\begin{figure}
\begin{tikzpicture}[thick,xscale=1.25,yscale=.8]
\draw (-3,4.5) node (p1) {$\ket{\phi_1^{(2)}}$};
\draw (-.6,4.5) node (p2) {$\ket{\phi_2^{(1)}}$};
\draw (-1.8,4.5) node (t12) {$\otimes$};
\draw (-1.8,2.25) node (p3) {$\ket{\phi_3^{(2)}}$};
\draw (.6,4.5) node (p4) {$\ket{\phi_4^{(1)}}$};
\draw (3,4.5) node (p5) {$\ket{\phi_5^{(1)}}$};
\draw (1.8,4.5) node (t45) {$\otimes$};
\draw (1.8,2.25) node (p6) {$\ket{\phi_6^{(1)}}$};
\draw (0,2.25) node (t36) {$\otimes$};
\draw (0,0) node (p7) {$\ket{\phi_7^{(2)}}$};
\draw[darkblue] (p1) -- (t12) (p3) -- (t36);
\draw[darkred] (t12) -- (p2) (p4) -- (t45) -- (p5) (t36) -- (p6);
\draw[darkblue,|->] (t12) -- (p3);
\draw[darkred,|->] (t45) -- (p6);
\draw[darkblue,|->] (t36) -- (p7);
\foreach \x/\y in {-3/3.5,-1.8/1.25,0/-1} {
    \begin{scope}[shift={(\x,\y)},scale=.09,thin]
    \draw[double=yellow,double distance=.9pt] (-2,0) -- (-1.33,3) --
        (0,.3) -- (1.33,3) -- (2,0);
    \draw[double=yellow,double distance=.9pt] (-2,0) -- (-2.25,2.75) --
        (-1,.35) -- (0,3.25) -- (1,.35) -- (2.25,2.75) -- (2,0) -- cycle; 
    \end{scope} }
\end{tikzpicture}
\caption{A depth-first collimation sieve.  The phase vectors are created
    in the order of their subscripts; they superscript indicate whether
    they are single spot or double spot.  The double-spot phase vectors
    propagate as a royal lineage.}
\label{f:royal}
\end{figure}

\subsection{The sieve}
\label{ss:sieve}

In this subsection, we combine the ideas of Sections~\ref{ss:create},
\ref{ss:collim}, and \ref{ss:measure} into an algorithm for AHShP, which we
have divided into three algorithms.  Following the author's second algorithm
\cite{K:dhsp2}, the heart of the current algorithm is \Alg{a:shsieve}, a
multi-stage sieve based on the collimation step in \Sec{ss:collim}.  As in
\Sec{ss:create}, the sieve begins with \Alg{a:shhidden}, which creates
approximate phase qubits $\ket{\phi_{\vy}}$ with nearly uniformly random
values of $\vy \in H^\#$.  After tensoring a moderate number of these two
make longer phase vectors, the sieve eventually produces an approximate
qubit $\ket{\phi_{\vy}}$ with a prescribed (rather than random) $\vy
\in H^\#$.  As in \Sec{ss:measure}, after tensoring together a moderate
number of these qubits, \Alg{a:shfinal} determines the hidden shift $\vs$
with a Fourier measurement.

As in previous algorithms for DHSP or AHShP, if $\norm{\vs}_\bit = n$, the
sieve progressively collimates $\vy$ in $\Theta(\sqrt{n})$ stages to reach
a desired total precision of $n$ bits.  In the current $k$-dimensional
context, this means that each component of $\vy$ is collimated to $n/k$
bits of precision.  The phase vectors in the middle of the sieve have
length $2^{\Theta(\sqrt{n})}$, and the total work of the algorithm is also
$2^{\Theta(\sqrt{n})}$.  Moreover, as Regev realized \cite{Regev:dhsp},
a collimation sieve can be depth first rather than breadth first.

The current algorithm is a collimation sieve that can look roughly the
same for any periodicity $H \subseteq \Z^k$. As stated in \Thm{th:shift},
the time complexity is then nearly uniform in the choice of $H$.  The main
new technique is sieve steps that combine single-spot and double-spot phase
vectors to produce a double-spot phase vector at the end.  As explained in
\Sec{ss:collim}, a collimation step in the sieve combines a double-spot phase
vector and a single-spot phase vector to produce a double-spot phase vector,
or it combines two single-spot phase vectors to produce a single-spot phase
vector. As indicated in \Fig{f:royal}, the resulting sieve is asymmetrical.
The sieve is organized into a tree with a single path of double-spot phase
vectors; the other nodes are all single-spot phase vectors.

Here is the algorithm in three parts.

\begin{algorithm}{SH}[Phase qubit from hiding function] The input is a
dimension $k > 0$, a generator matrix $M$ for a lattice $H \subseteq \Z^k$,
a joint hiding function $f$ as in \eqref{e:combhide}, and positive integer
parameters $Q = 1/q$ and $G = 1/g$.  The output is an approximate qubit
phase vector $\ket{\phi_{\vy}}$, where $\vy \in H^\# \subseteq (\R/\Z)^k$
is approximately uniformly random.
\begin{enumerate}
\item Make the state $\ket{\psi_{GC}}$ as in step 1 of \Alg{a:abelian},
with $G = 1/g$ replacing $S = 1/s$ there. Tensor it with $\ket{+}$.
\item Let $U_f$ be the unitary dilation of the hiding function $f$
restricted to the discrete cube $\norm{\vx}_\infty < Q/2$.  Then evaluate:
\[ U_f(\ket{\psi_{GC}} \tensor \ket{+}) \propto \back\back
    \sum_{\substack{\vx \in \Z^k,\ b \in \{0,1\} \\ \norm{\vx}_\infty < Q/2}}
    \back\back \exp(-\pi g^2 \norm{\vx}_2^2) \ket{\vx,b,f_b(\vx)}. \]
Discard the output register to obtain a mixed state $\rho$ on the $(\vx,b)$
input register.
\item Using the embedding of the discrete interval $\abs{x} < Q/2$ into
$\Z/Q$, apply the Fourier operator $F_{(\Z/Q)^k}$ to the $\vx$ register
of the state $\rho$.  Then measure the Fourier mode $\vy_0 \in (\Z/Q)^k$,
and let
\[ \vy = q\vy_0 \in q(\Z/Q)^k \subseteq (\R/\Z)^k. \]
\item Return the posterior state, which approximates $\ket{\phi_{\vy}}$.
\end{enumerate}
\label{a:shhidden} \end{algorithm}

\begin{algorithm}{SC}[Collimation sieve] The input is three integers $m >
0$, $j \ge 0$, $p \in \{1,2\}$, and a target vector $\vy \in H^\#$,
in addition to the input to \Alg{a:shhidden}.  If $p=1$, the output is
an $R$-collimated phase vector $\ket{\phi_Y}$ of length $4^{km}
\le \ell < 4^{km+1}$, where
\[ R = [-r,r]^k \qquad r = 2^{-jm-1}. \]
If $p=2$ and $j < km$, the output is an $(R,S)$-collimated phase vector
$\ket{\phi_{Y,Z}}$, where also
\[ S = [-r,r]^k + \vy, \]
and where the length of each summand is $4^{km} \le \ell < 4^{km+1}$.
If $p=2$ and $j = km$, the output is a phase qubit that approximates
$\ket{\phi_{\vy}}$.
\begin{enumerate}
\item If $j=0$ and $p=1$, form $\ket{\phi_Y}$ as the tensor product of $2km$
phase qubits produced by \Alg{a:shhidden}, and return the result.
\item If $j=0$ and $p=2$, form 
\[ \ket{\phi_{Y,Z}} = \ket{\phi_Y} \oplus \ket{\phi_Z} \]
as the tensor product of $2km+1$ phase qubits produced by \Alg{a:shhidden},
where the $2 \cdot 4^{km}$ phase terms are split equally and randomly between
the two summands.  Return the result.
\item If $j > 0$ and $p=1$, do the following:
\begin{enumerate}
\item Recursively call \Alg{a:shsieve} twice with $(j',p') = (j-1,1)$
(and the same values of $m$ and $\vy$), to obtain two single-spot phase
vectors $\ket{\phi_U}$ and $\ket{\phi_V}$, both collimated at the scale $r'
= 2^{-(j-1)m-1}$.
\item Perform a collimation measurement as in \Fig{f:collim}(a) to obtain
a single-spot phase vector $\ket{\phi_Y}$ collimated at the scale $r =
2^{-jm-1}$.
\item If $|Y| < 4^{km}$, reject $\ket{\phi_Y}$ and go back to step 3(a).
\item If $|Y| \ge 4^{km+1}$ (as is usually
the case), partition $\ket{\phi_Y}$ into summands of nearly equal lengths
$\ell$ with $4^{km} \le \ell < 4^{km+1}$, and measure $\ket{\phi_Y}$
according to this partition to shorten it, and reassign $Y$
to call the shortened version $\ket{\phi_Y}$.
\item Return $\ket{\phi_Y}$.
\end{enumerate}
\item If $j > 0$ and $p=2$, do the following:
\begin{enumerate}
\item Recursively call \Alg{a:shsieve} with $(j',p') = (j-1,2)$ to obtain
a double-spot phase vector $\ket{\phi_{U,V}}$, and a second time with
$(j',p') = (j-1,1)$ to obtain a single-spot phase vector $\ket{\phi_W}$,
both collimated at the scale $r' = 2^{-(j-1)m-1}$.
\item Perform a collimation measurement as in \Fig{f:collim}(b) to obtain
a double-spot phase vector $\ket{\phi_{Y,Z}}$, collimated at the scale $r
= 2^{-jm-1}$.
\item If either $|Y| < 4^{km}$ or $|Z| < 4^{km}$, reject $\ket{\phi_{Y,Z}}$
and go back to step 4(a).
\item If either $|Y| \ge 4^{km+1}$ or $|Z| \ge 4^{km+1}$, apply the
shortening procedure in Step 3(d) to each summand $\ket{\phi_Y}$ and
$\ket{\phi_Z}$ separately, and reassign $Y$ and $Z$ to call the shortened
version $\ket{\phi_{Y,Z}}$.
\item If $j = km$, partition $\ket{\phi_{Y,Z}}$ into pairs of basis
states, one each in $\ket{\phi_Y}$ and $\ket{\phi_Z}$, plus unpaired
states in either $\ket{\phi_Y}$ or $\ket{\phi_Z}$, whichever is longer.
Measure $\ket{\phi_{Y,Z}}$ according to this partition.  If the posterior
state is a qubit, return it; it approximates $\ket{\phi_{\vy}}$.  If the
posterior state is unpaired, go back to Step 4(a).
\item If $j < km$, return $\ket{\phi_{Y,Z}}$.
\end{enumerate}
\end{enumerate}
\label{a:shsieve} \end{algorithm}

\begin{algorithm}{SF}[Final measurement] The input is an integer $t > 0$
such that the hidden shift $\vs$ satisfies $\norm{\vs}_\infty \le 2^{t-1}$,
in addition to the input to \Alg{a:shhidden}.  The output is a value $\vs
\in \Z^k/H$ which is the hidden shift with probability at least $1/2$.
\begin{enumerate}
\item Let $m \ge 2$ be the first integer such that
\begin{eq}{e:mbound} 2^{km^2} > k(n+2h)2^t. \end{eq}
\item As in \Sec{ss:collim}, construct the group
\[ A = A_1 + A_2 \subseteq H^\# \]
where $A_1$ is a complement of $H_1^\#$ in $H^\#$ and $A_2$ is the group
of $2^t$-torsion of points in $H_1^\#$.
\item Decompose $A$ into a product of cyclic factors.
For each cyclic factor $D \cong \Z/d$ of
$A$, do the following, as in \Sec{ss:collim}:
\begin{enumerate}
\item Pick a cyclic generator $\vy \in D \subseteq (\R/\Z)^n$.  Let $e =
\ceil{\log_2(d)}$ be the first power of 2 that is at least $d$.
\item For each $\ell$ with $0 \le \ell < e$, call \Alg{a:shsieve}
with $\vy' = 2^\ell\vy$, the given value of $m$, $j = km$,
and $p=2$, to produce an approximate version of $\ket{\phi_{\vy'}}$.
\item Perform a boolean measure the tensor product of the states
\[ \ket{\phi_Y} = \ket{\phi_{\vy}} \otimes \ket{\phi_{2\vy}}
    \otimes \cdots \otimes \ket{\phi_{2^{e-1}\vy}} \]
from step 3(b) in the computational basis with the posterior state
$\ket{\phi_D}$ if the measurement succeeds.  If the measurement fails,
go back to step 3(b).
\end{enumerate}
\item Once $\ket{\phi_D}$ is created for every cyclic factor $D = D_j$
of $A$, form
\[ \ket{\phi_A} = \ket{\phi_{D_1}} \tensor \ket{\phi_{D_2}} \tensor \cdots
    \tensor \ket{\phi_{D_\ell}}. \]
Perform a Fourier measurement on $\ket{\phi_A}$ to obtain $\vs \in \Z^k/H$,
and return this answer.
\end{enumerate}
\label{a:shfinal} \end{algorithm}

Here are some points that partly explain why the overall algorithm works.

Note that \Alg{a:shhidden} produces an approximation $\rho$ (in general
a mixed state if the hiding function value is left unmeasured) to
$\ket{\phi_{\vy}}$ with $\vy$ random.  Then \Alg{a:shsieve} produces
$\ket{\phi_{\vy'}}$ which only approximates $\ket{\phi_{\vy}}$, because
$\vy'$ approximates $\vy$.  Each algorithm produces approximations which
are then used in the next algorithm as if they were exact.  This is valid if
the total trace distance between every output $\rho$ and each corresponding
$\ketbra{\phi_y}$ is low enough to leave a good chance that the algorithms
will not notice the discrepancies.

If we know that \Alg{a:shhidden} is only called $2^{O(\sqrt{n})}$ times,
then \equ{e:gqbound} states that $G$ and $Q$ can be both $2^{O(n)}$
and large enough for accurate output.  Since \Alg{a:shhidden} runs in
quantum polynomial time in $\norm{G}_\bit$ and $\norm{Q}_\bit$, this
contributes an acceptable factor to the total running time, one bounded by
$2^{O(\sqrt{n})}$ if the query cost of the hiding function is also bounded
by $2^{O(\sqrt{n})}$.

\Alg{a:shsieve} produces a sufficiently accurate approximation to each target
$\ket{\phi_{\vy}}$ when \eqref{e:rktbound} holds, which is then satisfied
by \eqref{e:mbound}.  Meanwhile, the running time of \Alg{a:shsieve}
is $2^{O(km)}$, provided that the chance of rejection in steps 3(c) and
4(c) is bounded away from 1.  To bound the running time in terms of $n$,
we can combine \eqref{e:mbound} with the relation $n=kt$ and the theorem
hypothesis $n = \Omega((k+\log(h))^2)$ to obtain:
\[ k^2m^2 < k\bigl[(t+2)+\log(k)+\log(n+2h)\bigr] = O(n). \]
Thus $km = O(\sqrt{n})$ and the running time of \Alg{a:shsieve} is
$2^{O(\sqrt{n})}$.  Meanwhile, \Alg{a:shfinal} runs in quantum polynomial
time, outside of its function calls to the other algorithms.

The remaining question in the algorithm analysis is whether \Alg{a:shsieve}
usually passes steps 3(c) and 4(c).  It does on average, because the
phase vectors in the sieve are generously long in relation to the amount
of collimation.  \Ie, given that $|U|,|V|,|W| \ge 4^{km}$, the typical
values of $|Y|$ and $|Z|$ are at least
\[ \frac{4^{2km}}{2^{k(m+1)}} = 2^{3km-k} \ge 4^{k(m+1)}. \]
However, as indicated in \Fig{f:rhwindow} combined with \Fig{f:collim},
the phase multipliers in $\ket{\phi_U} \tensor \ket{\phi_V}$ (for example)
cannot be uniformly distributed in the partition of the cube $[-2r',2r']^k$,
because they are restricted to $H^\#$.  Even if $H^\# = (\R/\Z)^k$, we
cannot expect the phase multipliers in $\ket{\phi_U} \tensor \ket{\phi_V}$
to be uniformly distributed.  If the multipliers in $\ket{\phi_U}$
and $\ket{\phi_V}$ are both uniformly distributed in $[-r',r']^k$,
the convolved multipliers $U*V$ in $\ket{\phi_U} \tensor \ket{\phi_V}$
are not uniformly distributed in $[-2r',2r']^k$, because the convolution
of the uniform distribution on an interval with itself is not uniform;
the density is a tent function.

However, step 3(c) in \Alg{a:shsieve} is robust regardless of the
distribution on $\ket{\phi_U} \tensor \ket{\phi_V}$, while step 4(c) is
robust under a heuristic assumption.  If we think of each tile $[-r,r]^k+\vv$
as a ``household'' in which the phase multipliers ``reside'', then we can
apply the ``average household'' principle:  Given $N$ people distributed
arbitrarily among $M$ households, then a randomly chosen household has at
least $N/2M$ residents with probability at least $1/2$.

To show that double-spot collimation in step 4(c) usually succeeds,
we must show that the two summands $\ket{\phi_Y}$ and $\ket{\phi_Z}$ of
$\ket{\phi_{Y,Z}}$ are often long at the same time, while the average
household principle only shows that $\ket{\phi_{Y,Z}}$ is often long.
To address this, note that the intersection $H^\# \cap ([-2r',2r']^k+\vy)$
is a translate of $H^\# \cap [-2r',2r']^k$, because $\vy H^\#$ and $H^\#$
is a subgroup of $(\R/\Z)^k$.  If we assume that the multipliers of long
phase vectors in \Alg{a:shsieve} resemble independent samples in the
infinite-length limit, then the non-uniformity from the convolution in the
multiplier list $U*V$ of $\ket{\phi_U} \tensor \ket{\phi_W}$ likewise has
the same form in $\ket{\phi_V} \tensor \ket{\phi_W}$.  In other words,
both of these sources of non-uniformity in the measurement in step 4(c)
and \Fig{f:collim}(b) are the same in tandem in the two spots.  Given the
statistical assumption, the summands $\ket{\phi_Y}$ and $\ket{\phi_Z}$
will usually be close in length.

Finally, we point out that \Alg{a:shsieve} (which is the core of the
algorithms taken together) is not fully optimized.  We have presented it
in a simplified form to establish a
\[ 2^{O(\sqrt{n})} = O(C^{\sqrt{n}}) \]
time complexity bound.  It should be possible to optimize the constant $C$,
and after that look for smaller accelerations, with methods similar to those
explored by the author \cite{K:dhsp2} and Peikert \cite{Peikert:csidh}.
Among other things, the phase vectors in \Alg{a:shsieve} are longer than
necessary.  It is also probably not optimal to use the same value of $m$
at every stage in the sieve.  It might also be better to collimate the
phase vectors in one or just a few directions at each stage, rather than
all $k$ directions in parallel.

That said, full optimization of \Alg{a:shsieve} (or any style of collimation
sieve) depends on the specific query cost of the hiding function, and can
depend on specific properties of the periodicity lattice $H \subseteq
\Z^k$.  Moreover, if there is an algorithm that is much faster than
$2^{O(\sqrt{n})}$, then finding it would be a different type of breakthrough,
because of known reductions from lattice problems (with no hiding function)
to the abelian hidden shift problem \cite{Regev:quantum}.

\bibliography{books,fa,gr,gt,mg,nt,qp,me}

\newcommand{\etalchar}[1]{$^{#1}$}
\providecommand{\bysame}{\leavevmode\hbox to3em{\hrulefill}\thinspace}
\providecommand{\MR}{\relax\ifhmode\unskip\space\fi MR }
\providecommand{\MRhref}[2]{%
  \href{http://www.ams.org/mathscinet-getitem?mr=#1}{#2}
}
\providecommand{\href}[2]{#2}
\providecommand{\eprint}{\begingroup \urlstyle{rm}\Url}
\begin{thebibliography}{ECH{\etalchar{+}}92}

\bibitem[AB09]{AB:modern}
Sanjeev Arora and Boaz Barak, \emph{Computational complexity: a modern
  approach}, Cambridge University Press, 2009.

\bibitem[AD16]{AD:unique}
Divesh Aggarwal and Chandan Dubey, \emph{Improved hardness results for unique
  shortest vector problem}, Inform. Process. Lett. \textbf{116} (2016), no.~10,
  631--637.

\bibitem[AD19]{AD:small}
Goulnara Arzhantseva and Cornelia Dru\c{t}u, \emph{Geometry of infinitely
  presented small cancellation groups and quasi-homomorphisms}, Canad. J. Math.
  \textbf{71} (2019), no.~5, 997--1018, \eprint{arXiv:1212.5280}.

\bibitem[AK07]{K:qcap}
Scott Aaronson and Greg Kuperberg, \emph{Quantum versus classical proofs and
  advice}, Theory Comput. \textbf{3} (2007), 129--157,
  \eprint{arXiv:quant-ph/0604056}.

\bibitem[AKS04]{AKS:primes}
Manindra Agrawal, Neeraj Kayal, and Nitin Saxena, \emph{{PRIMES} is in {P}},
  Ann. of Math. (2) \textbf{160} (2004), no.~2, 781--793.

\bibitem[Alp93]{Alperin:psl}
Roger~C. Alperin, \emph{Notes: {$\operatorname{PSL}_2(\mathbb{Z}) =
  \mathbb{Z}_2 \ast \mathbb{Z}_3$}}, Amer. Math. Monthly \textbf{100} (1993),
  no.~4, 385--386.

\bibitem[BBBV97]{BBBV:strengths}
Charles~H. Bennett, Ethan Bernstein, Gilles Brassard, and Umesh Vazirani,
  \emph{Strengths and weaknesses of quantum computing}, SIAM J. Comput.
  \textbf{26} (1997), no.~5, 1510--1523, \eprint{arXiv:quant-ph/9701001}.

\bibitem[BN09]{BN:sampling}
Johannes {Bl\"omer} and Stefanie Naewe, \emph{Sampling methods for shortest
  vectors, closest vectors and successive minima}, Theoret. Comput. Sci.
  \textbf{410} (2009), no.~18, 1648--1665.

\bibitem[Boo59]{Boone:word}
William~W. Boone, \emph{The word problem}, Ann. of Math. (2) \textbf{70}
  (1959), 207--265.

\bibitem[Bou90]{Bourbaki:algebra2}
Nicolas Bourbaki, \emph{Algebra. {II}. {Chapters} 4--7}, Elements of
  Mathematics (Berlin), Springer-Verlag, 1990, Translated from the French by P.
  M. Cohn and J. Howie.

\bibitem[Bru61]{Bruhat:distrib}
Fran\c{c}ois Bruhat, \emph{Distributions sur un groupe localement compact et
  applications {\`a} {l'\'etude} des {repr\'esentations} des groupes
  $p$-adiques}, Bull. Soc. Math. France \textbf{89} (1961), 43--75.

\bibitem[BS84]{BS:matrix}
L{\'a}szl{\'o} Babai and Endre Szemer{\'e}di, \emph{On the complexity of matrix
  group problems {I}}, 25th {Annual Symposium on Foundations of Computer
  Science}, 1984, pp.~229--240.

\bibitem[CJS14]{CJS:isogenies}
Andrew~M. Childs, David Jao, and Vladimir Soukharev, \emph{Constructing
  elliptic curve isogenies in quantum subexponential time}, J. Math. Cryptol.
  \textbf{8} (2014), no.~1, 1--29, \eprint{arXiv:1012.4019}.

\bibitem[CvD10]{CD:algebraic}
Andrew~M. Childs and Wim van Dam, \emph{Quantum algorithms for algebraic
  problems}, Rev. Modern Phys. \textbf{82} (2010), no.~1, 1--52,
  \eprint{arXiv:0812.0380}.

\bibitem[ECH{\etalchar{+}}92]{ECHLSPT:word}
David B.~A. Epstein, James~W. Cannon, Derek~F. Holt, Silvio V.~F. Levy,
  Michael~S. Paterson, and William~P. Thurston, \emph{Word processing in
  groups}, Jones and Bartlett Publishers, 1992.

\bibitem[EHK04]{EHK:hsp}
Mark Ettinger, Peter H{\o}yer, and Emanuel Knill, \emph{The quantum query
  complexity of the hidden subgroup problem is polynomial}, Inform. Process.
  Lett. \textbf{91} (2004), no.~1, 43--48, \eprint{arXiv:quant-ph/0401083}.

\bibitem[GR02]{GR:prob}
Lev Grover and Terry Rudolph, \emph{Creating superpositions that correspond to
  efficiently integrable probability distributions}, 2002,
  \eprint{arXiv:quant-ph/0208112}.

\bibitem[Gre60]{Greendlinger:word}
Martin Greendlinger, \emph{Dehn's algorithm for the word problem}, Comm. Pure
  Appl. Math. \textbf{13} (1960), 67--83.

\bibitem[Gro96]{Grover:fast}
Lov Grover, \emph{A fast quantum mechanical algorithm for database search},
  {ACM} Symposium on Theory of Computing, 1996,
  \eprint{arXiv:quant-ph/9706033}, pp.~212--219.

\bibitem[GT87]{GT:topo}
Jonathan~L. Gross and Thomas~W. Tucker, \emph{Topological graph theory},
  Wiley-Interscience Series in Discrete Mathematics and Optimization, John
  Wiley \& Sons, 1987.

\bibitem[HH00]{HH:fourier}
Lisa Hales and Sean Hallgren, \emph{An improved quantum {Fourier} transform
  algorithm and applications}, 41st {Annual Symposium on Foundations of
  Computer Science}, 2000, pp.~515--525.

\bibitem[H{\"o}r90]{Hormander:partial1}
Lars H{\"o}rmander, \emph{The analysis of linear partial differential
  operators. {I}}, second ed., Grundlehren der Mathematischen Wissenschaften
  [Fundamental Principles of Mathematical Sciences], vol. 256, Springer-Verlag,
  1990.

\bibitem[HRTS00]{HRT:normal}
Sean Hallgren, Alexander Russell, and Amnon Ta-Shma, \emph{Normal subgroup
  reconstruction and quantum computation using group representations}, {ACM}
  Symposium on Theory of Computing, 2000, pp.~627--635.

\bibitem[HW08]{HW:numbers}
G.~H. Hardy and E.~M. Wright, \emph{An introduction to the theory of numbers},
  sixth ed., Oxford University Press, 2008, Revised by D. R. Heath-Brown and J.
  H. Silverman, With a foreword by Andrew Wiles.

\bibitem[Ing37]{Ingham:consec}
Albert~E. Ingham, \emph{On the difference between consecutive primes}, Q. J.
  Math. \textbf{8} (1937), no.~1, 255--266.

\bibitem[IP01]{IP:ksat}
Russell Impagliazzo and Ramamohan Paturi, \emph{On the complexity of
  {$k$}-{SAT}}, J. Comput. System Sci. \textbf{62} (2001), no.~2, 367--375.

\bibitem[KB79]{KB:smith}
Ravindran Kannan and Achim Bachem, \emph{Polynomial algorithms for computing
  the {Smith} and {Hermite} normal forms of an integer matrix}, SIAM J. Comput.
  \textbf{8} (1979), no.~4, 499--507.

\bibitem[Kit95]{Kitaev:stab}
Alexei Kitaev, \emph{Quantum measurements and the {Abelian} stabilizer
  problem}, 1995, \eprint{arXiv:quant-ph/9511026}.

\bibitem[KMS17]{KMS:complex}
Olga Kharlampovich, Alexei Myasnikov, and Mark Sapir, \emph{Algorithmically
  complex residually finite groups}, Bull. Math. Sci. \textbf{7} (2017), no.~2,
  309--352, \eprint{arXiv:1204.6506}.

\bibitem[Kup05]{K:dhsp}
Greg Kuperberg, \emph{A subexponential-time quantum algorithm for the dihedral
  hidden subgroup problem}, SIAM J. Comput. \textbf{35} (2005), no.~1,
  170--188, \eprint{arXiv:quant-ph/0302112}.

\bibitem[Kup13]{K:dhsp2}
\bysame, \emph{Another subexponential-time quantum algorithm for the dihedral
  hidden subgroup problem}, 8th Conference on the Theory of Quantum
  Computation, Communication and Cryptography \textbf{22} (2013), 20--34,
  \eprint{arXiv:1112.3333}.

\bibitem[KW08]{KW:wave}
Alexei Kitaev and William~A. Webb, \emph{Wavefunction preparation and
  resampling using a quantum computer}, 2008, \eprint{arXiv:0801.0342}.

\bibitem[LLL82]{LLL:factoring}
Arjen~K. Lenstra, Hendrik~W. {Lenstra Jr.}, and {L\'aszl\'o} {Lov\'asz},
  \emph{Factoring polynomials with rational coefficients}, Math. Ann.
  \textbf{261} (1982), no.~4, 515--534.

\bibitem[Loh14]{Lohrey:groups}
Markus Lohrey, \emph{The compressed word problem for groups}, {SpringerBriefs
  in Mathematics}, Springer, New York, 2014.

\bibitem[LS77]{LS:combin}
Roger~C. Lyndon and Paul~E. Schupp, \emph{Combinatorial group theory},
  Ergebnisse der Mathematik und ihrer Grenzgebiete, vol.~89, Springer-Verlag,
  Berlin-New York, 1977.

\bibitem[ME99]{ME:hidden}
Michele Mosca and Artur Ekert, \emph{The hidden subgroup problem and eigenvalue
  estimation on a quantum computer}, Quantum computing and quantum
  communications ({Palm Springs, CA}, 1998), Lecture Notes in Comput. Sci.,
  vol. 1509, Springer, Berlin, 1999, \eprint{arXiv:quant-ph/9903071},
  pp.~174--188.

\bibitem[MV13]{MV:voronoi}
Daniele Micciancio and Panagiotis Voulgaris, \emph{A deterministic single
  exponential time algorithm for most lattice problems based on {Voronoi} cell
  computations}, SIAM J. Comput. \textbf{42} (2013), no.~3, 1364--1391.

\bibitem[Nov55]{Novikov:word}
Pyotr~S. Novikov, \emph{On the algorithmic unsolvability of the word problem in
  group theory}, Trudy Mat. Inst. im. Steklov., vol.~44, Izdat. Akad. Nauk
  SSSR, Moscow, 1955.

\bibitem[Ols91]{Olshanskii:geometry}
Alexander~Yu. Olshanskii, \emph{Geometry of defining relations in groups},
  Mathematics and its Applications (Soviet Series), vol.~70, Kluwer Academic,
  1991, Translated from the 1989 Russian original by Yu. A. Bakhturin.

\bibitem[Osb75]{Osborne:schwartz}
M.~Scott Osborne, \emph{On the {Schwartz-Bruhat} space and the {Paley-Wiener}
  theorem for locally compact abelian groups}, J. Functional Analysis
  \textbf{19} (1975), 40--49.

\bibitem[Pei96]{Peifer:artin}
David Peifer, \emph{Artin groups of extra-large type are biautomatic}, J. Pure
  Appl. Algebra \textbf{110} (1996), no.~1, 15--56.

\bibitem[Pei20]{Peikert:csidh}
Chris Peikert, \emph{He gives c-sieves on the {CSIDH}}, Advances in Cryptology
  -- EUROCRYPT 2020, Lecture Notes in Computer Science, vol. 12106, Springer,
  2020, pp.~463--492.

\bibitem[Pyb93]{Pyber:enum}
L.~Pyber, \emph{Enumerating finite groups of given order}, Ann. of Math. (2)
  \textbf{137} (1993), no.~1, 203--220.

\bibitem[Reg04a]{Regev:quantum}
Oded Regev, \emph{Quantum computation and lattice problems}, SIAM J. Comput.
  \textbf{33} (2004), no.~3, 738--760, \eprint{arXiv:cs.DS/0304005}.

\bibitem[Reg04b]{Regev:dhsp}
\bysame, \emph{A subexponential time algorithm for the dihedral hidden subgroup
  problem with polynomial space}, 2004, \eprint{arXiv:quant-ph/0406151}.

\bibitem[Reg09]{Regev:lwe}
\bysame, \emph{On lattices, learning with errors, random linear codes, and
  cryptography}, J. ACM \textbf{56} (2009), no.~6, Art. 34, 40.

\bibitem[Rud90]{Rudin:fourier}
Walter Rudin, \emph{Fourier analysis on groups}, Wiley Classics Library, John
  Wiley \& Sons, 1990, Reprint of the 1962 original.

\bibitem[Sho97]{Shor:factor}
Peter~W. Shor, \emph{Polynomial-time algorithms for prime factorization and
  discrete logarithms on a quantum computer}, SIAM J. Comput. \textbf{26}
  (1997), no.~5, 1484--1509, \eprint{arXiv:quant-ph/9508027}.

\bibitem[Sim70]{Sims:study}
Charles~C. Sims, \emph{Computational methods in the study of permutation
  groups}, Computational Problems in Abstract Algebra (Proc. Conf., Oxford,
  1967), Pergamon, Oxford, 1970, pp.~169--183.

\bibitem[Sim97]{Simon:power}
Daniel~R. Simon, \emph{On the power of quantum computation}, SIAM J. Comput.
  \textbf{26} (1997), no.~5, 1474--1483.

\bibitem[Sob01]{Sobrinho:qos}
{Jo\~ao}~L. Sobrinho, \emph{Algebra and algorithms for {QoS} path computation
  and hop-by-hop routing in the {Internet}}, Proceedings IEEE INFOCOM 2001.
  Conference on Computer Communications. Twentieth Annual Joint Conference of
  the IEEE Computer and Communications Society, vol.~2, IEEE, 2001,
  pp.~727--735.

\bibitem[Str90]{Strebel:small}
Ralph Strebel, \emph{Appendix. {Small} cancellation groups}, Sur les groupes
  hyperboliques d'apr\`es {Mikhael Gromov} ({Bern}, 1988), Progr. Math.,
  vol.~83, Birkh{\"a}user Boston, 1990, pp.~227--273.

\bibitem[Tra22]{Tran:thesis}
Austin Tran, \emph{Improved quantum algorithms for solving abelian hidden
  subgroup problems for subgroups of arbitrary index}, Ph.D. thesis, University
  of California, Davis, 2022.

\bibitem[Zal99]{Zalka:optimal}
Christof Zalka, \emph{Grover's quantum searching algorithm is optimal}, Phys.
  Rev. A \textbf{60} (1999), no.~4, 2746--2751,
  \eprint{arXiv:quant-ph/9711070}.

\bibitem[{\'Z}ra10]{Zralek:pollard}
Bartosz {\'Z}ra{\l}ek, \emph{A deterministic version of {P}ollard's {$p-1$}
  algorithm}, Math. Comp. \textbf{79} (2010), no.~269, 513--533,
  \eprint{arXiv:0707.4102}.

\end{thebibliography}

\end{document}